\documentclass[11pt,leqno]{article}

\usepackage{amsthm,amsfonts,amssymb,amsmath}
\usepackage{fullpage}
\usepackage{graphicx}
\usepackage{mathrsfs}
\usepackage{xcolor}
\usepackage{caption,subcaption}
\usepackage[colorinlistoftodos]{todonotes}
\usepackage{overpic}

\numberwithin{equation}{section}

\newcommand{\R}{\mathbb R}

\newcommand{\uu}{\mathrm{u}}
\newcommand{\s}{\mathrm{s}}

\newcommand{\f}{\mathrm{f}}
\newcommand{\bb}{\mathrm{b}}

\newcommand{\re}{\mathrm{e}}
\newcommand{\de}{\mathrm{d}}

\newtheorem{theorem}{Theorem}[section]

\newtheorem{corollary}[theorem]{Corollary}
\newtheorem{lemma}[theorem]{Lemma}
\newtheorem{remark}[theorem]{Remark}
\theoremstyle{definition}

\newtheorem{hypothesis}[theorem]{Hypothesis}

\allowdisplaybreaks[3]

\title{A traveling wave bifurcation analysis of turbulent pipe flow}

\author{Maximilian Engel\thanks{
Department of Mathematics, Freie Universit{\"a}t Berlin, Arnimallee 6, 14195 Berlin, Germany; \texttt{maximilian.engel@fu-berlin.de}.}~, Christian Kuehn\thanks{Department of Mathematics, Technical University of Munich, Boltzmannstr.~3, 85748 Garching b.~M\"unchen, Germany; \texttt{ckuehn@ma.tum.de}}~, Bj\"orn de Rijk\thanks{Institut f\"ur Analysis, Dynamik und Modellierung, Universit\"at Stuttgart, Pfaffenwaldring 57, 70569 Stuttgart, Germany; \texttt{bjoern.derijk@mathematik.uni-stuttgart.de}}}

\date{\today}

\begin{document}

\maketitle

\begin{abstract}
Using various techniques from dynamical systems theory, we rigorously study an experimentally validated model by [Barkley et al., Nature, 526:550-553, 2015], which describes the rise of turbulent pipe flow via a PDE system of reduced complexity. The fast evolution of turbulence is governed by reaction-diffusion dynamics coupled to the centerline velocity, which evolves with advection of Burgers' type and a slow relaminarization term. Applying to this model a spatial dynamics ansatz and geometric singular perturbation theory, we prove the existence of a heteroclinic loop between a turbulent and a laminar steady state and establish a cascade of bifurcations of various traveling waves  mediating the transition to turbulence. The most complicated behaviour can be found in an intermediate Reynolds number regime, where the traveling waves exhibit arbitrarily long periodic-like dynamics indicating the onset of chaos. Our analysis provides a systematic mathematical approach to identifying the transition to spatio-temporal turbulent structures that may also be applicable to other models arising in fluid dynamics.
\end{abstract}

	

\section{Introduction}

Understanding turbulence in fluids has been among the most challenging scientific problems for many years. Even in spatial domains with a relatively simple geometry, such as pipe flow, one has still not fully understood the transition mechanism(s) to spatio-temporal turbulence. Already back in the late 19th century, experiments of Reynolds~\cite{Reynolds} indicated that above a critical velocity, turbulence seems to persist in pipe flow for quite a broad range of initial conditions~\cite{Reynolds1}. It is now common to express the critical velocity via the Reynolds number, defined for pipe flow by $\textnormal{Re} = U\rho/n$, where $U$ is the mean velocity, $\rho$ stands for pipe diameter, and $n$ is the kinematic viscosity. There is no precise single Reynolds number value for the transition from laminar flow to fully turbulent flow~\cite{EckhardtSchneiderHofWesterweel, Kerswell}. The transition occurs within an entire range of Reynolds numbers. Yet, even the boundaries of the parameter region are still somewhat unclear~\cite{Eckhardt}. One main obstacle to study the transition is that the parabolic laminar profile in pipe flow is linearly stable for all Reynolds numbers~\cite{MesguerTrefethen, SalwenCottonGrosch} hinting at a more global dynamical effect. Experiments~\cite{DarbyshireMullin,HofJuelMullin, Reynolds} and direct numerical simulations~\cite{Avilaetal,MoxeyBarkley, Shanetal} for the Navier-Stokes equation show that the transition to turbulence can be caused via finite-size perturbations from the laminar flow. A key component in the transition to turbulence are finite-time localized patterns, so-called turbulent puffs~\cite{WygnanskiChampagne,WygnanskiSokolovFriedman}. Puffs are turbulent patches existing within the laminar flow decaying after a finite~\cite{HofWesterweelSchneiderEckhardt} but very long~\cite{AvilaWillisHof,Hofetal,KuikPoelmaWesterweel} time. Puffs can not only decay but can also split so that balancing splitting and decay rates is one possible option to estimate a critical Reynolds number for the transition to turbulence~\cite{MukundHof} and to determine the interaction length of puffs~\cite{SamantaDeLozarHof}. It is highly desirable to understand the precise dynamical mechanisms~\cite{SreenivasanRamshankhar}, e.g., near the onset of the transition. The splitting and decay processes of puffs lead one to consider stochastic processes as possible models, e.g., exploiting the analogy to coupled map lattice dynamics~\cite{ChateManneville,Kaneko2} or chemical/ecological systems~\cite{ShihHsiehGoldenfeld}. Near the onset, the directed percolation universality class~\cite{Hinrichsen} matches many experiments and simulations remarkably well~\cite{Lemoultetal,SanoTamai, SiposGoldenfeld}. Although a simplified statistical description is extremely helpful~\cite{Pomeau1}, it does not illuminate the transition mechanisms and geometry in phase space~\cite{Barkleyetal,ShimizuManneville}. A direct approach would be to mathematically analyze the Navier-Stokes equations~\cite{DoeringGibbon}. However, even many elementary-looking questions about Navier-Stokes quickly run into technical problems~\cite{BardosTiti} recognized already at the beginning of the twentieth century~\cite{Prandtl}. One could even argue that the situation has recently further worsened as weak solutions to Navier-Stokes are not even unique~\cite{BuckmasterVicol}. 

A natural approach is to consider models of ``intermediate complexity'', which are more tractable than Navier-Stokes but still capture many essential spatio-temporal features, dynamical mechanisms and statistics of turbulence. In this work, we focus on one of these models recently proposed, and experimentally validated extensively for pipe flow, by Barkley et al.~\cite{Barkleyetal}
\begin{equation}
\label{SYS1}
\begin{array}{lcl}
\partial_t q &=& D\partial_x^2q+ (\zeta-u)\partial_xq + f(q,u;r),\\
\partial_t u &=& -u\partial_x u + \varepsilon g(q,u),
\end{array}
\end{equation}
where $t \geq 0$ represents time, $x \in \R$ is interpreted as the stream-wise coordinate, $u=u(x,t)$ represents the centerline velocity, $q=q(x,t)$ models the turbulence level, and the reaction terms are given by
\begin{equation}
\label{SYS11}
f(q,u;r)=q(r+u-2-(r+0.1)(q-1)^2)\quad \text{and}\quad g(q,u)=2-u+2q(1-u).
\end{equation}
Regarding the parameters, $D > 0$ controls the coupling of turbulent patches to the laminar flow via diffusion, $r > 0$ models the Reynolds number in a suitable rescaling, $\zeta > 0$ takes into account the slower time scale of turbulent advection in comparison to the centerline velocity, while the small parameter $\varepsilon > 0$ controls the time scale separation between fast excursions of $q$ relative to slow recovery of $u$ after relaminarization. Structurally, one observes that~\eqref{SYS1} is a mixed system combining a bistable reaction-diffusion system with advective nonlinear terms of Burgers' type. Both individual elements are quite classical intermediate complexity simplifications, e.g., in modeling approaches~\cite{BecKhanin,Kida, Pomeau} as well as in localized reduced amplitude equations~\cite{IoossMielke,Schneider6}.

In this work, we are going to rigorously establish the existence of a wide variety of different traveling waves for the pipe flow turbulence model~\eqref{SYS1}-\eqref{SYS11}. There is ample motivation to study traveling waves in more detail~\cite{KawaharaUhlmannVanVeen}. Navier-Stokes simulations and experiments strongly indicate that the transition to turbulence in pipe flow is intimately connected to traveling waves~\cite{BudanurHof1,BudanurHof,FaistEckhardt,PringleKerswell,TohItano}. In particular, one current conjecture is that the existence of a boundary crisis~\cite{PeixinhoMullin,RitterMellibovskyAvila}, where an attractor collides with its basin boundary, generates a chaotic saddle. This saddle (or edge) state~\cite{Mellibovskyetal,SchneiderEckhardtYorke} seems to contain several interesting traveling waves~\cite{DuguetWillisKerswell}.

Our analysis of traveling waves in the turbulence model \eqref{SYS1}-\eqref{SYS11} is based on several steps. (S1) We re-write the problem via a standard spatial dynamics ansatz~\cite{KuehnBook1,Sandstede1} obtaining a three-dimensional system of ordinary differential equations (ODEs). In this context, bounded orbits correspond to traveling waves, e.g., homoclinic orbits to impulses, heteroclinic orbits to fronts, and periodic orbits to wave trains. (S2) The ODEs are singularly perturbed in the small parameter $\varepsilon > 0$ having the structure of a fast-slow dynamical system~\cite{FEN2,Kaper,KuehnBook,Jones} with two fast and one slow variable. We exploit the associated geometric decomposition in phase space and very explicitly construct singular heteroclinic orbits, which correspond to laminar-to-turbulent fronts and backs. Employing geometric singular perturbation theory in combination with Melnikov's method, one obtains the persistence of these orbits for the full three-dimensional ODE system. (S3) Using the existence of these heteroclinic orbits, we distinguish two parameter regimes, one for large Reynolds number, represented by the parameter $r$, and one for intermediate $r$. For the large $r$ case, we establish that the heteroclinic orbits form a twisted heteroclinic loop, while for the intermediate $r$ scenario a double-twisted heteroclinic loop is proven to exist. (S4) Having the existence of these orbits corresponding to spatio-temporal puff structures, we employ results due to Deng~\cite{Deng91a} as well as Homburg and Sandstede~\cite{HomburgSandstede}, to study bifurcations under parameter variation, which yields additional heteroclinic and also homoclinic structures forming connections between the laminar and turbulent states. All in all, we establish:

\begin{theorem}[Informal statement]
Let $\zeta > \frac{2}{5}$.
\begin{itemize}
\item[(i)] \textbf{Large Reynolds number regime, single twist:} For sufficiently large $r > 0$ and sufficiently small $\varepsilon > 0$ there exists an open interval $I_{\varepsilon,r} \subset (0,\infty)$ of diffusion rates such that for $D \in I_{\varepsilon,r}$ the pipe flow model~\eqref{SYS1}-\eqref{SYS11} exhibits simple laminar-to-turbulent front and back solutions, which are traveling-wave solutions, whose profiles possess a single interface. The fronts and backs can be propagating upstream or downstream, depending on the value of $\zeta$.
\item[(ii)] \textbf{Intermediate Reynolds number regime, double twist:} There exists $\gamma > 0$ such that for $r \in (\frac{2}{3},\frac{2}{3} + \gamma)$ and sufficiently small $\varepsilon > 0$ there exists an open interval $I_{\varepsilon,r} \subset (0,\infty)$ of diffusion rates such that for $D \in I_{\varepsilon,r}$ the pipe flow model~\eqref{SYS1}-\eqref{SYS11} admits infinitely many laminar-to-turbulent $k$-front and $k$-back solutions for arbitrary $k\in \mathbb{N}_0$, which are traveling waves whose profiles exhibit $k$ well-separated patches of turbulence, before converging towards fixed (laminar or turbulent) states at $\pm \infty$. The $k$-fronts and $k$-backs can be propagating upstream or downstream, depending on the value of $\zeta$.
\end{itemize}
In both of the above parameter regimes, the pipe flow model~\eqref{SYS1}-\eqref{SYS11} admits impulse solutions, which are traveling waves whose profiles are either laminar with a localized patch of turbulence, or they exhibit a localized absence of turbulence.
\end{theorem}

More details can be found in~\S\ref{sec:mainresults}, where we rigorously state our main results. The technical challenges of the associated mathematical proofs lie in finding a suitable approximation of the Melnikov integrals and a detailed geometric analysis of the twist regimes. 

The most interesting part of our result is the case of intermediate Reynolds number (ii). This is the regime, where the model~\eqref{SYS1}-\eqref{SYS11} was cross-validated both experimentally and via full Navier-Stokes simulations in~~\cite{Barkleyetal}. Our main results establish in the intermediate Reynolds number regime that infinitely many spatio-temporal invariant structures with different arbitrarily long transient periodic-like dynamics exist in the model~\eqref{SYS1}-\eqref{SYS11}, cf.~Figure~\ref{fig:1}. The existence of infinitely many different periodic structures is a \emph{common hallmark} feature among various definitions of \emph{chaos}~\cite{GH,Wiggins1}. Therefore, we have identified, via a concrete phase space construction, dynamical solutions that can organize the transition to turbulence. Furthermore, our steps (S1)-(S4) provide a general strategy, how to very explicitly and fully mathematically rigorously unravel many important parts of the spatio-temporal features of turbulent dynamics. Hence, we have not only obtained important results for the turbulence model~\eqref{SYS1}-\eqref{SYS11} concerning pipe flow but established a systematic approach to identify spatio-temporal turbulent solutions rigorously that can be applicable for other models in fluid dynamics and potentially in the long-term even directly to the Navier-Stokes equations. Regarding our mathematical results (i) for large Reynolds numbers for~\eqref{SYS1}-\eqref{SYS11}, additional experimental and/or numerical cross-validation would be necessary to make them directly applicable; see also~\S\ref{sec:outlook}.   

The remainder of the paper is structured as follows. In~\S\ref{sec:mainresults} we rigorously state our main results after introducing the necessary terminology. The proofs of our main results can be found in~\S\ref{sec:slowfast} and~\S\ref{sec:proof_twists}. In particular, in~\S\ref{sec:slowfast} we establish a heteroclinic loop, whereas in~\S\ref{sec:proof_twists} we analyze the bifurcating traveling waves in the large and intermediate Reynolds number regimes. We conclude with a discussion and outlook in~\S\ref{sec:outlook}.

\section{Main results} \label{sec:mainresults}

In this paper we establish a wide variety of traveling waves for the pipe flow turbulence model~\eqref{SYS1}-\eqref{SYS11}. \emph{Traveling waves} are solutions to~\eqref{SYS1} of the form
\begin{align} \left(q(x,t),u(x,t)\right) = \left(q_*(x-st),u_*(x-st)\right),\label{TW} \end{align}
which propagate with a fixed speed $s \in \R$ without changing their profile. Inserting the ansatz~\eqref{TW} into~\eqref{SYS1} and introducing the co-moving variable $\xi = x-st$, we arrive at the so-called \emph{traveling-wave equation}. We adopt a spatial dynamics formulation, cf.~\cite{KuehnBook1,Sandstede1}, and write the traveling-wave equation as a three-dimensional dynamical system
\begin{align}
\begin{split}
q_\xi &= p,\\
p_\xi &= D^{-1}\left((u + \mu)p - f(q,u;r)\right),\\
u_\xi &= \frac{\varepsilon g(q,u)}{u-s},
\end{split} \label{SYS2}
\end{align}
where we have conveniently replaced the variable $\zeta$ in~\eqref{SYS1}, which accounts for the difference in advection between turbulence and the centerline velocity, by the new variable $\mu = -\zeta-s$, which then represents this difference in advection relative to the speed $s$ of the traveling wave.

Bounded orbits $(q_*(\xi),p_*(\xi),u_*(\xi))$ in~\eqref{SYS2} directly correspond to traveling-wave solutions to~\eqref{SYS1}. In particular, the parabolic laminar flow in~\eqref{SYS1}, exhibiting no turbulence and a constant centerline velocity $u = 2$, corresponds to the equilibrium $X_1 = (0,0,2)$ in~\eqref{SYS2}. We are interested in the Reynolds number regime $r > \frac{2}{3}$, where~\eqref{SYS2} admits a second, $r$-dependent equilibrium $X_2 = (q_{\bb,+}(r),0,u_{\bb}(r))$, which corresponds to a turbulent steady state in~\eqref{SYS1} with constant (but non-zero) turbulence level $q = q_{\bb,+}(r) > 0$ and centerline velocity $u = u_{\bb}(r) \in (\frac{6}{5},\frac{4}{3})$.

The stability of the steady states $(q,u) = (0,2)$ and $(q,u) = (q_{\bb,+}(r),u_{\bb}(r))$ in the pipe flow model~\eqref{SYS1} is readily established by computing the spectra of the linearizations of~\eqref{SYS1} about both states. Thus, for $r > \frac{2}{3}$ one finds that the pipe flow model~\eqref{SYS1} is \emph{bistable}.

Our results strongly rely on the existence of forward and a backward heteroclinic connections between the equilibria $X_1$ and $X_2$ in~\eqref{SYS2}. Such heteroclinic connections directly correspond to traveling laminar-to-turbulent fronts and backs in~\eqref{SYS1}. A \emph{traveling front} is a solution to~\eqref{SYS1} of the form~\eqref{TW} with
\begin{equation*}
\lim_{\xi \to -\infty} \left(q_*(\xi),q_*'(\xi),u_*(\xi)\right) = X_1, \quad \lim_{\xi \to \infty} \left(q_*(\xi),q_*'(\xi),u_*(\xi)\right) = X_2.
\end{equation*}
Similarly, a \emph{traveling back} is a solution to~\eqref{SYS1} of the form~\eqref{TW} satisfying
\begin{equation*}
\lim_{\xi \to -\infty} \left(q_*(\xi),q_*'(\xi),u_*(\xi)\right) = X_2, \quad \lim_{\xi \to \infty} \left(q_*(\xi),q_*'(\xi),u_*(\xi)\right) = X_1.
\end{equation*}
Thus, depending on the sign of the speed $s$ in~\eqref{TW} such front and backs travel up- or downstream, and are describing either invasion of the parabolic laminar flow by the turbulent state or a recovery process of the laminar flow from the turbulent state, see Figure~\ref{fig:1}.

We will identify parameter regimes such that the pipe flow model~\eqref{SYS1} admits both a traveling front and a traveling back propagating with the same speed $s$ (hence one must be invading and the other recovering the laminar state). Together, the associated forward and backward heteroclinic connection in~\eqref{SYS2} form a so-called \emph{heteroclinic loop or cycle}, see Figure~\ref{fig:2}. We will prove that the heteroclinic loop satisfies certain twisting conditions, so that the general bifurcation theory of Deng~\cite{Deng91a} applies, see also~\cite{HomburgSandstede}. More specifically, we establish that the heteroclinic loop is single twisted in the large Reynolds number regime, whereas for intermediate Reynolds numbers we show that it is double twisted. This leads us to infinitely many nearby heteroclinics and homoclinics in~\eqref{SYS2} connecting the equilibria $X_1$ and $X_2$, which directly correspond to traveling waves in the pipe flow model~\eqref{SYS1} exhibiting arbitrarily long transient periodic-like dynamics. Before stating our main results, we introduce the necessary terminology to specify the type of traveling waves generated by such heteroclinic and homoclinic connections.

\begin{remark}
{\upshape We note that the pipe flow model~\eqref{SYS1} was also studied (and experimentally validated) in the bistable regime $r > \frac{2}{3}$ in~\cite{Barkleyetal}. Here, the existence of traveling laminar-to-turbulent fronts and backs was also recognized (albeit formally), see Remark~\ref{rem:expand} for more details. However, the existence of a single or double twisted heteroclinic loop in~\eqref{SYS2} and the large variety of nearby homoclinic and heteroclinic structures, yielding traveling waves in~\eqref{SYS1} with arbitrarily long transient periodic-like dynamics, seems to not have been observed before. Global homoclinic/heteroclinic structures often act as organizing centers of chaotic dynamics~\cite{Bertozzi1,GH,Ottino,Wiggins1} but are incredibly difficult to verify in spatio-temporal dynamics. This makes it very remarkable that we can, mathematically rigorously and via explicit calculations, construct such an organizing center for~\eqref{SYS1}. 
}
\end{remark}

\begin{figure}
\centering
\includegraphics[scale=0.45]{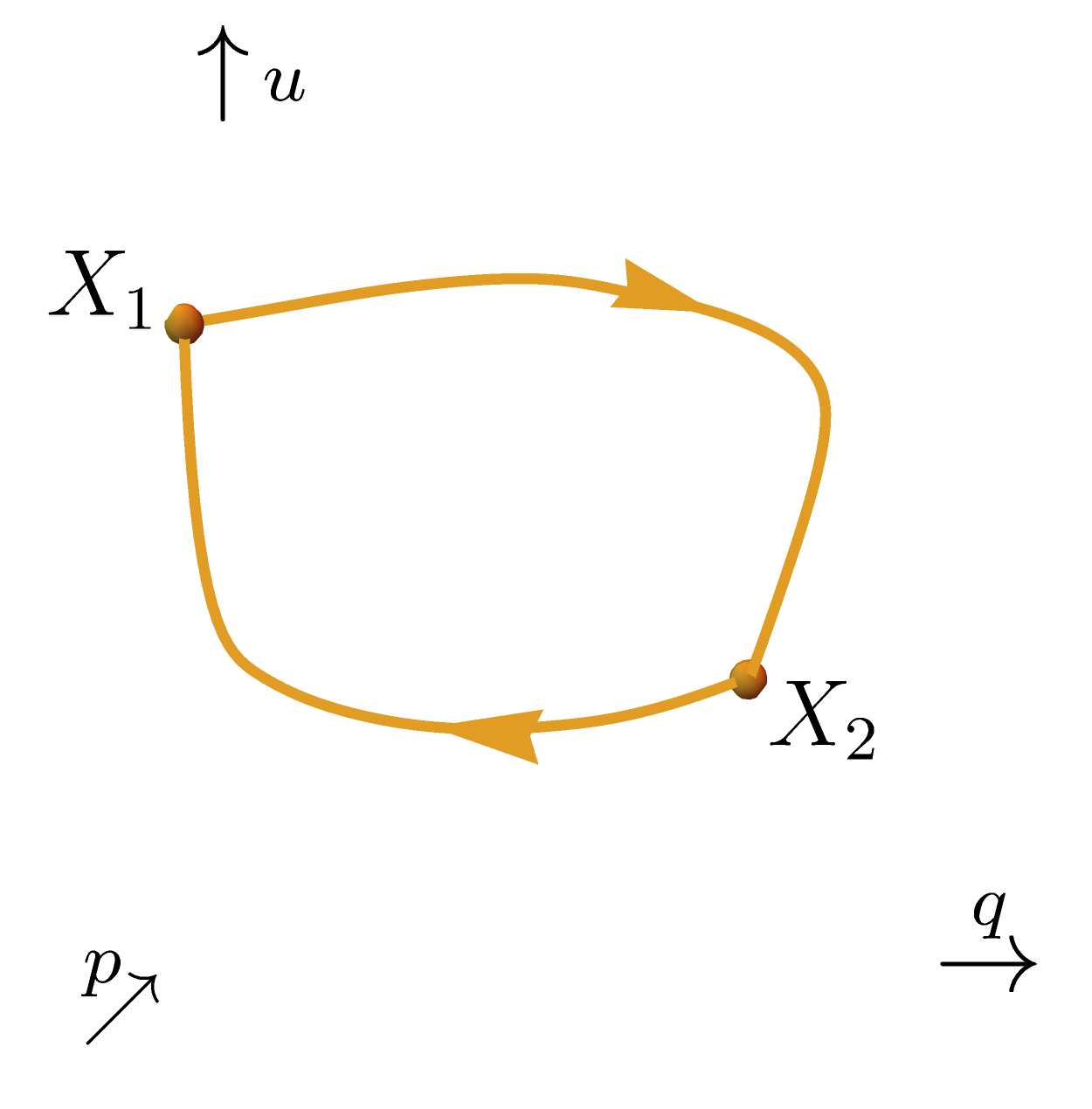}
\caption{\label{fig:2}
A heteroclinic loop or cycle connecting the equilibria $X_1$ and $X_2$ in the dynamical system~\eqref{SYS2}. The forward and backward connection between the equilibria $X_1$ and $X_2$ correspond to a simple front and back solution in the pipe flow model~\eqref{SYS1}, respectively, see Figure~\ref{fig:1}.}
\end{figure}

\subsection{Terminology}

First, we call a traveling wave~\eqref{TW} an \emph{impulse} of $X_1$ if
\begin{equation*}
\lim_{\xi \to \pm \infty} \left(q_*(\xi),q_*'(\xi),u_*(\xi)\right) = X_1,
\end{equation*}
which corresponds to an orbit homoclinic to the equilibrium $X_1$ in the dynamical system~\eqref{SYS2}. We will say that a traveling wave~\eqref{TW} has a \emph{pulse} or a \emph{puff} if there is a closed interval $\xi_0 \leq \xi \leq \xi_1$ such that, as $\xi$ increases, the associated solution $(q_*(\xi),q_*'(\xi),u_*(\xi))$ to~\eqref{SYS2} leaves a small neighborhood of $X_1$, enters a small neighborhood of $X_2$, and then goes back to the neighborhood of $X_1$ where it is found at $\xi = \xi_1$. That way, we can define \emph{$k$-fronts}, \emph{$k$-backs} and \emph{$k$-impulses} as fronts, backs and impulses with $k$ pulses, see Figure~\ref{fig:1}. A front (back) without a pulse is called a \emph{simple front (back)}, whereas an impulse with a single pulse is called a \emph{simple impulse}. For traveling fronts, backs or impulses the number of pulses, $k$, can be identified with the winding number of the corresponding heteroclinic or homoclinic orbits in~\eqref{SYS2} in a small tubular neighborhood of the heteroclinic loop, see Figure~\ref{fig:2}, with one full circumvolution corresponding to a single pulse. Thus, these orbits are called \emph{$k$-heteroclinic} or \emph{$k$-homoclinic} (and again simple for heteroclinics with $k=0$ and homoclinics with $k=1$).

\begin{figure}
\centering
\begin{subfigure}[b]{0.3\textwidth}
\centering
\includegraphics[scale=0.405]{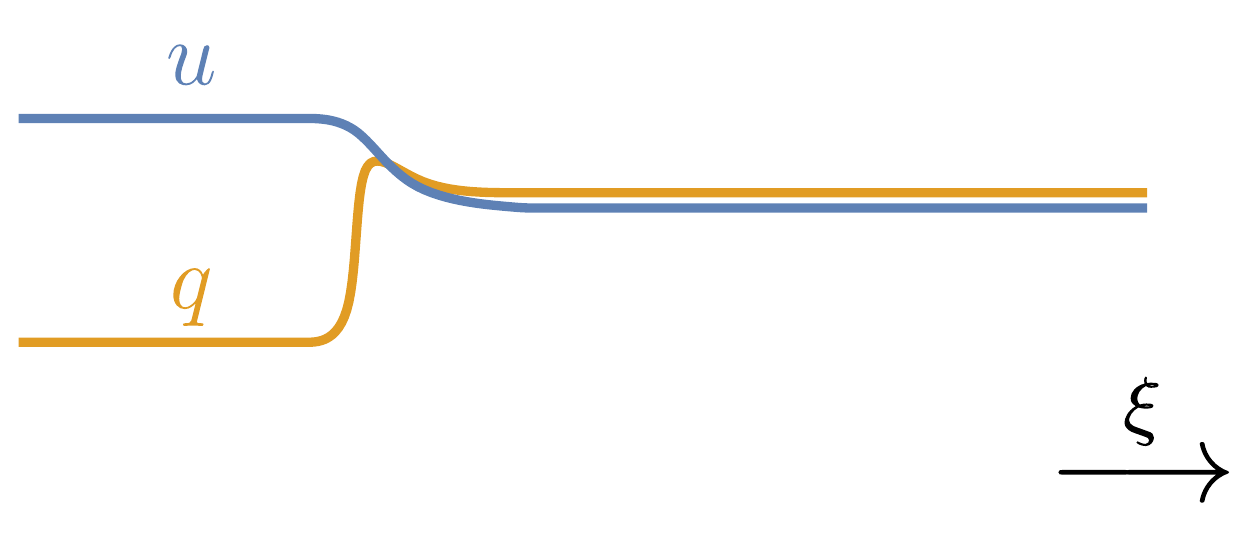}
\caption{Simple front}
\end{subfigure}
\hfill
\begin{subfigure}[b]{0.3\textwidth}
\centering
\includegraphics[scale=0.405]{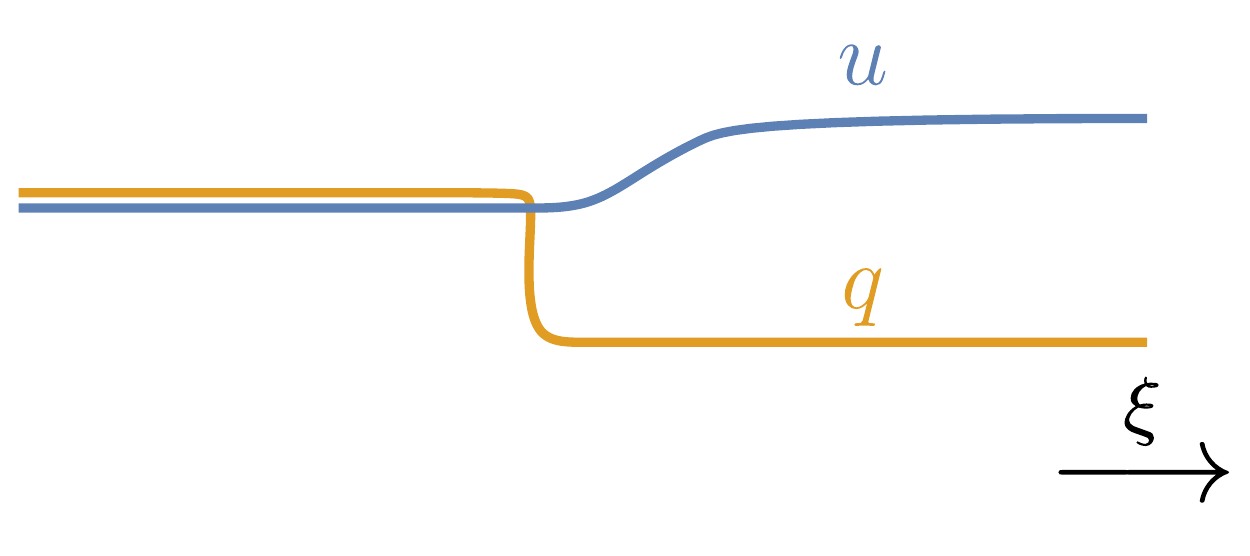}
\caption{Simple back}
\end{subfigure}
\hfill
\begin{subfigure}[b]{0.3\textwidth}
\centering
\includegraphics[scale=0.405]{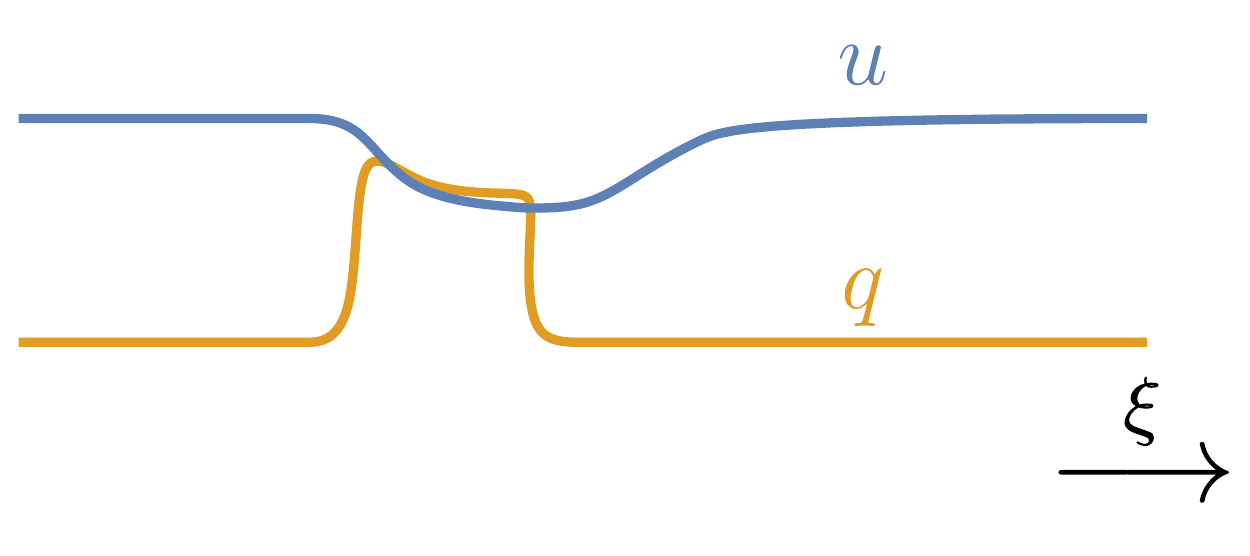}
\caption{Simple impulse}
\end{subfigure}\\
\begin{subfigure}[b]{0.3\textwidth}
\centering
\includegraphics[scale=0.405]{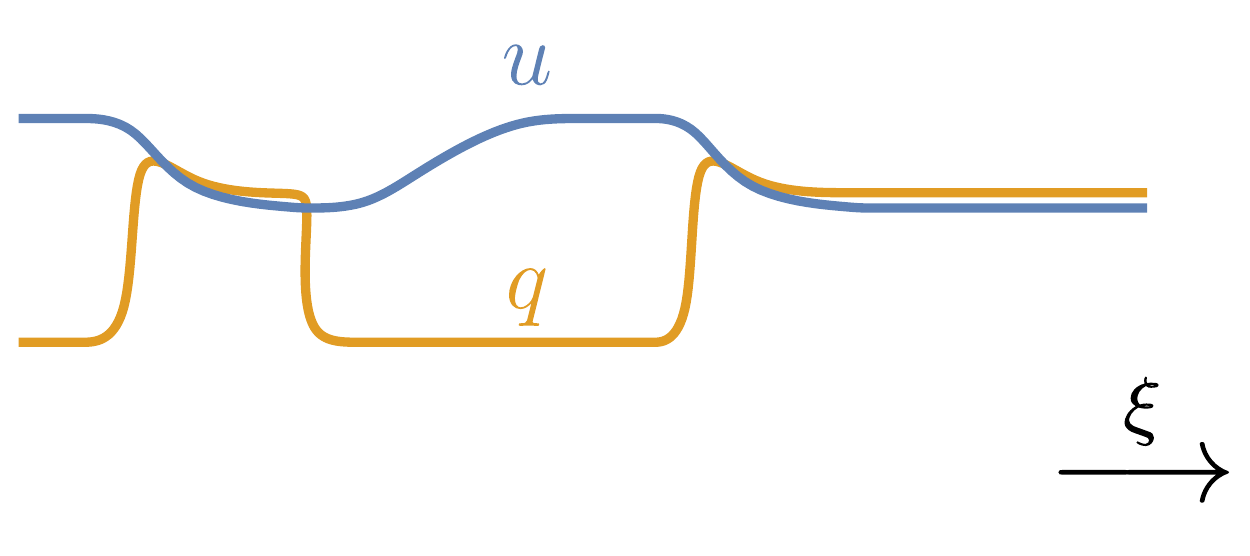}
\caption{$1$-front}
\end{subfigure}
\hfill
\begin{subfigure}[b]{0.3\textwidth}
\centering
\includegraphics[scale=0.405]{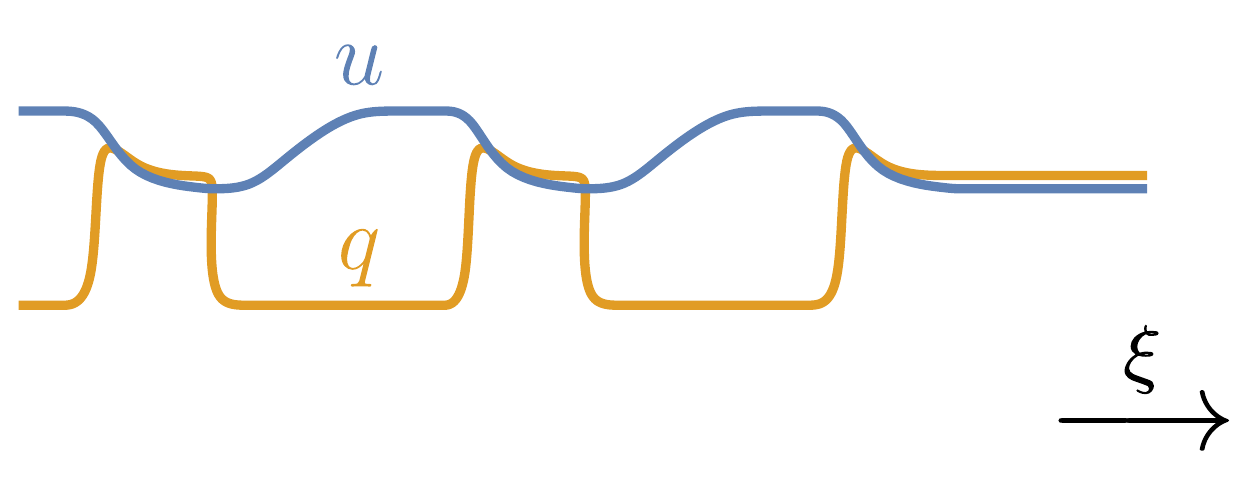}
\caption{$2$-front}
\end{subfigure}
\hfill
\begin{subfigure}[b]{0.3\textwidth}
\centering
\includegraphics[scale=0.405]{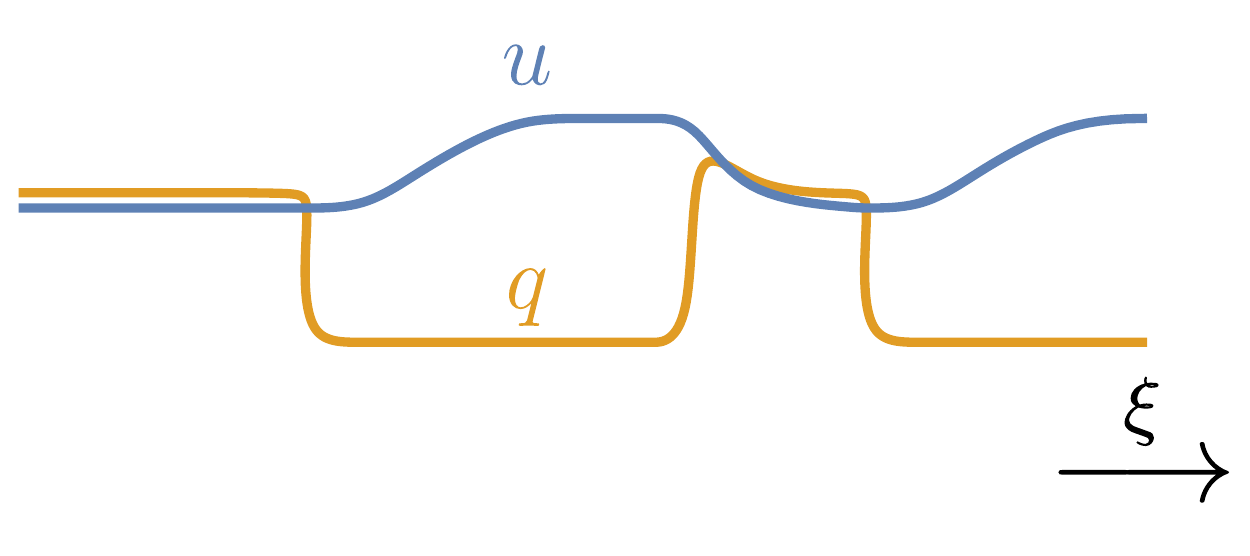}
\caption{$1$-back}
\end{subfigure}
\caption{\label{fig:1}
Schematic depiction of the profiles of various traveling waves exhibited by the pipe flow model~\eqref{SYS1}. The profiles exhibit sharp transitions in turbulence level $q$. Along such transitions the centerline velocity $u$ stays to leading order constant. In between the sharp transitions the turbulence level and centerline velocity evolve slowly. More precisely, in the presence of turbulence the centerline velocity slowly decreases down to the point where a new balance between turbulence and the centerline velocity is reached (corresponding to the equilibrium $X_2$ in~\eqref{SYS2}). On the other hand, in the absence of turbulence the centerline velocity slowly recovers up to the point where it reaches the laminar state (corresponding to the equilibrium $X_1$ in~\eqref{SYS2}). We note that the associated orbits in~\eqref{SYS2} lie in a tubular neighborhood of the heteroclinic loop depicted in Figure~\ref{fig:2}. Thus, the simple front and back depicted in the first two panels are the building blocks for the more complicated profiles depicted in the other panels.}
\end{figure}

\subsection{Existence of a heteroclinic loop}

We are now in the position to formulate our first result, which establishes a parameter regime in which the traveling-wave equation~\eqref{SYS2} admits a heteroclinic loop connecting the equilibria $X_1$ and $X_2$, see Figure~\ref{fig:2}.

We obtain such a heteroclinic loop by exploiting the fact that $\varepsilon$, which controls the time scale separation in~\eqref{SYS1} between fast excursions of $q$ relative to slow recovery of $u$ after relaminarization, is a small parameter. That is, by taking the limit $\varepsilon \downarrow 0$ in properly scaled versions of~\eqref{SYS2}, we arrive at the so-called fast and slow subsystems; we note that the fast subsystem is sometimes called layer equation and the slow subsystem is referred to as reduced system, which can be slightly confusing as both subsystems effectively ``reduce'' the dimensionality of the problem in the singular limit $\varepsilon \downarrow 0$ as we shall see below. Next, we look for parameter values for which a \emph{singular} heteroclinic loop exists, which is a concatenation of orbits in these slow and fast subsystems. Geometrically piecing these orbits together yields two algebraic matching conditions in the parameters $D,\mu$ and $r$ in~\eqref{SYS2}, which can then be explicitly solved for $D$ and $\mu$, leading to expressions $\mu_0(r)$ and $D_0(r)$, see Figure~\ref{fig:10}. An additional condition arises by requiring that the flow on the slow orbit segments is directed towards the equilibria $X_1$ and $X_2$. That is, fixing the free parameter $r > \frac{2}{3}$, we will observe that the wave speed $s$ has to satisfy $s < \min\{u_{\bb}(r),-\mu_0(r)\}$. Since the algebraic matching conditions yield $\mu_0(r) > -\frac{8}{5}$ and since we have $u_{\bb}(r) > \frac{6}{5}$, we will assume $\zeta > \frac{2}{5}$ in the following so that the variables $\mu$ and $s = -\zeta - \mu$ are interchangeable. We stress that $\zeta > \frac{2}{5}$ captures the physically relevant regime. Indeed, in~\cite{Barkleyetal}, one finds $\zeta = 0.79$ for pipe flow and $\zeta = 0.56$ for duct flow to be in good accordance with experimental data.

Having established a singular heteroclinic loop, we can then employ Melnikov's method and geometric singular perturbation theory, cf.~\cite{KuehnBook} and references therein, to prove that, for each $\varepsilon > 0$ sufficiently small, an actual heteroclinic loop exists in the vicinity of the singular one for parameter values $(\mu,D)$ close to $(\mu_0(r),D_0(r))$. 

The Melnikov analysis provides the technical nondegeneracy assumption $\smash{\widehat{M}}(r) \neq 0$ in terms of the function $\smash{\widehat{M}} \colon (\frac{2}{3},\infty) \to \R$, defined in~\eqref{defhatM}. We emphasize that $\smash{\widehat{M}}(r)$ consists of Melnikov integrals, which are all fully explicit in terms of $r$.\footnote{This can be readily seen by using identities~\eqref{qbbdef},~\eqref{defqf},~\eqref{defqb-},~\eqref{defphi},~\eqref{defD0mu0},~\eqref{defhatMf},~\eqref{defMfr},~\eqref{defhatMb} and~\eqref{defMbr} and the fact that $u_{\bb}(r)$ is the smallest root of the cubic~\eqref{defcubic}, see Lemma~\ref{lem:equilibria}.} Thus, the nondegeneracy assumption $\smash{\widehat{M}}(r)\neq0$ could theoretically be verified. However, we refrain from doing so as the resulting expressions are highly involved. A numerical computation, see also the plot in Figure~\ref{fig:122}, suggests that the assumption is in fact satisfied for all $r > \frac{2}{3}$. On the other hand, we will theoretically establish $\lim_{r \to \infty} \smash{\widehat{M}}(r) \neq 0$, so that the assumption is rigorously satisfied for all $r > \frac{2}{3}$ sufficient large.

All in all, we establish the following result, which is valid in the regime of intermediate and large Reynolds number.

\begin{theorem} \label{heteroclinicloop}
There are smooth functions $\mu_0 \colon \left(\frac{2}{3},\infty\right) \to \left(-\frac{8}{5}, \smash{\frac{1}{66} \left(3 \sqrt{115}-65\right)}\right)$ and $D_0 \colon \left(\frac{2}{3},\infty\right)  \to \left(0,\infty\right)$ satisfying
\begin{align}\lim_{r \downarrow \frac{2}{3}} \mu_0(r) = \frac{1}{66} \left(3 \sqrt{115}-65\right), \qquad \lim_{r \to \infty} \mu_0(r) = -\frac{8}{5}, \label{limitsspeed}\end{align}
and
\begin{align}\lim_{r \downarrow \frac{2}{3}} D_0(r) = \frac{10}{363} \left(34 + 3 \sqrt{115}\right), \qquad \lim_{r \to \infty} D_0(r) = 0,\label{limitsspeed2}\end{align}
such that for each fixed Reynolds parameter $r > \frac{2}{3}$ satisfying $\smash{\widehat{M}}(r) \neq 0$, cf.~\eqref{defhatM}, there exists $\varepsilon_0(r) > 0$ such that the following holds:
for each $\varepsilon \in (0,\varepsilon_0(r))$ there exist a diffusion rate $D = \smash{\widehat{D}}(\varepsilon,r)$ and a velocity $\mu = \smash{\widehat{\mu}}(\varepsilon,r)$ such that system~\eqref{SYS2} admits a heteroclinic loop, which consists of a simple heteroclinic front and a simple heteroclinic back connecting the equilibria $X_1$ and $X_2$. The functions $\smash{\widehat{D}}(\varepsilon,r)$ and $\smash{\widehat{\mu}}(\varepsilon,r)$ are smooth in their variables and it holds
\begin{align} \lim_{\varepsilon \downarrow 0} \smash{\widehat{D}}(\varepsilon,r) = D_0(r), \qquad \lim_{\varepsilon \downarrow 0} \smash{\widehat{\mu}}(\varepsilon,r) = \mu_0(r). \label{limits}\end{align}
\end{theorem}

\begin{figure}
\centering
\includegraphics[scale=0.5]{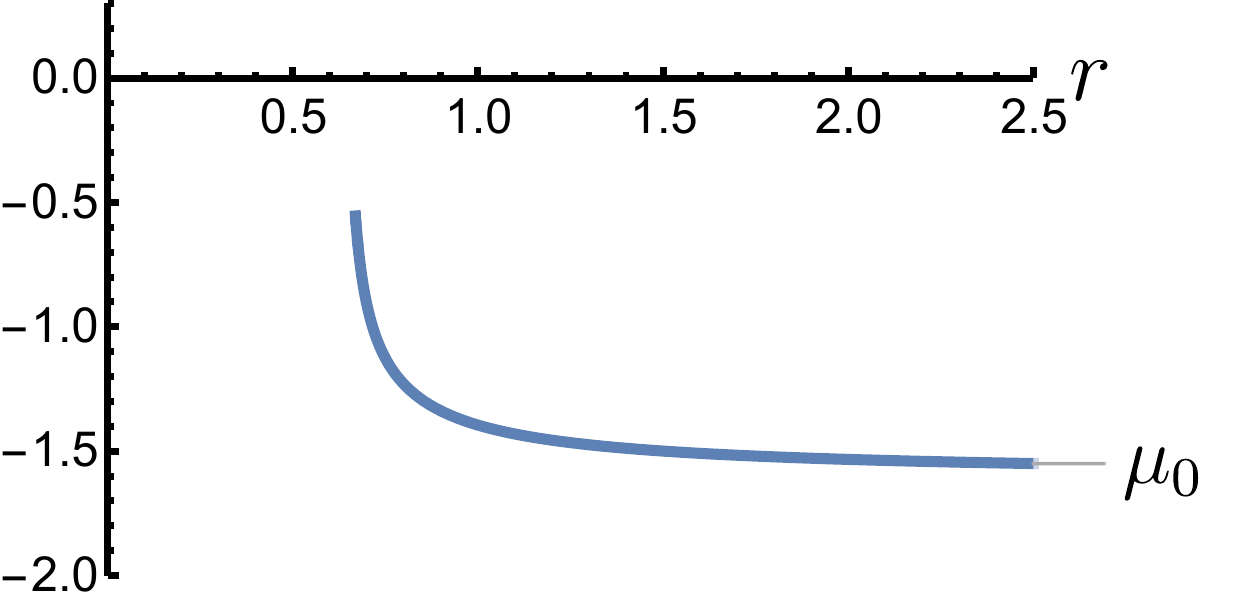} \hspace{1cm} \includegraphics[scale=0.5]{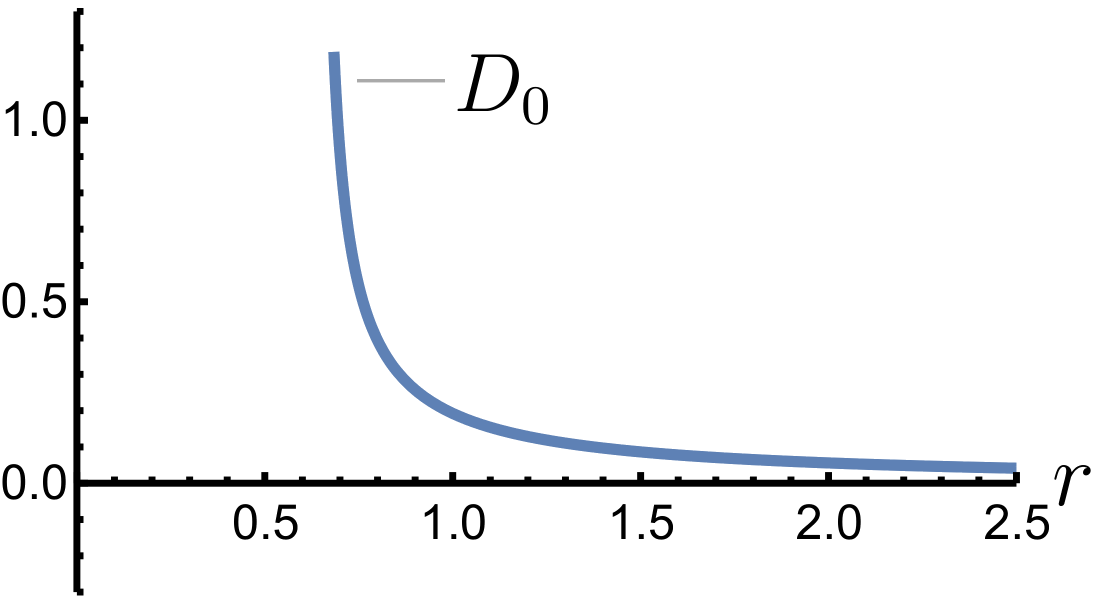}
\caption{Plots of the functions $D_0(r)$ and $\mu_0(r)$ established in Theorem~\ref{heteroclinicloop}.}
\label{fig:10}
\end{figure}

We emphasize that our analysis provides explicit expressions for the functions $\mu_0(r)$ and $D_0(r)$ in Theorem~\ref{heteroclinicloop}, see~\eqref{defD0mu0}. Moreover, the wave speed of the orbits in the heteroclinic loop is selected through $\smash{\widehat{s}}(\varepsilon,r,\zeta) = -\zeta-\smash{\widehat{\mu}}(\varepsilon,r)$. Noting that
\begin{align}s_0(r,\zeta) := \lim_{\varepsilon \downarrow 0} \smash{\widehat{s}}(\varepsilon,r,\zeta) = - \zeta - \mu_0(r), \label{defs0}\end{align}
we observe that the wave speed $\smash{\widehat{s}}(\varepsilon,r,\zeta)$ can be positive or negative for all $\varepsilon \in (0,\varepsilon_0(r))$, depending on the values of $r$ and $\zeta$.  That is, the heteroclinic loop can consist of a simple front recovering the laminar state and a simple back invading the laminar state, or a simple front invading the laminar state and a simple back recovering the laminar state, depending on the values of $r$ and $\zeta$. For instance, taking $\zeta = 0.79$ as in~\cite{Barkleyetal}, we find by~\eqref{limitsspeed} that the wave speed $\smash{\widehat{s}}(\varepsilon,r,\zeta)$ is negative for intermediate $r$-value close to $\frac{2}{3}$ and positive for large $r$-values. For all further details and the proof of Theorem~\ref{heteroclinicloop} we refer to~\S\ref{sec:slowfast}.

\begin{remark} \label{rem:expand}
{\upshape
Also in~\cite{Barkleyetal} simple fronts and backs are constructed for the pipe flow model~\eqref{SYS1} in the bistable regime $r > \frac{2}{3}$ (albeit formally). One finds the same matching conditions for the existence of the associated singular forward and backward heteroclinic connecting $X_1$ with $X_2$ in~\eqref{SYS2}, see~\eqref{cond:front} and~\eqref{cond:back}. However, instead of solving this as a \emph{system} of two algebraic equations in $D,\mu,r$ with respect to $D$ and $\mu$ in order to find a singular heteroclinic loop, the equations are solved separately with respect to $\mu$ leading to two solutions $\mu_1(r,D)$ and $\mu_2(r,D)$. That is, fixing free parameters $D,r > 0$, one finds a simple front and a simple back in~\eqref{SYS1} propagating with \emph{different} speeds. These fronts and backs are then formally pieced together in~\cite{Barkleyetal} yielding an expanding plateau state in~\eqref{SYS1}. We stress that such expanding states do no satisfy the definition of a traveling wave (as their profiles are not fixed), and are therefore necessarily different from the traveling-wave solutions constructed in this paper.
}\end{remark}

\begin{remark}
{\upshape
We note that given free parameters $r > \frac{2}{3}$, $\zeta > \frac{2}{5}$ and $0 < \varepsilon < \varepsilon_0(r)$, Theorem~\ref{heteroclinicloop} selects a diffusion coefficient $D = \smash{\widehat{D}}(\varepsilon,r)$ and a speed $s = \smash{\widehat{s}}(\varepsilon,r,\zeta) = -\zeta-\smash{\widehat{\mu}}(\varepsilon,r)$ for which a heteroclinic loop exists. Although the selection of a wave speed by the model parameters is natural, see also~\cite{Barkleyetal}, the selection of a diffusion coefficient indicates that the existence of a heteroclinic loop is a codimension one-phenomenon. In contrast, the traveling waves, whose profiles lie in the vicinity of the heteroclinic loop and which are constructed in upcoming Theorems~\ref{thm:single_twist} and~\ref{thm:double_twist}, exist in \emph{open} regions of the $(r,\zeta,D,\varepsilon)$-parameter space with only the wave speed being selected. More precisely, for any combination of model parameters $r > \frac{2}{3}$, $\zeta > \frac{2}{5}$, $\varepsilon \in (0,\varepsilon_0(r))$ and $D \in [\smash{\widehat{D}}(\varepsilon,r)-\delta(\varepsilon,r),\smash{\widehat{D}}(\varepsilon,r) + \delta(\varepsilon,r)]$, a wave speed $s$ is selected by the bifurcation curves in Figure~\ref{fig:twist}. 

We note that for the same parameter values there is no co-existence of the heteroclinic loop, which is established in Theorem~\ref{heteroclinicloop} and corresponds to the intersection point of the bifurcation curves in Figure~\ref{fig:twist}, and the traveling waves, which are established in Theorems~\ref{thm:single_twist} and~\ref{thm:double_twist} and correspond to all points on the bifurcation curves away from the intersection point in Figure~\ref{fig:twist}.}
\end{remark}

\subsection{Large Reynolds number regime}

We state our first bifurcation result about the heteroclinic loop, for the case of large Reynolds number.

\begin{theorem} \label{thm:single_twist}
Consider system~\eqref{SYS1} and let $\smash{\widehat{D}}(\varepsilon,r)$ and $\smash{\widehat{s}}(\varepsilon,r,\zeta) = -\zeta-\smash{\widehat{\mu}}(\varepsilon,r)$ be as in Theorem~\ref{heteroclinicloop}. For any $\zeta > \frac{2}{5}$ and any sufficiently large $r > \frac{2}{3}$ there exists $\varepsilon_0(r) > 0$ such that for any $\varepsilon \in (0,\varepsilon_0(r))$ there exists $\delta(\varepsilon,r) > 0$, which depends smoothly on $\varepsilon$ and $r$, such that the following holds:
\begin{enumerate}
\item \textbf{Simple fronts and back}: In the $(D,s)$-parameter plane there are the smooth curves $s= s_{i,0}(D)$, $i=\f,\bb$, defined on the interval $[\smash{\widehat{D}}(\varepsilon,r) - \delta(\varepsilon,r), \smash{\widehat{D}}(\varepsilon,r) + \delta(\varepsilon,r)]$, corresponding to simple fronts of speed $s_{\f,0}(D)$ and simple backs of speed $s_{\bb,0}(D)$, respectively, which intersect transversely at $(\smash{\widehat{D}}(\varepsilon,r),\smash{\widehat{s}}(\varepsilon,r,\zeta))$.
\item \textbf{Bifurcation of $1$-backs, and generation of impulses}: There are smooth curves $s_{\bb,1}(D)$ and $s_{\bb,\infty}(D)$, defined on $[\smash{\widehat{D}}(\varepsilon,r)-\delta(\varepsilon,r), \smash{\widehat{D}}(\varepsilon,r)]$, corresponding to $1$-backs of wave speed $s_{\bb,1}(D)$ and simple impulses of $X_2$ of wave speed $s_{\bb,\infty}(D)$, respectively. Furthermore, there is a smooth curve $s_{\f,\infty}(D)$, defined on $[\smash{\widehat{D}}(\varepsilon,r), \smash{\widehat{D}}(\varepsilon,r)+\delta(\varepsilon,r)]$, of waves speeds associated with simple impulses of $X_1$. The curves $s_{\f,\infty}(D), s_{\bb,\infty}(D)$ and $s_{\bb,1}(D)$ intersect at $(\smash{\widehat{D}}(\varepsilon,r),\smash{\widehat{s}}(\varepsilon,r,\zeta))$.
\end{enumerate}
\end{theorem}
The statements are illustrated in Figure~\ref{fig:twist}. Note that, for an open region in the $(r,\zeta,D,\varepsilon)$-parameter space corresponding to large Reynolds numbers, we have established the existence of various heteroclinic and homoclinic structures associated with different patterns of turbulent puffs. Whereas the homoclinics of $X_1$ and $X_2$ correspond to excursions from the laminar and turbulent states, respectively, the heteroclinics describe changes between the two regimes.

In the large Reynolds number regime it follows by~\eqref{limitsspeed} and~\eqref{defs0} that the selected wave speed $\smash{\widehat{s}}(\varepsilon,r,\zeta)$ can be positive or negative for all $\varepsilon \in (0,\varepsilon_0(r))$, depending on the value of $\zeta$. Indeed, taking $\zeta = 0.79$ as in~\cite{Barkleyetal} one finds a positive wave speed, whereas taking $\zeta > \frac{8}{5}$ leads to a negative wave speed for sufficiently large $r > \frac{2}{3}$ and sufficiently small $\varepsilon > 0$. Consequently, the fronts, backs and impulses established in Theorem~\ref{thm:single_twist} can be propagating upstream or downstream.

For the proof of Theorem~\ref{thm:single_twist}, as well as the upcoming Theorem~\ref{thm:double_twist}, we will follow the arguments by Bo Deng~\cite{Deng91b} and Homburg and Sandstede~\cite{HomburgSandstede}. Note that the statements in Theorems~\ref{thm:single_twist} and~\ref{thm:double_twist} correspond directly to statements on heteroclinic and homoclinic orbits in terms of the ODE~\eqref{SYS2}. The existence of such heteroclinic and homoclinic orbits follow from an application of Deng's general results on the bifurcations of a single twisted (Theorem~\ref{thm:single_twist}) and double twisted (Theorem~\ref{thm:double_twist}) heteroclinic loop~\cite{Deng91a}. To apply Deng's result we verify five conditions on the vector field of the ODE~\eqref{SYS2}. Here, we use a small variation of~\cite[Theorem 2.1]{Deng91b}, in combination with~\cite[Hypothesis 5.16 (ii)]{HomburgSandstede}. For all further details and the proof of Theorems~\ref{thm:single_twist} and~\ref{thm:double_twist} we refer to~\S\ref{sec:proof_twists}.

\begin{remark}
\label{rem:largeRE}
{\upshape
We note that, in the case of positive wave speed, the upstream $1$-backs are the most intricate solutions established in Theorem~\ref{thm:single_twist}, which correspond to turbulent flows that almost fully relaminarize before going back again to the vicinity of the turbulent steady state $X_2$, just to finally leave this regime and asymptotically relaminarize again, see Figure~\ref{fig:1}(f). This may seem surprising for large $r$, as one would rather expect the stronger concentration at the turbulent part. However, note that such an upstream $1$-back does not have to be a stable object.\footnote{For the general question of stability of the established traveling waves we refer to the discussion in~\S\ref{sec:outlook}.} Moreover, Theorem~\ref{thm:single_twist} does not imply that no other (bifurcating) traveling waves exist in~\eqref{SYS1}, which might represent invasion of turbulence, see also Remark~\ref{rem:compare}.
}\end{remark}

\subsection{Intermediate Reynolds number regime}

Our bifurcation result about the heteroclinic loop for the intermediate Reynolds number regime requires the following technical assumption:
\begin{hypothesis} \label{hyp}
There exists $\gamma > 0$ such that the function $\smash{\widetilde{M}}_\f \colon (\frac{2}{3},\infty) \to \R$ defined by~\eqref{defMtf} satisfies $\smash{\widetilde{M}}_\f(r) > 0$ for all $r \in (\frac{2}{3},\frac{2}{3} + \gamma)$.
\end{hypothesis}

We emphasize that $\smash{\widetilde{M}}_\f(r)$ is fully explicit in terms of $r$. Thus, Hypothesis~\ref{hyp} could theoretically be verified. However, we refrain from doing so as the resulting expressions are highly involved. A numerical computation, see also the plot in Figure~\ref{fig:12}, suggest that one can take $\gamma = 0.0627$.

We are now in the position to formulate our bifurcation result about the heteroclinic loop for the intermediate Reynolds number regime.

\begin{theorem} \label{thm:double_twist}
Assume Hypothesis~\ref{hyp} is satisfied and let $\smash{\widehat{D}}(\varepsilon,r)$ and $\smash{\widehat{s}}(\varepsilon,r,\zeta) = -\zeta-\smash{\widehat{\mu}}(\varepsilon,r)$ be as in Theorem~\ref{heteroclinicloop}. For system~\eqref{SYS1}, taking $\zeta > \frac{2}{5}$ and $r  \in (\frac{2}{3}, \frac{2}{3} + \gamma)$, satisfying $\smash{\widehat{M}}(r) \neq 0$, cf.~\eqref{defhatM}, there exists $\varepsilon_0(r) > 0$ such that for any $\varepsilon \in (0,\varepsilon_0(r))$ there exists $\delta(\varepsilon,r) > 0$, which depends smoothly on $\varepsilon$ and $r$, such that the following holds:
\begin{enumerate}
\item \textbf{Simple fronts and back}: In the $(D,s)$-parameter plane there are the smooth curves $s= s_{i,0}(D)$, $i=\f, \bb$, defined on the interval $[\smash{\widehat{D}}(\varepsilon,r) - \delta(\varepsilon,r), \smash{\widehat{D}}(\varepsilon,r) + \delta(\varepsilon,r)]$, corresponding to simple fronts of speed $s_{\f,0}(D)$ and simple backs of speed $s_{\bb,0}(D)$, respectively, which intersect transversely at $(\smash{\widehat{D}}(\varepsilon,r),\smash{\widehat{s}}(\varepsilon,r,\zeta))$.
\item \textbf{Bifurcation of $k$-fronts and $k$-backs, and generation of impulses}: There is a sequence $\{s_{\f,k}(D)\}_{k=1}^{\infty}$ of smooth curves, defined on $[\smash{\widehat{D}}(\varepsilon,r), \smash{\widehat{D}}(\varepsilon,r)+\delta(\varepsilon,r)]$, corresponding to $k$-fronts of wave speed $s_{\f,k}(D)$, 
and converging to a curve $s_{\f,\infty}(D)$ of waves speeds  associated with simple impulses of $X_1$. Similarly, there is a sequence $\{s_{\bb,k}(D)\}_{k=1}^{\infty}$ of smooth curves, defined on $[\smash{\widehat{D}}(\varepsilon,r)-\delta(\varepsilon,r), \smash{\widehat{D}}(\varepsilon,r)]$, corresponding to $k$-backs of wave speed $s_{\bb,k}(D)$, 
and converging to a curve $s_{\bb,\infty}(D)$ of wave speeds  associated with simple impulses of $X_2$. The curves $s_{\f,k}(D), s_{\bb,k}(D), k = 1,\ldots,\infty$ intersect at $(\smash{\widehat{D}}(\varepsilon,r),\smash{\widehat{s}}(\varepsilon,r,\zeta))$.
\end{enumerate}
\end{theorem}
The statements are illustrated in Figure~\ref{fig:twist}. The key observation here is that local bifurcations along the double-twisted heteroclinic loop generate infinitely many $k$-front and $k$-back solutions for arbitrary $k\in \mathbb{N}$. These are traveling waves whose profiles exhibit $k$ well separated patches of turbulence, before converging towards fixed states at $\pm \infty$.

Here, the wave speed $\smash{\widehat{s}}(\varepsilon,r,\zeta)$ can also be positive or negative for all $\varepsilon \in (0,\varepsilon_0(r))$ depending on the value of $\zeta$. Indeed, it follows by~\eqref{limitsspeed} and~\eqref{defs0} that for $\zeta \in \left(\frac{2}{5},\frac{1}{66} (65-3 \sqrt{115})\right)$ the wave speed is positive, whereas for $\zeta >  \frac{1}{66} (65-3 \sqrt{115})$ the wave speed is negative, provided $r > \frac{2}{3}$ lies sufficiently close to $\frac{2}{3}$ and $\varepsilon > 0$ is sufficiently small. Consequently, the $k$-fronts and $k$-backs established in Theorem~\ref{thm:double_twist} can be propagating upstream or downstream, which determines whether they are invading or recovering the laminar state.

Further details, as well as the proof of Theorem~\ref{thm:double_twist}, can be found in~\S\ref{sec:proof_twists}.

\begin{remark} \label{rem:compare}
{\upshape
Comparing Theorems~\ref{thm:single_twist} and~\ref{thm:double_twist} it may seem rather surprising that the result in the large Reynolds number regime provides fewer bifurcating traveling waves than in the intermediate Reynolds number regime, whereas one could naively expect that a large Reynolds number leads to ``more complex'' pipe flow. There are several reasons that may explain why we obtain fewer bifurcating travelling waves for large Reynolds numbers here. First, the (statistical) description of fully-developed turbulence may potentially be possible with a much simpler phase space structure compared to the intermediate regime characterizing the transition from laminar to turbulent flow. Second, we stress that Theorem~\ref{thm:single_twist} does not imply that no other bifurcating traveling waves exist. Third, the pipe flow model~\eqref{SYS1} might exhibit intricate solutions for large Reynolds number, which are not of traveling-wave type, or which are of traveling-wave type, but do not arise as bifurcating orbits from the heteroclinic loop. Finally, as outlined in~\S\ref{sec:outlook} the pipe flow model has been validated experimentally and numerically for the onset of turbulence and it is therefore an open question, whether it is valid in the large Reynolds number regime.
}\end{remark}

\begin{figure}
\centering
\begin{subfigure}[b]{0.45\textwidth}
\centering
\begin{overpic}[scale=0.4]{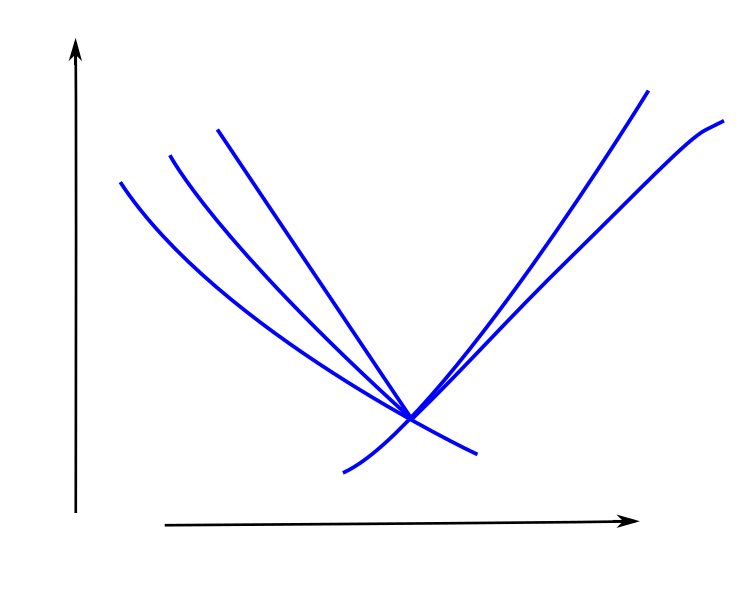}
\put(50, 3) { $D$}
\put(0, 40) { $s$}
\put(80, 45) { $s_{\f,\infty}$ }
\put(70, 60) { $s_{\f,0}$ }
\put(15, 62) { $s_{\bb,\infty}$ }
\put(33, 60) { $s_{\bb,1}$ }
\put(20, 35) { $s_{\bb,0}$ }
\end{overpic}
\caption{Single twist}
\end{subfigure}
\hfill
\begin{subfigure}[b]{0.45\textwidth}
\centering
\begin{overpic}[scale=0.4]{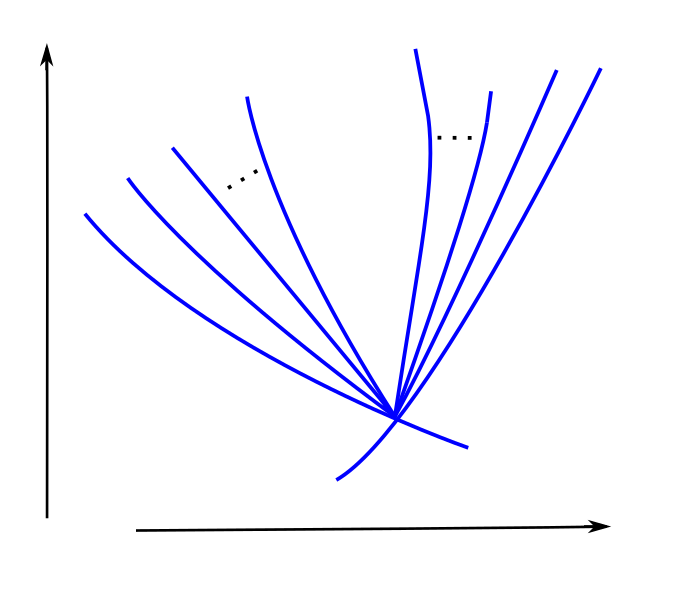}
\put(50, 3) { $D$}
\put(-2, 40) { $s$}
\put(57, 79) { $s_{\f,\infty}$ }
\put(76, 50) { $s_{\f,0}$ }
\put(64, 74) { $s_{\f,2}$ }
\put(74, 78) { $s_{\f,1}$ }
\put(25, 73) { $s_{\bb,\infty}$ }
\put(10, 62) { $s_{\bb,1}$ }
\put(17, 66) { $s_{\bb,2}$ }
\put(20, 35) { $s_{\bb,0}$ }
\end{overpic}
\caption{Double twist}
\end{subfigure}
\caption{Bifurcation diagrams for the single twisted (a) and double twisted (b) heteroclinic loop, as found in Theorem~\ref{thm:single_twist} and Theorem~\ref{thm:double_twist}, respectively. The figures show the various functions $s(D)$ in the $(D,s)$-parameter plane, along which different kinds of traveling-wave solutions to~\eqref{SYS1} exist. Note that in both figures the transversal intersection point of $s_{\bb,0}$ and $s_{\f,0}$ is precisely $(\smash{\widehat{D}}(\varepsilon,r),\smash{\widehat{s}}(\varepsilon,r,\zeta))$, where $\smash{\widehat{s}}(\varepsilon,r,\zeta) = - \zeta - \smash{\widehat{\mu}}(\varepsilon,r)$, which corresponds to the heteroclinic loop established in Theorem~\ref{heteroclinicloop}.}
\label{fig:twist}
\end{figure}

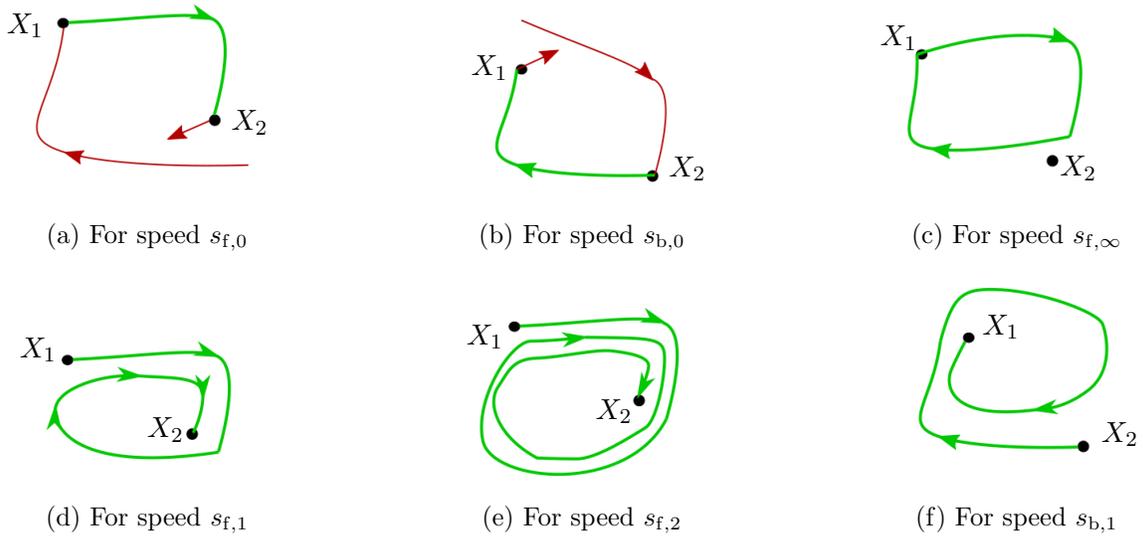
\begin{figure}
\centering
\begin{subfigure}[b]{0.3\textwidth}
\centering
\begin{overpic}[scale=0.4]{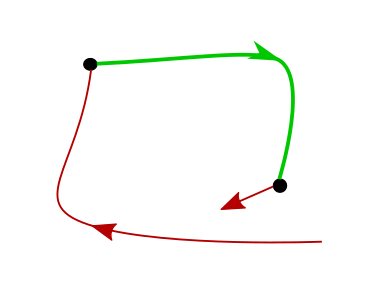}
\put(5,55){$X_1$}
\put(77,24){$X_2$}
\end{overpic}
\caption{For speed $s_{\f,0}$}
\end{subfigure}
\hfill
\begin{subfigure}[b]{0.3\textwidth}
\centering
\begin{overpic}[scale=0.4]{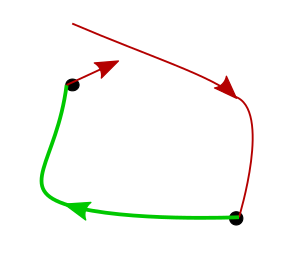}
\put(3,55){$X_1$}
\put(87,12){$X_2$}
\end{overpic}
\caption{For speed $s_{\bb,0}$}
\end{subfigure}
\hfill
\begin{subfigure}[b]{0.3\textwidth}
\centering
\begin{overpic}[scale=0.4]{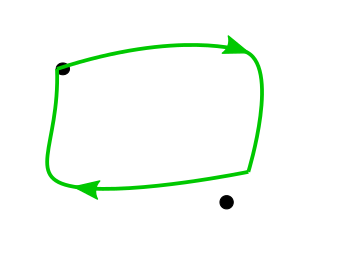}
\put(3,55){$X_1$}
\put(65,10){$X_2$}
\end{overpic}
\caption{For speed $s_{\f,\infty}$}
\end{subfigure}\\
\begin{subfigure}[b]{0.3\textwidth}
\centering
\begin{overpic}[scale=0.4]{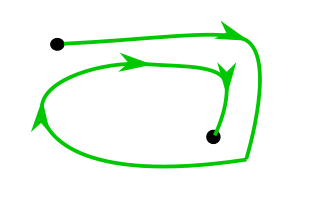}
\put(-1,52){$X_1$}
\put(50,20){$X_2$}
\end{overpic}
\caption{For speed $s_{\f,1}$}
\end{subfigure}
\hfill
\begin{subfigure}[b]{0.3\textwidth}
\centering
\begin{overpic}[scale=0.4]{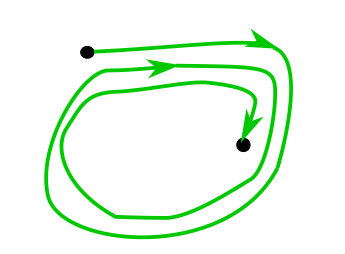}
\put(8,52){$X_1$}
\put(55,25){$X_2$}
\end{overpic}
\caption{For speed $s_{\f,2}$}
\end{subfigure}
\hfill
\begin{subfigure}[b]{0.3\textwidth}
\centering
\begin{overpic}[scale=0.4]{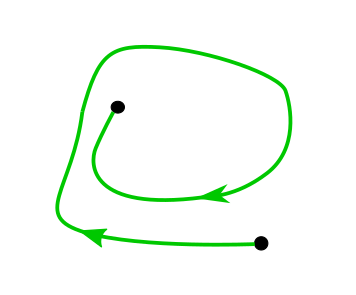}
\put(38,54){$X_1$}
\put(80,16){$X_2$}
\end{overpic}
\caption{For speed $s_{\bb,1}$}
\end{subfigure}
\caption{\label{fig:hets_and_homs}
Schematic depiction of various homoclinic and heteroclinic orbits exhibited by the traveling-wave ODE~\eqref{SYS2} associated with the different bifurcation curves shown in Figure~\ref{fig:twist}. The orbits lie in the vicinity of the heteroclinic loop depicted in Figure~\ref{fig:2} and are in one-to-one correspondence to the traveling waves of the pipe flow model~\eqref{SYS1} portrayed in Figure~\ref{fig:1}: (a) a simple heteroclinic from $X_1$ to $X_2$ (simple front), (b) a simple heteroclinic from $X_2$ to $X_1$ (simple back), (c) a simple homoclinic of $X_1$ (simple impulse of $X_1$), (d) a $1$-heteroclinic from $X_1$ to $X_2$ ($1$-front), (e) a $2$-heteroclinic from $X_1$ to $X_2$ ($2$-front), (f) a $1$-heteroclinic from $X_2$ to $X_1$ ($1$-back).}
\end{figure}

\section{Proof of Theorem~\ref{heteroclinicloop}} \label{sec:slowfast}

The traveling-wave equation~\eqref{SYS2} has the structure
\begin{align*}
\Psi_\xi &= F_{\text{fast}}(\Psi,\Phi),\\
\Phi_\xi &= \varepsilon F_{\text{slow}}(\Psi,\Phi),
\end{align*}
of a fast-slow dynamical system, where $\varepsilon > 0$ is a small parameter. The dynamics in such systems can geometrically be described by fast and slow subsystems that arise in the limit $\varepsilon \downarrow 0$. In our setting of the pipe flow model~\eqref{SYS1}, the fast subsystem captures sharp interfaces, which represent a quick drop or rise in turbulence level while the centerline velocity is unaffected. On the other hand, the slow subsystem describes the dynamics of~\eqref{SYS1} in between these interfaces, where the turbulence and centerline velocity are `slaved' to each other and evolve slowly, cf.~Figure~\ref{fig:1}.

We outline our approach to constructing the desired heteroclinic loop in the fast-slow system~\eqref{SYS2} and proving Theorem~\ref{heteroclinicloop}. The first step is to geometrically assemble a \emph{singular} heteroclinic loop by concatenating orbits of the fast and slow subsystems. We study the slow and fast subsystems in~\S\ref{sec:slow} and~\S\ref{sec:fast} and establish parameter regimes such that the relevant orbits exist. The construction of the singular heteroclinic loop can be found in~\S\ref{sec:singhet}. Subsequently, we will prove that an actual heteroclinic loop exists in the vicinity of the singular one, provided $\varepsilon > 0$ is sufficiently small. In more technical detail, finding heteroclinic connections between the equilibria $X_1$ and $X_2$ boils down to locating intersections between the stable and unstable manifolds of $X_1$ and $X_2$ in the dynamical system~\eqref{SYS2}. Indeed, since the stable manifold of the equilibrium $X_1$ consists of all orbits in~\eqref{SYS2} converging to $X_1$ as $\xi \to \infty$ and the unstable manifold of $X_2$ consists of all orbits in~\eqref{SYS2} converging to $X_2$ as $\xi \to -\infty$, any heteroclinic connecting $X_2$ with $X_1$ must lie in the intersection of the stable manifold of $X_1$ and unstable manifold of $X_2$. Therefore, we need good mathematical control on the relevant stable and unstable manifolds. We obtain such control through geometric singular perturbation theory in~\S\ref{sec:persistence_manifolds}, which allows us to describe the stable and unstable manifolds in terms of the fast and slow subsystems. Subsequently, we employ Melnikov's method in~\S\ref{sec:Meln} to locate intersections between the stable and unstable manifolds for $\varepsilon > 0$ sufficiently small, which yields the existence of the desired heteroclinic loop, see~\S\ref{sec:hetloop}.

\subsection{Slow subsystem} \label{sec:slow}

To capture the slow dynamics in the fast-slow dynamical system~\eqref{SYS2}, we introduce the `stretched' spatial coordinate $\tau = \varepsilon \xi$. In this rescaled spatial coordinate the system reads
\begin{align}
\begin{split}
\varepsilon q_\tau &= p,\\
\varepsilon p_\tau &= D^{-1}\left((u + \mu)p - f(q,u;r)\right),\\
(u-s)u_\tau &= g(q,u).
\end{split} \label{SYS22}
\end{align}
Subsequently setting $\varepsilon = 0$ we arrive at the \emph{slow subsystem}
\begin{align}
\begin{split}
0 &= p,\\
0 &= f(q,u;r),\\
(u-s)u_\tau &= g(q,u).
\end{split} \label{SLOW}
\end{align}
We note that~\eqref{SLOW} is a differential-algebraic system of equations in which the dynamics is one-dimensional as the $q$-component is slaved to the $u$-component through the relation $f(q,u;r) = 0$. That is, orbits in~\eqref{SLOW} are confined to the nullcline
\begin{align*} M_0 = \left\{(q,0,u) \in \R^3 : f(q,u;r) = 0\right\},\end{align*}
which is also called \emph{critical manifold}. It is the union of the line and the parabola
\begin{align*} M_1 = \{(0,0,u) : u \in \R\}, \qquad M_2 = \{(q,0,2-r+(r+0.1)(q-1)^2) : q \in \R\},\end{align*}
which intersect in the point $(0,0,2.1)$, see Figure~\ref{fig:3}. The parabola attains its global minimum value at the point $(1,0,2-r)$. The equilibria of both the slow subsystem~\eqref{SLOW} and the traveling-wave equation~\eqref{SYS2} with $\varepsilon > 0$ are located by intersecting $M_0$ with the second nullcline, which is the hyperbola
\begin{align*}H_0 = \left\{(q,0,u) \in \R^3 : g(q,u) = 0\right\}.\end{align*}
As already mentioned before, one of these equilibria is $X_1 = (0,0,2)$ which corresponds to the laminar flow profile in~\eqref{SYS1}. For Reynolds parameter $r > \frac{2}{3}$ we establish a second, $r$-dependent equilibrium $X_2$, which corresponds to a turbulent steady state in~\eqref{SYS1}.

\begin{figure}
\centering
\includegraphics[scale=0.5]{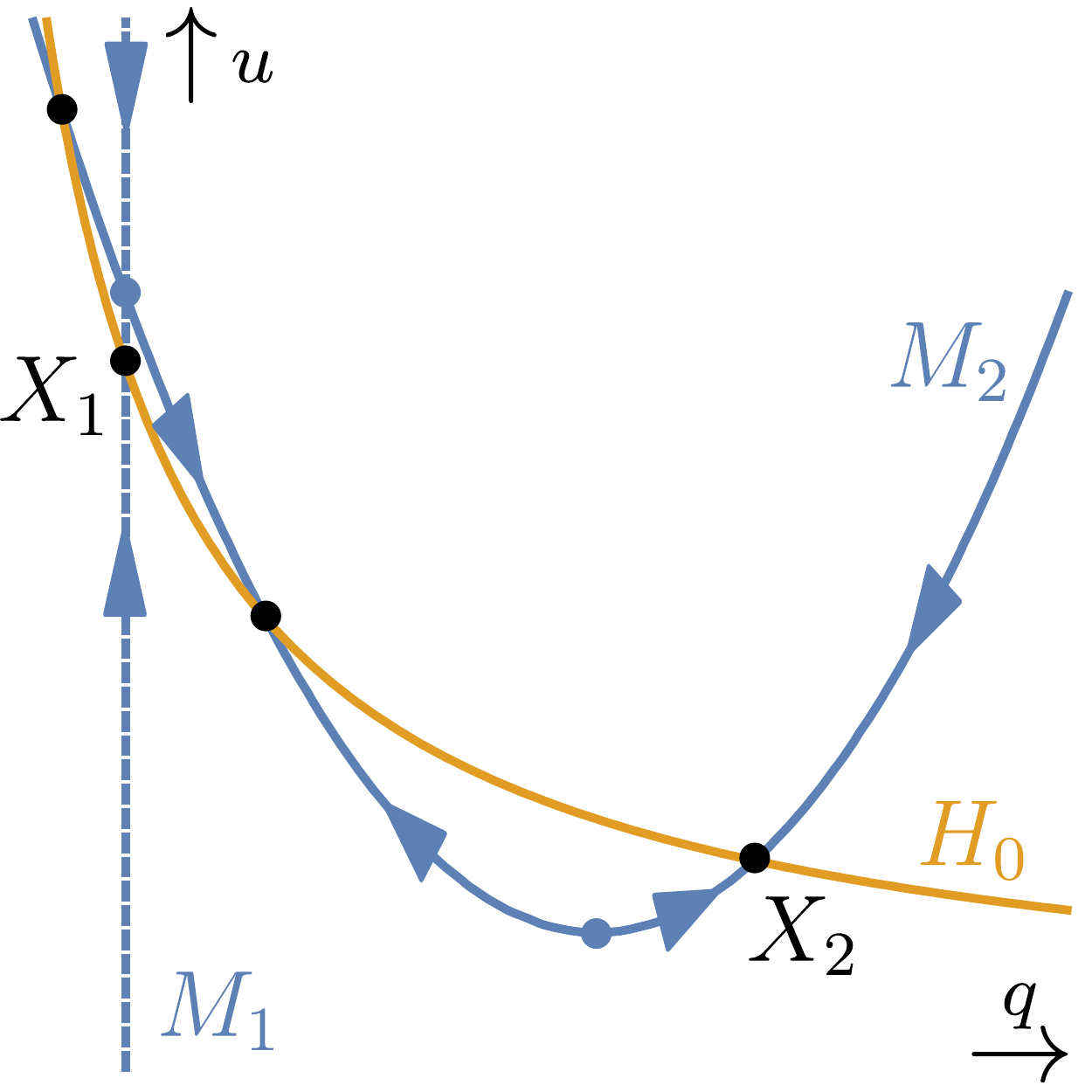}
\caption{Portrait of the dynamics of the slow subsystem~\eqref{SLOW} in the plane $p = 0$ for $r > \frac{2}{3}$ and $s < 2-r$. The dynamics is confined to the critical manifold $M_0 = M_1 \cup M_2$, which is the union of the line $M_1$ (dashed) and the parabola $M_2$ (solid). Equilibria arise at points where the nullclines $H_0$ and $M_0$ intersect. The relevant equilibria $X_1$ and $X_2$ are sinks. The blue dots on $M_0$ indicate fold points in which the flow of the slow subsystem is no longer defined. We stress that these fold points are irrelevant for our analysis as we do not consider the dynamics of~\eqref{SYS2} in their vicinity. We refer to~\cite{KuehnBook,MisRoz} for more background material and references on fold points.}
\label{fig:3}
\end{figure}

\begin{lemma} \label{lem:equilibria}
Let $D > 0$, $\mu,s \in \R$ and $\varepsilon \geq 0$. For each $r > \frac{2}{3}$ system~\eqref{SYS22} has an equilibrium $X_2 = (q_{\bb,+}(r),0,u_{\bb}(r))$, which resides on the right branch of the parabola $M_2$. We have $u_{\bb}(r) \in (\frac{6}{5},\frac{4}{3})$ with
\begin{align}\lim_{r \downarrow \frac{2}{3}} u_{\bb}(r) = \frac{4}{3}, \qquad  \lim_{r \to \infty} u_{\bb}(r) = \frac{6}{5}, \label{ubvalues}\end{align}
and
\begin{align} q_{\bb,+}(r) = 1 + \sqrt{\frac{r+u_{\bb}(r)-2}{r+0.1}} > 0. \label{qbbdef}\end{align}
\end{lemma}
\begin{proof}
The hyperbola $H_0$ has one, two or three intersection points with the parabola $M_2$, depending on the value of the Reynolds parameter $r > 0$. The $u$-values of these intersection points are readily seen to correspond to the roots of the cubic polynomial
\begin{align} K(u;r) = 40u^3 - \left(50 r + 169\right) u^2 + \left(160 r + 224\right) u - 120 r -96, \label{defcubic}\end{align}
in $u$. The discriminant of $K(u;r)$ is an upward-facing quartic polynomial in $r$ admitting two real roots, which are located at $r = -0.1$ and at $r = r_* := \frac{1}{50} (38 + 9 \sqrt[3]{3} - 9 \sqrt[3]{9}) > 0$. Thus, the transition from one to three intersection points of the hyperbola $H_0$ and the parabola $M_2$ takes place as the Reynolds parameter $r > 0$ crosses the critical value $r = r_*$. By evaluating the cubic $K(u;r)$ at $u = \frac{6}{5}$, $u = \frac{4}{3}$ and $u = 2$, these intersection points can be located with the aid of the intermediate value theorem. Thus, one finds that for each $r > 0$ there is an intersection point to the left of the line $M_1$. For $r \in (r_*,\frac{2}{3})$ there are three intersection points lying on the left branch of the parabola $M_2$, whereas for $r > \frac{2}{3}$ there is one intersection point $X_2 = (q_{\bb,+}(r),0,u_{\bb}(r))$ with $u_{\bb}(r) \in (\frac{6}{5},\frac{4}{3})$, which resides on the right branch of $M_2$ and satisfies~\eqref{ubvalues}, and two others located on the left branch of $M_2$, see also Figure~\ref{fig:3}. Since $u_{\bb}(r)$ is the smallest root of the cubic $K(u;r)$ for $r > r_*$ and it holds $f(q_{\bb,+}(r),u_{\bb}(r);r) = 0$, we arrive at~\eqref{qbbdef}.
\end{proof}

We require that the equilibria $X_1$ and $X_2$ are both sinks for the slow dynamics~\eqref{SLOW}. Thus, the flow of~\eqref{SLOW}, which is confined to the critical manifold $M_0$, must be directed towards the hyperbola $H_0$. Using that $X_2$ lies on the parabola $M_2$ with minimum $(1,0,2-r)$ and $u_{\bb}(r) \in (\frac{6}{5},\frac{4}{3})$ holds by Lemma~\ref{lem:equilibria}, one readily establishes for which $s$-values this is the case.

\begin{lemma} \label{lem:sinks}
Let $D > 0$, $r > \frac{2}{3}$ and $\mu, s \in \R$. The equilibria $X_1$ and $X_2$ in the slow subsystem~\eqref{SLOW} are sinks if and only if $s < u_{\bb}(r)$. A sufficient condition is $s < \max\{2-r,\frac{6}{5}\}$.
\end{lemma}

As already mentioned before, the dynamics in the slow subsystem~\eqref{SLOW} describe the regime in the pipe flow model~\eqref{SYS1} in between fast drops and rises of turbulence, where turbulence and centerline velocity are slaved to each other and evolve slowly. In particular, those orbits segments in~\eqref{SLOW} confined to the line $M_1$, which converge to the equilibrium $X_1$, represent slow recovery of the centerline velocity in~\eqref{SYS1} towards the laminar profile in the absence of turbulence. On the other hand, orbit segments on the right branch of the parabola $M_2$, converging to the equilibrium $X_2$, represent a slow decrease of the centerline velocity in the presence of turbulence up to the point where a new balance between turbulence and centerline velocity, away from the laminar profile, has been reached, see also Figure~\ref{fig:3}.

\subsection{Fast subsystem} \label{sec:fast}

Setting $\varepsilon = 0$ in~\eqref{SYS2} yields the \emph{fast subsystem} or \emph{layer problem}
\begin{align}
\begin{split}
q_\xi &= p,\\
p_\xi &= D^{-1}\left((u + \mu)p - f(q,u;r)\right),\\
u_\xi &= 0,
\end{split} \label{FAST}
\end{align}
in which the variable $u$ is constant, and thus can be regarded as a parameter. More specifically, for each fixed value of the centerline velocity $u$, one has a two-dimensional dynamical system describing the evolution of the turbulence $q$ in that `layer', see Figure~\ref{fig:4}. Thus, the fast subsystem~\eqref{FAST} captures fast transitions in turbulence level in the pipe flow model~\eqref{SYS1}, while the centerline velocity stays to leading order constant. In particular, the sharp interfaces of the traveling waves depicted in Figure~\ref{fig:1} correspond to heteroclinic connections between equilibria in the fast subsystem~\eqref{FAST}. One observes that the equilibria of~\eqref{FAST} are precisely those points on the critical manifold $M_0$.

We establish the relevant heteroclinic connections for the construction of the singular heteroclinic loop, which are the connections in the layers $u = 2$ and $u = u_{\bb}(r)$ in which the equilibria $X_1$ and $X_2$ for the full system~\eqref{SYS2} are located, cf.~Lemma~\ref{lem:equilibria}. In the layer $u = 2$ the fast subsystem~\eqref{FAST} has two additional equilibria residing on the left and right branch of the parabola $M_2$, which are given by $\left(q_{\f,\pm}(r),0,2\right)$ with
\begin{align}q_{\f,\pm}(r) = 1 \pm \sqrt{\frac{r}{r+0.1}}. \label{defqf}\end{align}
Moreover, in the layer $u = u_{\bb}(r)$ system~\eqref{FAST} admits the additional equilibrium $(0,0,u_{\bb}(r))$ on the line $M_1$ and the equilibrium
\begin{align}  \left(q_{\bb,-}(r),0,u_{\bb}(r)\right), \qquad q_{\bb,-}(r) =  1 - \sqrt{\frac{r+u_{\bb}(r)-2}{r+0.1}}, \label{defqb-}\end{align}
on the left branch of the parabola $M_2$. Thus, we wish to establish heteroclinic connections between $X_1$ and $(q_{\f,+}(r),0,2)$ and between $X_2$ and $(0,0,u_{\bb}(r))$. We proceed by rescaling~\eqref{FAST} so that it transforms to the well-known Nagumo (or real Ginzburg-Landau) equation for which the existence theory of heteroclinics is well-established. In particular, we obtain explicit expressions for the heteroclinic solutions, which are relevant for the Melnikov analysis in the upcoming~\S\ref{sec:Meln}. All in all, we arrive at the following result.

\begin{lemma} \label{lem:fastconnect}
Let $D > 0$, $r > \frac{2}{3}$ and $\mu,s \in \R$. Define $\phi \colon \R \to \R$ by
\begin{align} \phi(\chi) = \frac{1}{1+\mathrm{e}^{-\frac{1}{2}\sqrt{2} \chi}}. \label{defphi}\end{align}
If
\begin{align} 2 + \mu = \frac{1}{2}\sqrt{2D(r+0.1)} \left(q_{\f,+}(r) - 2 q_{\f,-}(r)\right), \label{cond:front}\end{align}
then the fast subsystem~\eqref{FAST} admits a heteroclinic front solution
\begin{align}  X_{\f}(\xi) = \left(q_{\f}(\xi;r),p_{\f}(\xi;r),u_{\f}(r)\right) = \left(q_{\f,+}(r)\,\phi\left(q_{\f,+}(r) \sqrt{\frac{r+0.1}{D}} \, \xi\right), q_{\f}'(\xi;r), 2\right), \label{frontsol}\end{align}
connecting the hyperbolic saddles $X_1$ to $(q_{\f,+}(r),0,2)$ within the layer $u = 2$. Moreover, if we have
\begin{align}u_{\bb}(r) + \mu = -\frac{1}{2} \sqrt{2D(r+0.1)}\left(q_{\bb,+}(r)  - 2q_{\bb,-}(r)\right), \label{cond:back} \end{align}
then~\eqref{FAST} possesses the heteroclinic back solution
\begin{align} X_{\bb}(\xi) = \left(q_{\bb}(\xi;r),p_{\bb}(\xi;r),u_{\bb}(r)\right) = \left(q_{\bb,+}(r)\,\phi\left(-q_{\bb,+}(r) \sqrt{\frac{r+0.1}{D}} \, \xi\right), q_{\bb}'(\xi;r), u_{\bb}(r)\right). \label{backsol}\end{align}
connecting the hyperbolic saddles $X_2$ to $(0,0,u_{\bb}(r))$ within the layer $u = u_{\bb}(r)$.
\end{lemma}
\begin{proof}
Our approach is to rescale the fast subsystem~\eqref{FAST} in the relevant layers, so that we arrive at the traveling-wave equation associated with the Nagumo (or real Ginzburg-Landau) equation
\begin{align}
q_t = q_{xx} + q(q-\alpha)(1-q), \label{nagumo}
\end{align}
with parameter $\alpha \in (0,1]$, cf.~\cite{KuehnBook1}. Thus, in the rescaled variables
\begin{align*} q = q_{\f,+}(r) \widetilde{q}, \qquad p = \left(q_{\f,+}(r)\right)^2 \sqrt{\frac{r+0.1}{D}} \, \widetilde{p}, \qquad \chi = q_{\f,+}(r) \sqrt{\frac{r+0.1}{D}} \, \xi,\end{align*}
the fast subsystem~\eqref{FAST} in the layer $u = 2$ reads
\begin{align} \label{nagumotw}
\begin{split}
\widetilde{q}_\chi &= \widetilde{p},\\
\widetilde{p}_\chi &= c_{\f}(r,\mu,D)\widetilde{p} - \widetilde{q}(\widetilde{q}-\alpha_{\f}(r))(1-\widetilde{q}),
\end{split}
\end{align}
with parameters
\begin{align*} c_{\f}(r,\mu,D) = \frac{2+\mu}{q_{\f,+}(r) \sqrt{D(r+0.1)}}, \qquad \alpha_{\f}(r) = \frac{q_{\f,-}(r)}{q_{\f,+}(r)} \in (0,1).\end{align*}
We note that~\eqref{nagumotw} indeed coincides with the traveling-wave equation associated with the Nagumo equation upon substituting the traveling-wave ansatz $q(x+ct)$ into~\eqref{nagumo}. Linearizing~\eqref{nagumotw} about its equilibria $(0,0)$ and $(1,0)$ we find that both are hyperbolic saddles, where we use $\alpha_{\f}(r) \in (0,1)$. We are looking for a heteroclinic connecting $(0,0)$ and $(1,0)$ in~\eqref{nagumotw}. Inserting the parabolic ansatz $\widetilde{p} = b\widetilde{q}(1-\widetilde{q})$ with $b \in \R$ into~\eqref{nagumotw}, one observes that such a heteroclinic exists if
\begin{align} c_{\f}(r,\mu,D) = \sqrt{2}\left(\frac{1}{2} - \alpha_{\f}(r)\right), \label{cond:front2} \end{align}
and is then explicitly given by (any translate of) $(\phi(\chi),\phi'(\chi))$. Thus, undoing the rescaling, we establish that the condition~\eqref{cond:front2} boils down to~\eqref{cond:front}, yielding the existence of the heteroclinic front solution~\eqref{frontsol} to the fast subsystem~\eqref{FAST} connecting the hyperbolic saddles $X_1$ and $(q_{\f,+}(r),0,2)$ in the layer $u = 2$. Using an analogous approach one finds the condition~\eqref{cond:back} for the existence of the heteroclinic back solution~\eqref{backsol} in the layer $u = u_{\bb}(r)$.
\end{proof}

\begin{figure}
\centering
\includegraphics[scale=0.5]{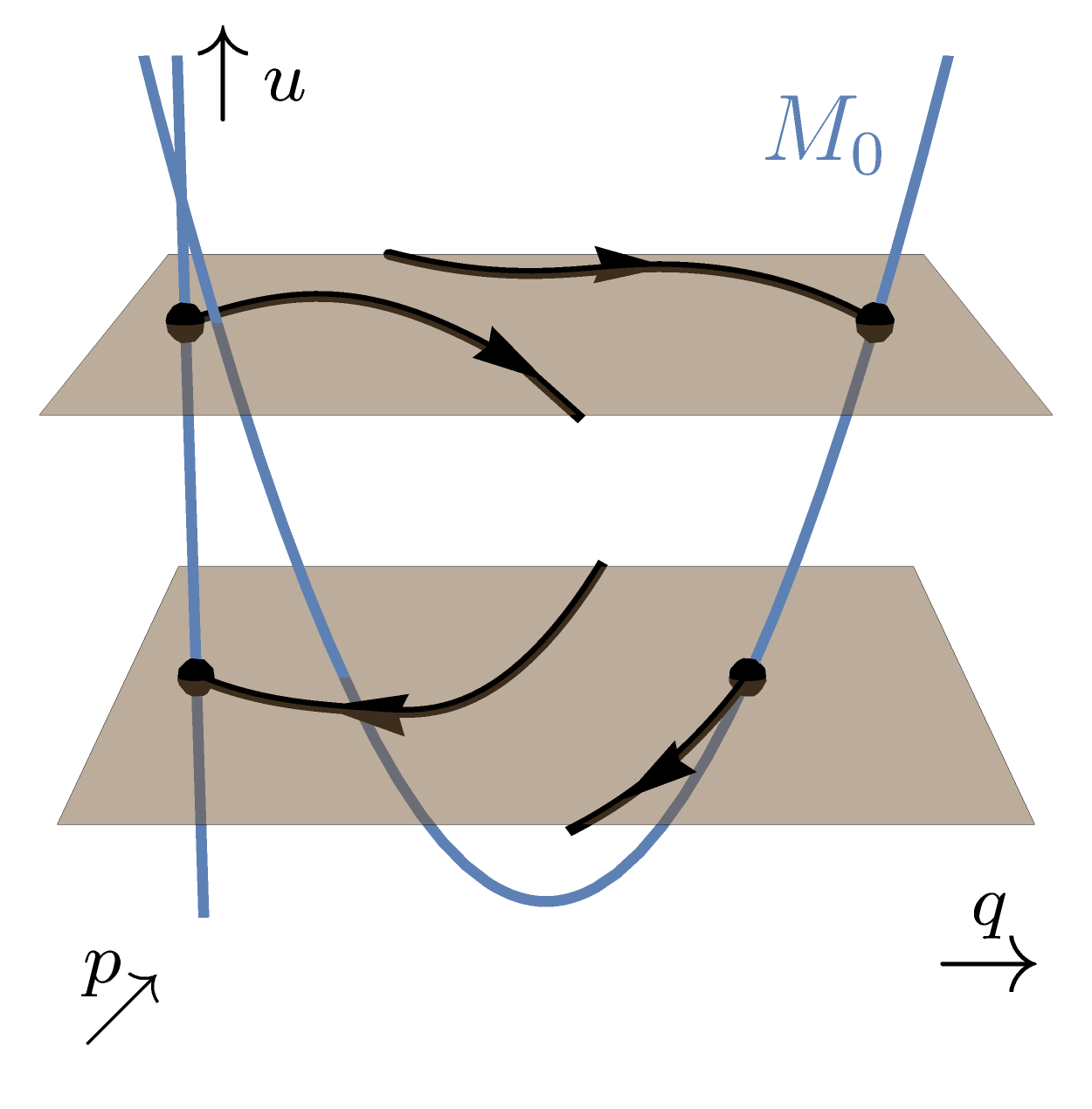} \hspace{1cm} \includegraphics[scale=0.5]{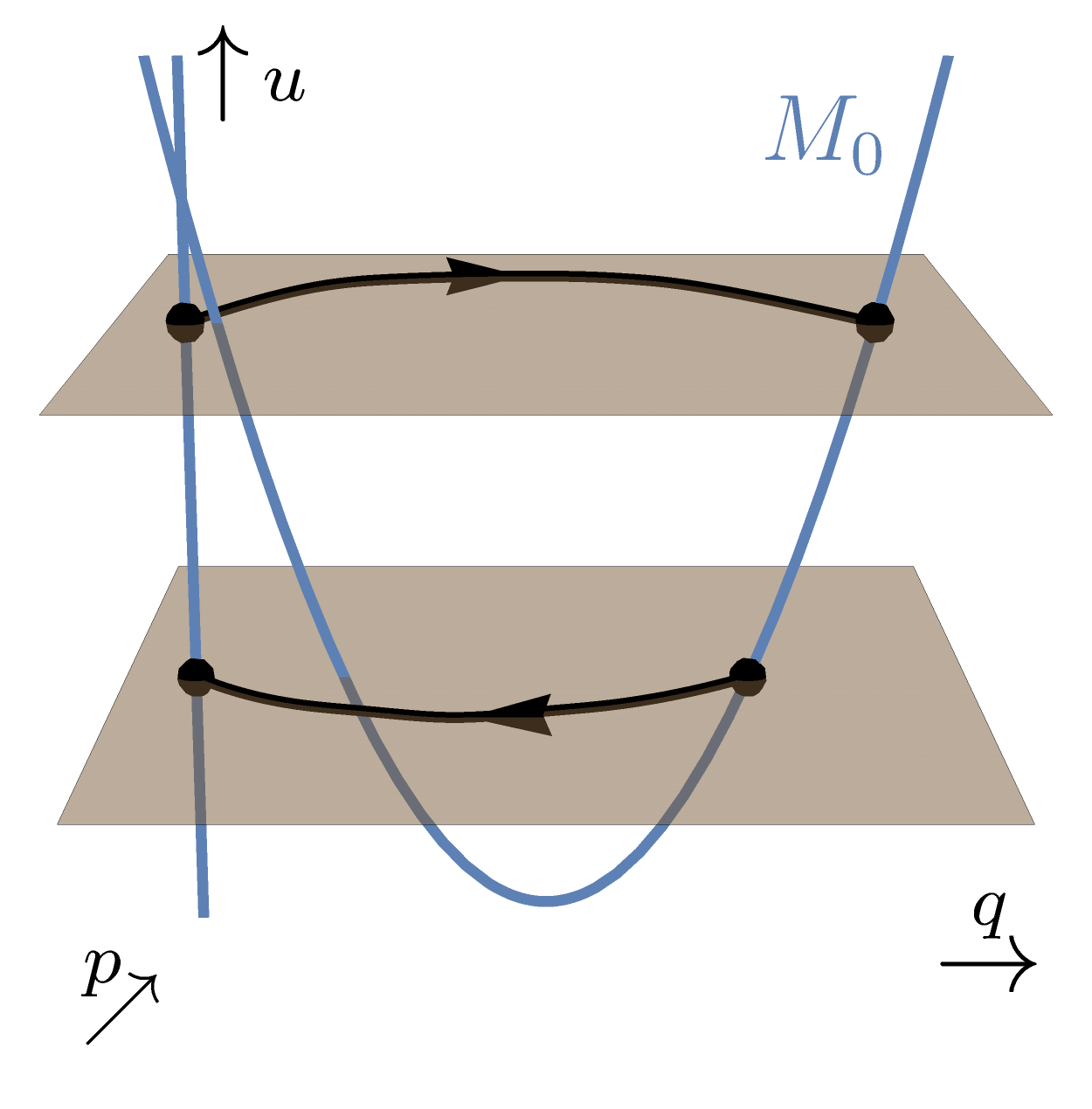}
\caption{Both panels depict the dynamics in two layers of the fast subsystem~\eqref{FAST}. Equilibria arise at intersections of the critical manifold $M_0$ with the layers. In each layer we depict an orbit converging to an equilibrium on the right branch of the parabola $M_2$ as $\xi \to \infty$ (upper layer) or $\xi \to -\infty$ (lower layer), and an orbit converging to an equilibrium on the line $M_1$ as $\xi \to -\infty$ (upper layer) or $\xi \to \infty$ (lower layer). In the right panel these two orbits coincide, so that both layers posses a heteroclinic connection, which is a necessary condition for the existence of a singular heteroclinic loop.}
\label{fig:4}
\end{figure}

\subsection{Construction of singular heteroclinic loop} \label{sec:singhet}

By piecing together orbit segments of the fast and slow subsystems we construct a so-called \emph{singular heteroclinic loop} connecting the equilibria $X_1$ and $X_2$. More specifically, the singular heteroclinic loop consists of four pieces. The first is the heteroclinic front connecting the equilibrium $X_1$ to $(q_{\f,+}(r),0,2)$ in the fast subsystem~\eqref{FAST}, which was established in Lemma~\ref{lem:fastconnect} and corresponds a sharp interface describing a quick rise in turbulence level from the laminar flow while the centerline velocity stays roughly constant. The second is the orbit segment in the slow subsystem~\eqref{SLOW} connecting the point $(q_{\f,+}(r),0,2)$ to the sink $X_2$ on the right-branch of the parabola $M_2$, which describes a gradual decrease of the centerline velocity in the presence of turbulence up to the point where a new balance between turbulence and center velocity has been reached. The third is the heteroclinic back connecting the equilibrium $X_2$ to $(0,0,u_{\bb}(r))$ in the fast subsystem~\eqref{FAST}, which corresponds to a sharp interface describing a quick drop to zero turbulence level while the centerline velocity is to leading order constant. Finally, the fourth is the orbit segment in the slow subsystem~\eqref{SLOW} connecting the point $(0,0,u_{\bb}(r))$ to the sink $X_1$ on the line $M_1$, which describes a gradual recovery of the centerline velocity towards the laminar profile in the absence of turbulence. We refer to Figure~\ref{fig:5} for a schematic depiction of the singular heteroclinic loop.

The existence of the singular heteroclinic loop is a direct consequence of Lemmas~\ref{lem:equilibria},~\ref{lem:sinks} and~\ref{lem:fastconnect}.

\begin{corollary} \label{cor:hetloop}
Fix a Reynolds parameter $r > \frac{2}{3}$ and a speed $s < \min\{u_{\bb}(r),-\mu_0(r)\}$. The functions $\mu_0 \colon (\frac{2}{3},\infty) \to \left(-\frac{8}{5}, \frac{1}{66} (3 \sqrt{115}-65)\right)$ and $D_0 \colon (\frac{2}{3},\infty) \to \left(0,\infty\right)$ given by
\begin{align}
\begin{split}
\mu_0(r) &= \frac{2\left(2q_{\bb,-}(r)-q_{\bb,+}(r)\right) + u_{\bb}(r) \left(2q_{\f,-}(r)-q_{\f,+}(r)\right)}{q_{\bb,+}(r)-2q_{\bb,-}(r)+q_{\f,+}(r)-2q_{\f,-}(r)},\\
D_0(r) &= \frac{2 \left(2-u_{\bb}(r)\right)^2}{(r+0.1) \left(2 q_{\bb,-}(r) - q_{\bb,+}(r) + 2 q_{\f,-}(r) - q_{\f,+}(r)\right)^2}.
\end{split} \label{defD0mu0}
\end{align}
satisfy~\eqref{limitsspeed} and~\eqref{limitsspeed2}. Moreover, the singular heteroclinic loop connecting the equilibria $X_1$ and $X_2$ exists for $D = D_0(r)$ and $\mu = \mu_0(r)$.
\end{corollary}
\begin{proof}
By Lemma~\ref{lem:sinks} the relevant orbits in the slow subsystem~\eqref{SLOW} exist. On the other hand, the conditions~\eqref{cond:front} and~\eqref{cond:back} for the existence of the relevant heteroclinic connections in the fast subsystem~\eqref{FAST} constitute an algebraic system of equations in the parameters $r,D$ and $\mu$, which can be uniquely solved for $D$ and $\mu$ yielding the solutions $D=D_0(r)$ and $\mu = \mu_0(r)$. Using $u_{\bb}(r) > 2-r$, one readily observes that
\begin{align*} q_{\f,+}(r) - 2q_{\f,-}(r) > 6\sqrt{\frac{5}{23}} - 1 > 1, \qquad q_{\bb,+}(r) - 2q_{\bb,-}(r) > -1,\end{align*}
for $r > \frac{2}{3}$. So, $D_0(r)$ is positive. Moreover, recalling  $u_{\bb}(r) \in (\frac{6}{5},\frac{4}{3})$ from Lemma~\ref{lem:equilibria}, we find $\mu_0(r) \in \left(-\frac{8}{5}, \frac{1}{66} (3 \sqrt{115}-65)\right)$. Finally,~\eqref{ubvalues} yields
\begin{align}
\lim_{r \to \infty} q_{j,+}(r) = 2, \qquad \lim_{r \to \infty} q_{j,-}(r) = 0, \qquad j = \f,\bb. \label{qlims}
\end{align}
Hence,~\eqref{limitsspeed} and~\eqref{limitsspeed2} follow with the aid of~\eqref{ubvalues} and~\eqref{qlims}.
\end{proof}

We emphasize that the singular heteroclinic loop is purely a geometric object and does not contain actual heteroclinic connections between the equilibria $X_1$ and $X_2$ in~\eqref{SYS2} (even for $\varepsilon = 0$). Indeed, such heteroclinic connections are smooth, whereas the singular heteroclinic loop has sharp edges at $(q_{\f,+}(r),0,2)$ and $(0,0,u_{\bb}(r))$, see Figure~\ref{fig:5}. However, we will prove in the next subsection that for $\varepsilon > 0$ sufficiently small an actual heteroclinic loop between $X_1$ and $X_2$ exists in~\eqref{SYS2} lying in the vicinity of the singular one.

\begin{remark}{ \upshape
In~\cite{Barkleyetal} one finds that the model parameters $D = 0.13$ and $r = 1.2$ in the pipe flow model~\eqref{SYS1} capture the regime where turbulence first begins to expand. It is interesting to note that $D_0(1.2) \approx 0.1286$, which yields that a singular heteroclinic loop exist for nearby parameter values upon selecting the wave speed $s = -\zeta - \mu_0(r)$. This indicates, in addition to the results in the current paper, that the existence of a heteroclinic loop in~\eqref{SYS2} might be intimately connected to the transition to fully turbulent flow.
}\end{remark}

\begin{figure}
\centering
\includegraphics[scale=0.5]{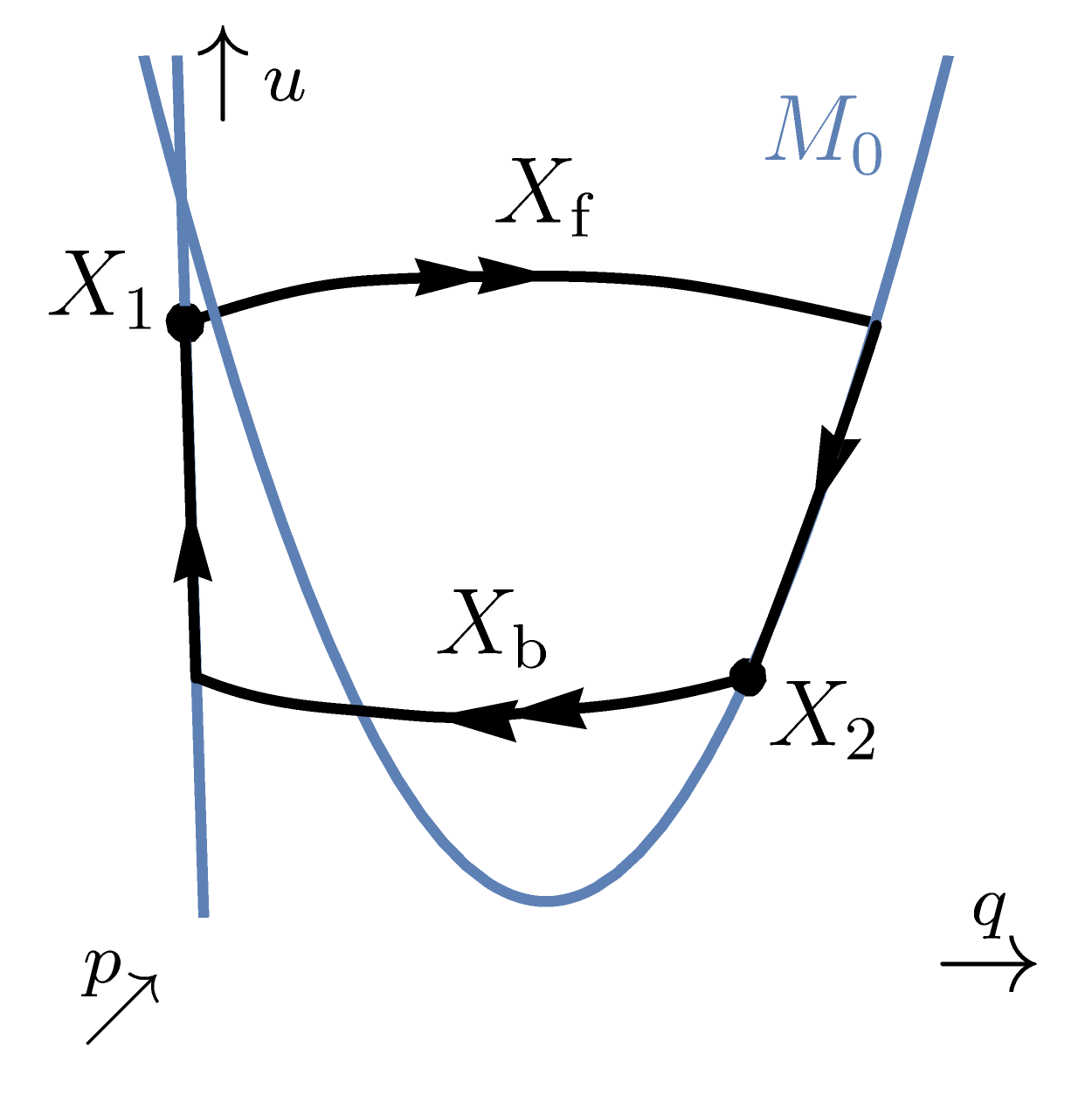} \hspace{1cm} \includegraphics[scale=0.5]{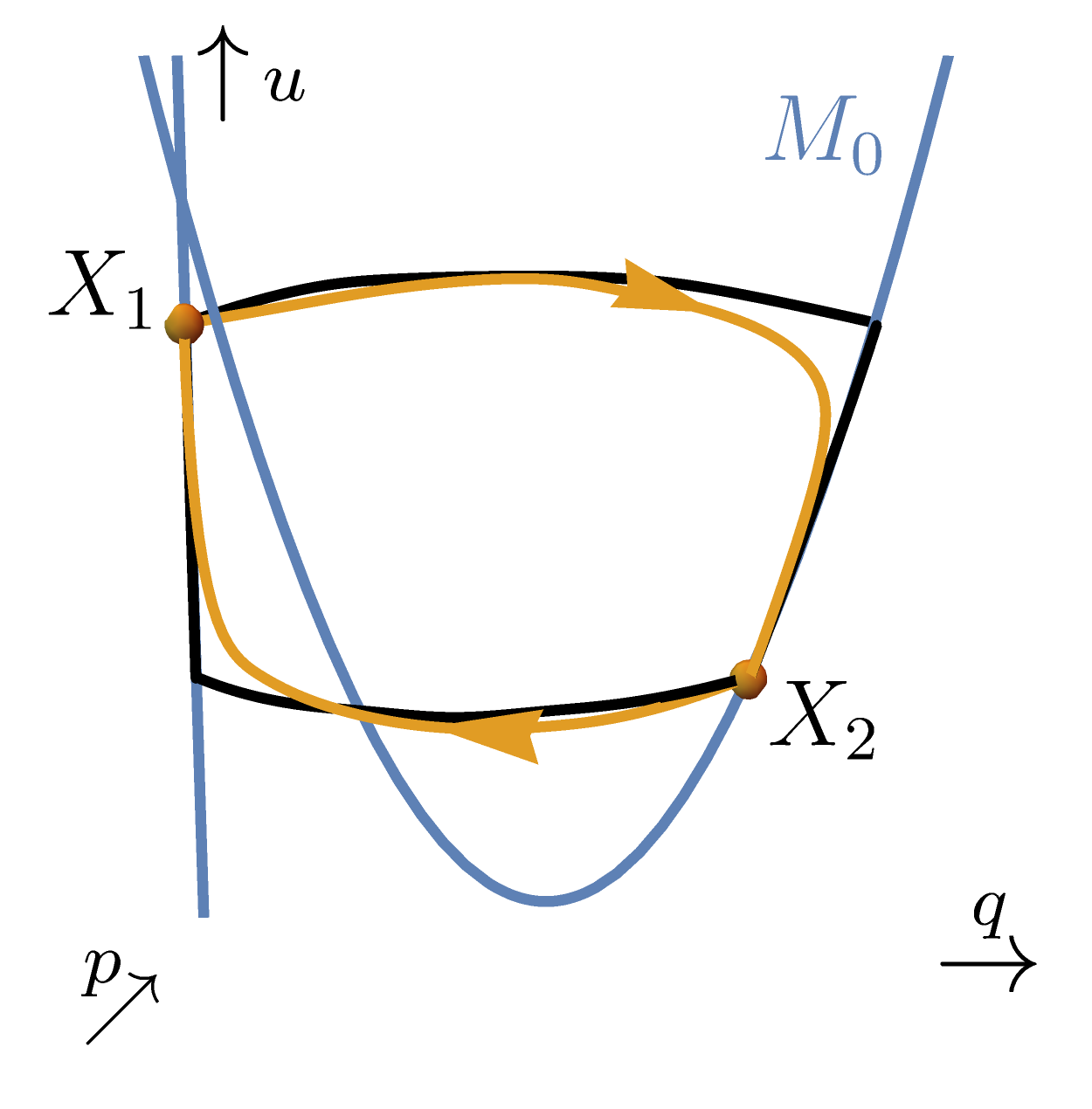}
\caption{The left panel depicts the singular heteroclinic loop consisting of the heteroclinic connections $X_{\f}$ and $X_{\bb}$ in the fast subsystem~\eqref{FAST}, which are established in Lemma~\ref{lem:fastconnect}, and the orbit segments in the slow subsystem~\eqref{SLOW} on the manifold $M_0$ connecting them. The right panel depicts an actual heteroclinic loop (orange) connecting the equilibria $X_1$ and $X_2$ in system~\eqref{SYS2} lying in the vicinity of the singular one (black).}
\label{fig:5}
\end{figure}

\subsection{Dynamics near the singular heteroclinic loop} \label{sec:persistence_manifolds}

We fix a Reynolds parameter $r > \frac{2}{3}$ and a speed $s < \min\{u_{\bb}(r),-\mu_0(r)\}$. By Corollary~\ref{cor:hetloop} there exists a singular heteroclinic loop at parameter values $\pmb \alpha = \pmb \alpha_0$, where we use the short-hand notation
$${\pmb \alpha} = (D,\mu,\varepsilon), \qquad {\pmb \alpha}_0 = {\pmb \alpha}_0(r) = (D_0(r),\mu_0(r),0).$$

In the remaining part of this section, we will prove that an actual heteroclinic loop between the equilibria $X_1$ and $X_2$ exists in the vicinity of the singular one for parameter values $\varepsilon > 0$ and $\pmb \alpha$ close to $\pmb \alpha_0$, see Figure~\ref{fig:5}. This requires knowledge about the dynamics in~\eqref{SYS2} near the singular heteroclinic loop. Of particular interest are the so-called stable and unstable manifolds of the equilibria $X_1$ and $X_2$. The \emph{stable manifold} $W_{\pmb \alpha}^{\s}(X_i)$ is the union of all orbits in~\eqref{SYS2} converging to $X_i$ as $\xi \to \infty$, whereas the \emph{unstable manifold} $W_{\pmb \alpha}^{\uu}(X_i)$ is the union of all orbits in~\eqref{SYS2} converging to $X_i$ as $\xi \to -\infty$. Thus, orbits in $W_{\pmb \alpha}^{\s}(X_1)$ and $W_{\pmb \alpha}^{\uu}(X_1)$ correspond to traveling waves in the pipe flow model~\eqref{SYS1}, whose profiles connect to the laminar state as $\xi \to +\infty$ and $\xi \to -\infty$, respectively. Similarly, orbits in $W_{\pmb \alpha}^{\s}(X_2)$ and $W_{\pmb \alpha}^{\uu}(X_2)$ relate to traveling waves, whose profiles connect to the turbulent steady state as $\xi \to +\infty$ and $\xi \to -\infty$, respectively.

Clearly, heteroclinic fronts connecting $X_1$ with $X_2$ must lie in $W_{\pmb \alpha}^{\uu}(X_1) \cap W_{\pmb \alpha}^{\s}(X_2)$ and heteroclinic backs connecting $X_2$ with $X_1$ lie in $W_{\pmb \alpha}^{\s}(X_1) \cap W_{\pmb \alpha}^{\uu}(X_2)$. Thus, establishing a heteroclinic loop in~\eqref{SYS2} boils down to identifying parameter values $\pmb \alpha$ for which the intersections $W_{\pmb \alpha}^{\uu}(X_1) \cap W_{\pmb \alpha}^{\s}(X_2)$ and $W_{\pmb \alpha}^{\s}(X_1) \cap W_{\pmb \alpha}^{\uu}(X_2)$ are both nonempty.

To understand how such intersections behave under perturbations, the dimension of the stable and unstable manifolds $W_{\pmb \alpha}^{\s/\uu}(X_i), \, i = 1,2$ as geometric objects is of interest. It is a basic result from dynamical systems theory that the dimension is determined by the eigenvalues of the linearization of~\eqref{SYS2} about the equilibrium $X_i$. In the following lemma, we establish that, for $\pmb \alpha$ close to $\pmb \alpha_0$, the linearization of~\eqref{SYS2} about $X_i$ possesses two real eigenvalues of opposite sign, which are bounded away from the imaginary axis, and one real negative eigenvalue converging to $0$ as $\varepsilon \downarrow 0$.

\begin{lemma}[The relative expansion of $X_i$] \label{lem:relative_expansion}
The equilibria $X_1$ and $X_2$ in~\eqref{SYS2} are hyperbolic and relatively expansive in the sense that the eigenvalues of the linearization of~\eqref{SYS2} about $X_i$ satisfy
\begin{equation*}
\lambda_1(X_i) < \lambda_2(X_i) < 0 < \lambda_3(X_i),  \ \text{ and } \lambda_2(X_i) + \lambda_3(X_i) > 0, \qquad i = 1,2,
\end{equation*}
when $\varepsilon >0$ is taken sufficiently small.
\end{lemma}
\begin{proof}
Taking the limit $\varepsilon \downarrow 0$ in system~\eqref{SYS2} we arrive at the fast subsystem~\eqref{FAST}. By Lemma~\ref{lem:fastconnect} the equilibria $X_1$ and $X_2$ are hyperbolic saddles in~\eqref{FAST} in the layers $u = 2$ and $u=u_{\bb}(r)$, respectively. Hence, for $\varepsilon>0$ sufficiently small, the linearization of~\eqref{SYS2} about $X_i$ possesses two real eigenvalues $\lambda_1(X_i)$ and $\lambda_3(X_i)$ of opposite sign, which are bounded away from $0$ as $\varepsilon \downarrow 0$, and one real eigenvalue $\lambda_2(X_i)$ converging to $0$ as $\varepsilon \downarrow 0$.

Recall that upon introducing the stretched spatial coordinate $\tau = \varepsilon \xi$ system~\eqref{SYS2} transforms into~\eqref{SYS22}. Thus, the eigenvalues of the linearizations of~\eqref{SYS2} and~\eqref{SYS22} about $X_i$ are also related through rescaling by a factor $\varepsilon$. Taking the limit $\varepsilon \downarrow 0$ in system~\eqref{SYS22} we arrive at the slow subsystem~\eqref{SLOW}, for which the equilibrium $X_i$ is a sink by Lemma~\ref{lem:sinks}. Thus, provided $\varepsilon > 0$ is sufficiently small, the linearization of~\eqref{SYS22} about $X_i$ possesses a negative eigenvalue, which stays bounded as $\varepsilon \downarrow 0$. Hence, this eigenvalue must be $\varepsilon^{-1} \lambda_2(X_i)$. This proves the claim.
\end{proof}

Thus, the stable manifold $W_{\pmb \alpha}^{\s}(X_i)$ is two-dimensional and the unstable manifold $W_{\pmb \alpha}^{\uu}(X_i)$ is one-dimensional.
The eigenvalues $\lambda_2(X_i)$ and $\lambda_3(X_i)$ are called the \emph{principal stable and unstable eigenvalues} with corresponding \emph{principal stable and unstable eigenvectors} $e_2(X_i)$ and $e_3(X_i)$, respectively.

We aim to establish a heteroclinic loop by perturbing off the singular one. As outlined before, the singular heteroclinic loop, which arises at $\pmb \alpha = \pmb \alpha_0$, is not an actual heteroclinic loop connecting the equilibria $X_1$ and $X_2$ in~\eqref{SYS2}. Indeed, setting $\varepsilon = 0$ in~\eqref{SYS2} yields the fast subsystem~\eqref{FAST} in which the dynamics is layered, so that no heteroclinic connections between $X_1$ and $X_2$ can exist. Although this seems a serious obstruction to obtain a heteroclinic loop by perturbing off the singular one, it is still possible using \emph{geometric singular perturbation theory (GSPT)}, cf.~\cite{FEN2,Kaper,KuehnBook,Jones}. The crucial observation is that the singular heteroclinic loop is an actual heteroclinic loop connecting the \emph{sets} $M_1$ and $M_2$. Indeed, at $\pmb \alpha = \pmb \alpha_0$, there exist, by Corollary~\ref{cor:hetloop}, forward and backward heteroclinic connections between the line $M_1$ and the (right branch of) the parabola $M_2$, see Figure~\ref{fig:6}.

Hence, it makes sense to define stable and unstable manifolds associated with the relevant segments in $M_1$ and $M_2$. Thus, we set
\begin{align*} Z_{1}(u) = (0,0,u), \qquad Z_{2}(u) = \left(1 + \sqrt{\frac{r+u-2}{r+0.1}},0,u\right),\end{align*}
and take compact subsets
$$K_{1,0} = \{Z_{1}(u) : u \in U_1\} \subset M_1, \qquad K_{2,0} = \{Z_{2}(u) : u \in U_2\} \subset M_2,$$
with $U_1, U_2 \subset \R$ such that the orbit segments of the singular heteroclinic loop on the line $M_1$ and on the right branch of the parabola $M_2$ are strictly contained in $K_{1,0}$ and $K_{2,0}$, respectively. The \emph{stable manifold $W_{\pmb \alpha_0}^{\s}(K_{i,0})$ of the set} $K_{i,0}$ consists of all orbits (locally) converging to $K_{1,0}$ as $\xi \to \infty$, whereas the \emph{unstable manifold $W_{\pmb \alpha_0}^{\uu}(K_{i,0})$ of the set} $K_{i,0}$ consists of all orbits (locally) converging to $K_{i,0}$ as $\xi \to -\infty$. Clearly, the existence of the singular heteroclinic loop, cf.~Corollary~\ref{cor:hetloop}, implies that the unstable manifold $W_{\pmb \alpha_0}^{\uu}(K_{2,0})$ intersects the stable manifold $W_{\pmb \alpha_0}^{\s}(K_{1,0})$ in system~\eqref{FAST} along the heteroclinic back $X_{\bb}(\xi)$, and the unstable manifold $W_{\pmb \alpha_0}^{\uu}(K_{1,0})$ intersects $W_{\pmb \alpha_0}^{\s}(K_{2,0})$ along the heteroclinic front $X_{\f}(\xi)$. We prove in the upcoming sections that these intersections are transversal, i.e.~the singular heteroclinic loop is \emph{nondegenerate}, which is an important prerequisite for applying the results of Deng~\cite{Deng91a}.

\begin{figure}
\centering
\includegraphics[scale=0.5]{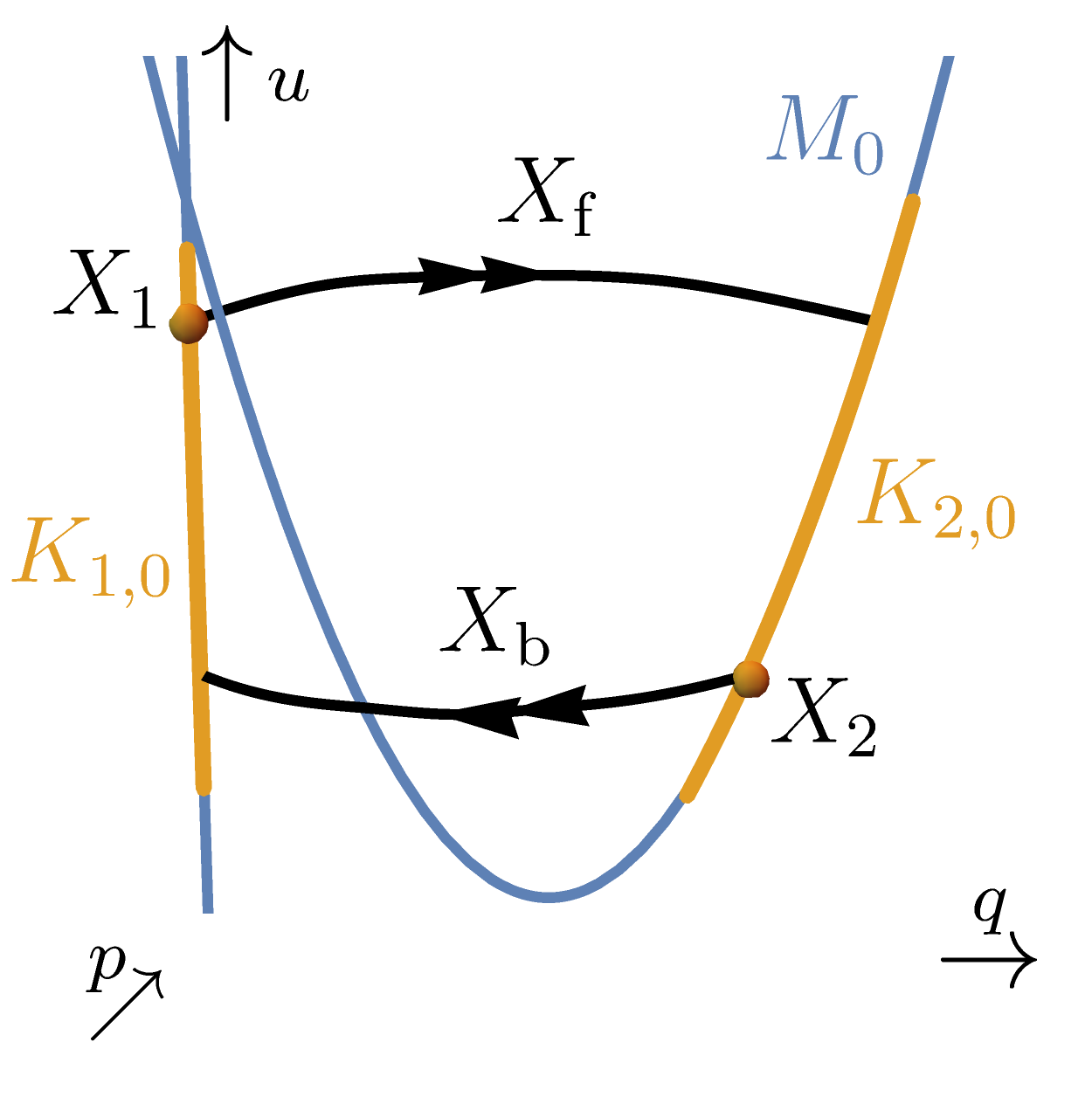}
\caption{Depicted are the compact sets $K_{1,0}$ and $K_{2,0}$ of equilibria in the fast subsystem~\eqref{FAST}, which lie on the critical manifold $M_0$, and the heteroclinic connections $X_{\f}$ and $X_{\bb}$, established in Lemma~\ref{lem:fastconnect}, connecting them.}
\label{fig:6}
\end{figure}

We now employ GSPT to describe the sets $K_{1,0}$ and $K_{2,0}$ and the associated stable and unstable manifolds for $\pmb \alpha$ close to $\pmb \alpha_0$. Linearizing the fast subsystem~\eqref{FAST} about each equilibrium point in $K_{i,0}$ for $i = 1,2$, yields two nonzero eigenvalues so that it is a hyperbolic saddle in the associated layers of the fast subsystem~\eqref{FAST}, i.e.~the compact manifold $K_{i,0}$ is \emph{normally hyperbolic} --- the calculation is similar to the one in the proof of Lemma~\ref{lem:fastconnect}.
GSPT now implies that the manifold $K_{i,0}$ persists as a (locally) invariant one-dimensional manifold $K_{i,{\pmb \alpha}}$ in system~\eqref{SYS2}, which depends smoothly on ${\pmb \alpha}$, provided ${\pmb \alpha}$ is close to ${\pmb \alpha_0}$.

The equilibrium $X_i \in K_{i,0}$ is a global attractor for the dynamics of the slow subsystem~\eqref{SLOW} restricted to $K_{i,0}$ for $i = 1,2$, cf.~Lemma~\ref{lem:sinks}. Since $X_i$ is also an equilibrium of system~\eqref{SYS2} for ${\pmb \alpha} \neq {\pmb \alpha}_0$, it must serve as a global attractor for the dynamics of~\eqref{SYS2} restricted to the invariant manifold $K_{i,{\pmb \alpha}}$, too. Therefore, GSPT implies that the stable manifold $W^{\s}_{\pmb \alpha}(X_i)$ of the equilibrium $X_i$ in system~\eqref{SYS2} coincides for $\varepsilon > 0$ with the stable manifold $W^{\s}_{\pmb \alpha}(K_{i,\pmb \alpha})$ of the invariant manifold $K_{i,\pmb \alpha}$ for $i = 1,2$.

The (un)stable manifold $W^{\uu/\s}_{\pmb \alpha}(K_{i,\pmb \alpha})$ of the invariant manifold $K_{i,\pmb \alpha}$ in~\eqref{SYS2} depends smoothly on ${\pmb \alpha}$ for ${\pmb \alpha}$ close to ${\pmb \alpha}_0$, and is at ${\pmb \alpha} = {\pmb \alpha_0}$ given by the two-dimensional union of (un)stable fibers
\begin{align} W_{\pmb \alpha_0}^{\uu/\s}(K_{i,0}) = \bigcup_{u \in U_i} W_{\pmb \alpha_0}^{\uu/\s}(Z_{i}(u)), \label{fibering}\end{align}
where $W_{\pmb \alpha_0}^{\uu/\s}(Z_{i}(u))$ is the one-dimensional (un)stable manifold of the equilibrium $Z_i(u) \in K_{i,0}$ of the fast subsystem~\eqref{FAST}. Clearly, this implies that $W^{\uu/\s}_{\pmb \alpha}(K_{i,\pmb \alpha})$ is a two-dimensional geometric object.

\subsection{Melnikov analysis} \label{sec:Meln}

As outlined in the previous subsection, our approach is to establish a heteroclinic loop in~\eqref{SYS2} by identifying parameter values $\pmb \alpha$ close to $\pmb \alpha_0$ for which intersections between $W_{\pmb \alpha}^{\uu}(X_1)$ and $W_{\pmb \alpha}^{\s}(X_2)$ and between $W_{\pmb \alpha}^{\s}(X_1)$ and $W_{\pmb \alpha}^{\uu}(X_2)$ exist, where we exploit that $W_{\pmb \alpha}^{\s}(X_i), \, i = 1,2$ coincides with the stable manifold $W_{\pmb \alpha}^{\s}(K_{i,\pmb \alpha})$ of the locally invariant set $K_{i,\pmb \alpha}$ for $\varepsilon > 0$.

To locate such intersections we employ Melnikov's method, cf.~\cite{Holmes,Melnikov,PAL,Szmolyan1}. Let us introduce the necessary mathematical framework, see also Figure~\ref{fig:7}. Thus, let $\Sigma$ be a plane perpendicular to the heteroclinic front $X_{\f}(\xi)$ at $\xi = 0$. Then, the one-dimensional unstable manifold $W^\uu_{\pmb \alpha_0}(X_1)$ intersects $\Sigma$ transversely at $X_{\f}(0)$. Hence, by smooth dependency on parameters, there exists a unique intersection point $X_{\pmb \alpha}^\uu$ between $W^\uu_{\pmb \alpha}(X_1)$ and $\Sigma$  for ${\pmb \alpha}$ close to ${\pmb \alpha}_0$ satisfying $X_{\pmb \alpha_0}^\uu = X_{\f}(0)$. On the other hand, by~\eqref{fibering} the intersection of the two-dimensional stable manifold $W_{\pmb \alpha_0}^{\s}(K_{2,0})$ with $\Sigma$ is a curve through the point $X_{\f}(0)$ parameterized by $u$. Consequently, the vector $e_{\f} = (-p_{\f}'(0;r),q_{\f}'(0;r),0) \in \Sigma$, which is perpendicular to the tangent vector $X_{\f}'(0)$, is transverse to the tangent vector of the curve $\Sigma \cap W_{\pmb \alpha_0}^{\s}(K_{2,0})$ at $X_{\f}(0)$. Moreover, by smooth dependency on parameters, $\Sigma \cap W_{\pmb \alpha}^{\s}(K_{2,\pmb \alpha})$ is also a one-dimensional curve depending smoothly on ${\pmb \alpha}$ for ${\pmb \alpha}$ close to ${\pmb \alpha}_0$. So, the line $\ell_{\pmb \alpha}\subset \Sigma$ through $X_{\pmb \alpha}^\uu$ parallel to $e_{\f}$ intersects the curve $\Sigma \cap W_{\pmb \alpha}^{\s}(K_{2,0})$ in a unique point $X_{\pmb \alpha}^\s$, provided ${\pmb \alpha}$ is close to ${\pmb \alpha_0}$, with $X_{\pmb \alpha_0}^\s = X_{\f}(0)$. We have
\begin{align*}
X_{\pmb \alpha}^\uu - X_{\pmb \alpha}^\s = Q_{\f}({\pmb \alpha};r)e_{\f},
\end{align*}
for some smooth function $Q_{\f} \colon \mathcal U \times (\frac{2}{3},\infty) \to \R$, where $\mathcal U \subset \R^3$ is a small neighborhood of ${\pmb \alpha_0}$.

\begin{figure}
\centering
\includegraphics[scale=0.5]{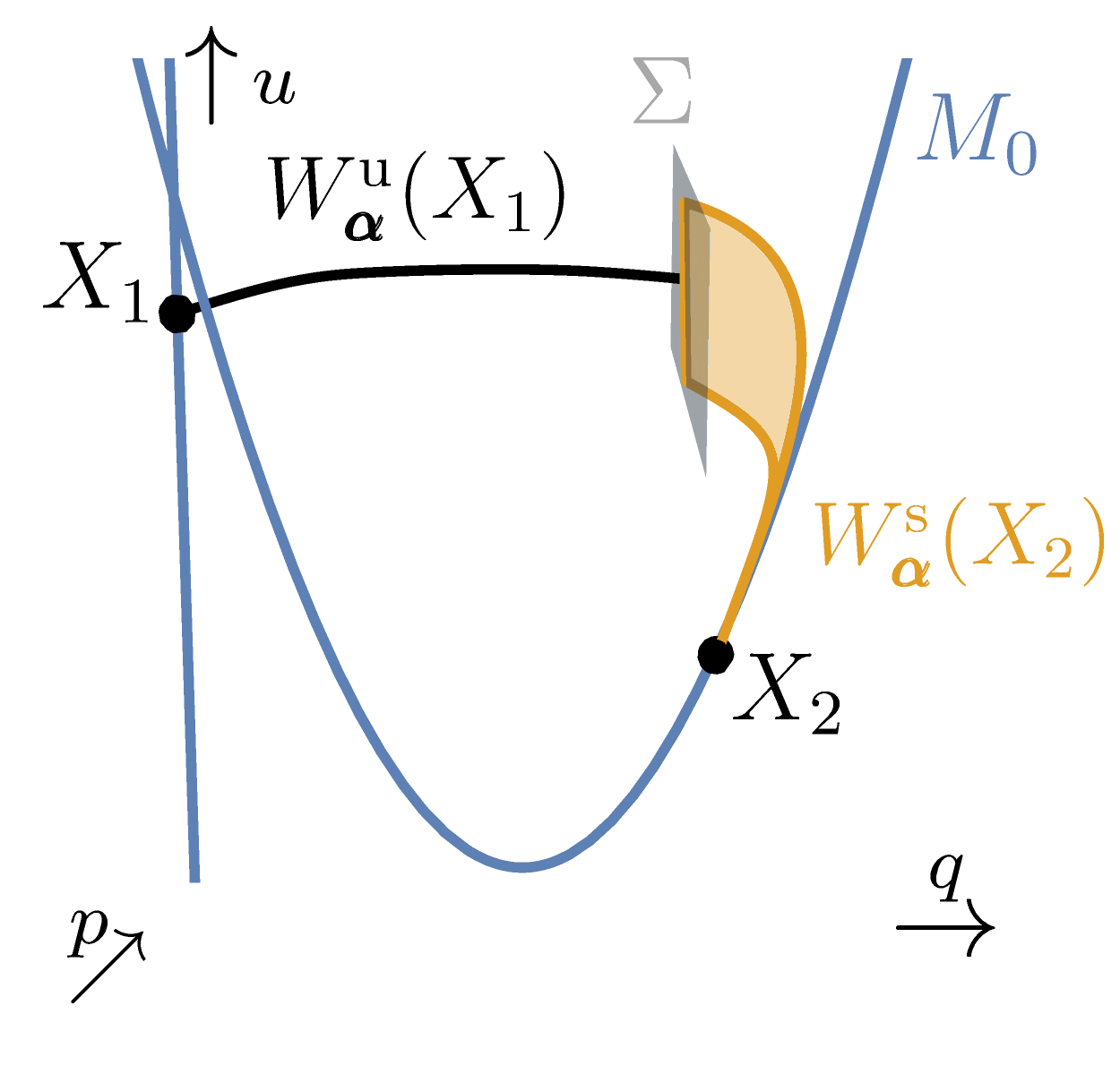} \hspace{1cm} \includegraphics[scale=0.5]{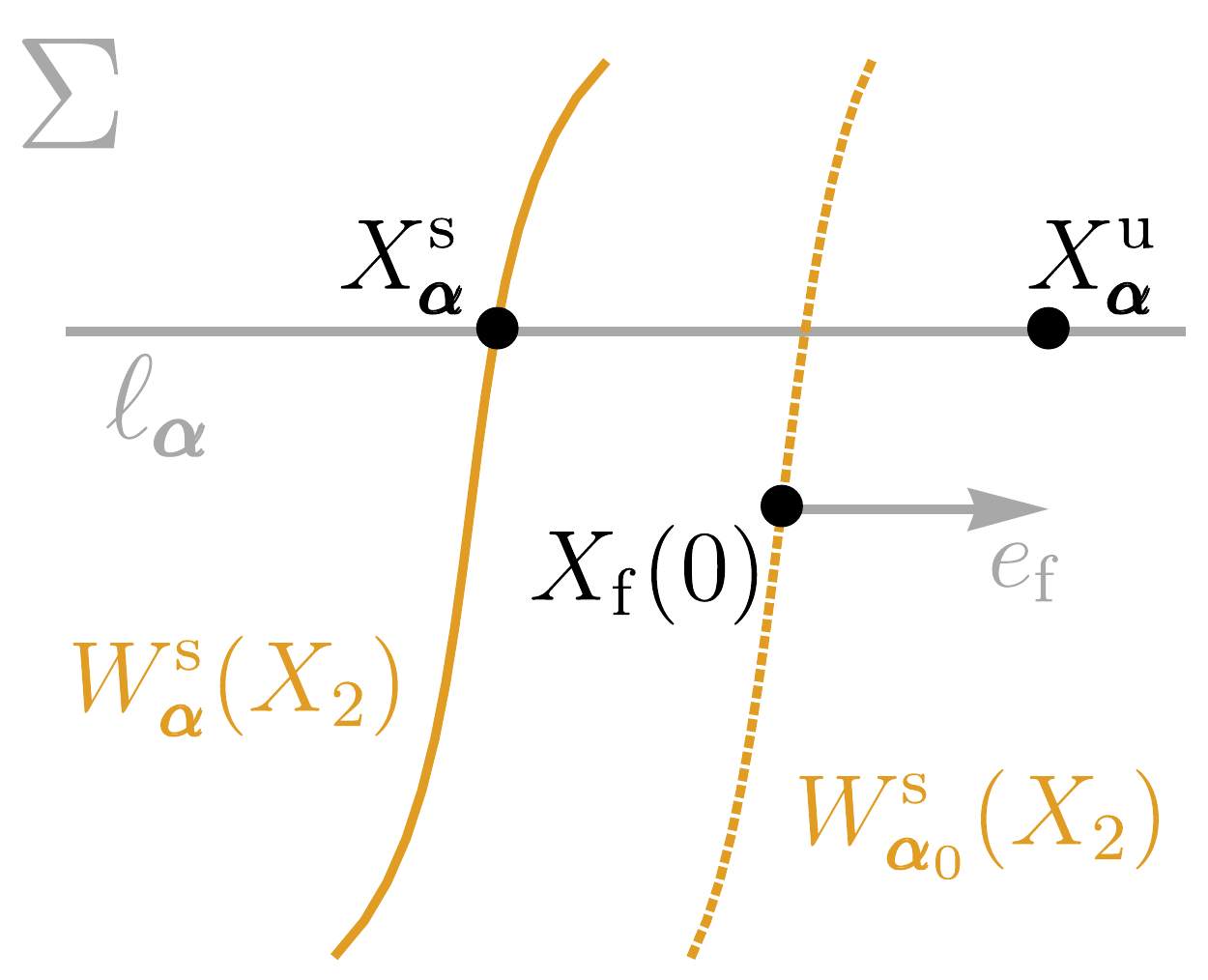}
\caption{The left panel shows the equilibria $X_1$ and $X_2$ on the critical manifold $M_0$ in the dynamical system~\eqref{SYS2}. A top view of the section $\Sigma$, which is perpendicular to the heteroclinic front $X_{\f}(\xi)$ at $\xi = 0$, is depicted in the right panel. The one-dimensional unstable manifold $W^\uu_{\pmb \alpha}(X_1)$ intersects $\Sigma$ transversally in a unique point $X_{\pmb \alpha}^\uu$ satisfying $X_{\pmb \alpha_0}^\uu = X_{\f}(0)$. Moreover, the two-dimensional stable manifold $W_{\pmb \alpha}^{\s}(X_2)$ intersects $\Sigma$ in a curve (orange). The vector $e_{\f}$ is transversal to the curve $W_{\pmb \alpha_0}^{\s}(X_2) \cap \Sigma$ (dashed) at the point $X_{\f}(0)$. The line $\ell_{\pmb \alpha}$ through $X_{\pmb \alpha}^{\uu}$ is parallel to $e_{f}$ and intersects $W_{\pmb \alpha}^{\s}(X_2) \cap \Sigma$ in the point $X_{\pmb \alpha}^{\s}$, where $X_{\pmb \alpha_0}^\s = X_{\pmb \alpha_0}^\uu = X_{\f}(0)$. The Melnikov function $Q_{\f}(\pmb \alpha)$ measures the distance between the points $X_{\pmb \alpha}^\s$ and $X_{\pmb \alpha}^\uu$.}
\label{fig:7}
\end{figure}
The roots of this so-called \emph{Melnikov function} $Q_{\f}(\cdot;r)$ coincide with the parameter values $\pmb \alpha$ for which an intersection of the unstable manifold $W^\uu_{\pmb \alpha}(X_1)$ with the stable manifold $W_{\pmb \alpha}^\s(K_{2,{\pmb \alpha}})$ exists. Since we have $W_{\pmb \alpha}^\s(K_{2,{\pmb \alpha}}) = W^{\s}_{\pmb \alpha}(X_2)$ for $\varepsilon > 0$, such an intersection for $\varepsilon > 0$ thus yields a heteroclinic front in~\eqref{SYS2} connecting the equilibrium $X_1$ to $X_2$. Analogously, one constructs a Melnikov function $Q_{\bb}(\pmb \alpha;r)$, whose roots for $\varepsilon > 0$ correspond to the parameter values $\pmb \alpha$ for which an intersection of the unstable manifold $W^\uu_{\pmb \alpha}(X_2)$ with the stable manifold $W_{\pmb \alpha}^\s(K_{1,{\pmb \alpha}}) = W^{\s}_{\pmb \alpha}(X_1)$ exists, yielding a heteroclinic back in~\eqref{SYS2} connecting $X_2$ to $X_1$.

Hence, to establish a heteroclinic loop one must solve $Q_{\f}(\pmb \alpha;r) = 0 = Q_{\bb}(\pmb \alpha;r)$ with $\pmb \alpha = (D,\mu,\varepsilon)$ and $\varepsilon > 0$. Our approach in~\S\ref{sec:hetloop} will be to apply the implicit function theorem and show that these two equations can be solved for the parameters $D$ and $\mu$ for ${\pmb \alpha}$ close to $\pmb \alpha_0 = \pmb \alpha_0(r)$ yielding functions $D(\varepsilon;r)$ and $\mu(\varepsilon;r)$ with $D(0;r) = D_0(r)$ and $\mu(0;r) = \mu_0(r)$. To verify the conditions for the implicit function theorem, one needs to compute the derivatives $\partial_D Q_{\f/\bb}(\pmb \alpha_0(r);r)$ and $\partial_\mu Q_{\f/\bb}(\pmb \alpha_0(r);r)$. Moreover, we later need the signs of $\partial_D Q_{\f/\bb}(\pmb \alpha_0(r);r)$ and $\partial_\mu Q_{\f/\bb}(\pmb \alpha_0(r);r)$ to determine whether the heteroclinic loop is single or double twisted, as elaborated in the proofs of Theorems~\ref{thm:single_twist} and~\ref{thm:double_twist} in~\S\ref{sec:proof_twists}.

We write the dynamical system~\eqref{SYS2} in the abstract form
\begin{align*} \partial_\xi X = F(X;\pmb \alpha,r),\end{align*}
with $F \colon \R^3 \times \mathcal U \times (\frac{2}{3},\infty) \to \R^3$. Following~\cite{Holmes,Melnikov,PAL}, one finds that the derivatives of $Q_j(\pmb \alpha;r)$ are explicitly given by the so-called \emph{Melnikov integrals}
\begin{align} \partial_{i} Q_{j}(\pmb \alpha_0(r);r) = -\int_\R \Psi_{j}(\xi;r) \cdot \partial_{i} F(X_{j}(\xi);\pmb \alpha_0(r),r) \de \xi, \qquad i = D,\mu, \, j = \f,\bb,\label{melnint}\end{align}
where $\Psi_{j}(\xi;r)$ is the solution to the adjoint variational equation
\begin{align*} \partial_\xi \Psi = -\left(\partial_X F(X_{j}(\xi);\pmb \alpha_0(r),r)\right)^\top \Psi,\end{align*}
about the heteroclinic $X_j(\xi)$ with initial condition $\Psi_{j}(0;r) = e_{j}$. One readily checks that
\begin{align} \Psi_{\f}(\xi;r) = \re^{-\frac{\mu_0(r)+2}{D_0(r)} \xi} \left(-p_{\f}'(\xi),p_{\f}(\xi;r),0\right), \qquad \Psi_{\bb}(\xi;r) = \re^{-\frac{\mu_0(r)+u_{\bb}(r)}{D_0(r)} \xi} \left(-p_{\bb}'(\xi),p_{\bb}(\xi;r),0\right), \label{adjointsol}\end{align}
cf.~\eqref{frontsol} and~\eqref{backsol}. We compute the Melnikov integrals~\eqref{melnint} in~\S\ref{meln:compfront} and~\S\ref{meln:compback}.

To apply the theory of Deng and prove Theorems~\ref{thm:single_twist} and~\ref{thm:double_twist}
we need to show that the singular heteroclinic loop satisfies certain twisting conditions. That is, we need to understand the orientation of the intersection of the two-dimensional unstable manifold $W_{\pmb \alpha_0}^{\uu}(K_{1,0})$ with the two-dimensional stable manifold $W_{\pmb \alpha_0}^{\s}(K_{2,0})$ along the heteroclinic front $X_{\f}(\xi)$ in the fast subsystem~\eqref{FAST}, and similarly for the intersection of $W_{\pmb \alpha_0}^{\uu}(K_{2,0})$ with $W_{\pmb \alpha_0}^{\s}(K_{1,0})$ along the heteroclinic back $X_{\bb}(\xi)$. Since $u$ can be regarded as a parameter in the fast subsystem~\eqref{FAST}, the problem reduces to the study of the Melnikov functions $\mathcal{Q}_{\f}(u,D,\mu)$ associated with the heteroclinic connections $X_{\f}(\xi)$ and $X_{\bb}(\xi)$ in the two-dimensional system
\begin{align} \label{FASTR}
\begin{split}
q_\xi &= p,\\
p_\xi &= D^{-1}\left((u+\mu)p - f(q,u;r)\right),
\end{split}
\end{align}
arising at $(u,D,\mu) = (u_{\f}(r),D_0(r),\mu_0(r))$, with $u_{\f}(r) = 2$, and at $(u,D,\mu) = (u_{\bb}(r),D_0(r),\mu_0(r))$, respectively. We proceed as before. Thus, we write the dynamical system~\eqref{FASTR} in the abstract form
\begin{align*} \partial_\xi Y = \mathcal F(Y;u,D,\mu,r),\end{align*}
with $\mathcal F \colon \R^2 \times \R \times \mathcal V \times (\frac{2}{3},\infty) \to \R^2$, where $\mathcal V  \subset \R^2$ is a small neighborhood of $(D_0(r),\mu_0(r))$, and find that the associated derivatives are explicitly given by
\begin{align} \partial_{i} \mathcal Q_{j}(u_{j}(r),D_0(r),\mu_0(r);r) = -\int_\R \widetilde{\Psi}_{j}(\xi;r) \cdot \partial_{i} \mathcal F(Y_{j}(\xi);u_{j}(r),D_0(r),\mu_0(r),r) \de \xi, \label{melnint2}\end{align}
for $i = u,D,\mu$ and $j = \f,\bb$ with
\begin{align} \widetilde{\Psi}_{\f}(\xi;r) = \re^{-\frac{\mu_0(r)+2}{D_0(r)} \xi} \left(-p_{\f}'(\xi),p_{\f}(\xi;r)\right), \qquad \widetilde{\Psi}_{\bb}(\xi;r) = \re^{-\frac{\mu_0(r)+u_{\bb}(r)}{D_0(r)} \xi} \left(-p_{\bb}'(\xi),p_{\bb}(\xi;r)\right). \label{adjointsol2}\end{align}
Clearly, it holds
\begin{align*} \partial_{i} \mathcal Q_{j}(u_{j}(r),D_0(r),\mu_0(r);r) = \partial_{i} Q_{j}(\pmb \alpha_0(r);r),\end{align*}
for $i = D,\mu$ and $j = \f,\bb$, cf.~\eqref{melnint}. Thus, it remains to determine $\partial_u \mathcal Q_{j}(u_{j}(r),D_0(r),\mu_0(r);r)$, whose computation can also be found in the upcoming sections~\S\ref{meln:compfront} and~\S\ref{meln:compback}.

\subsubsection{Computation of Melnikov integrals along the heteroclinic front} \label{meln:compfront}

Using~\eqref{melnint} and~\eqref{adjointsol} we compute
\begin{align*} \partial_\mu Q_{\f}(\pmb \alpha_0(r);r) &= -\int_\R \frac{\re^{-\frac{\mu_0(r)+2}{D_0(r)} \xi} p_{\f}(\xi;r)^2}{D_0(r)} \de \xi = \frac{q_{\f,+}(r)^2 (r+0.1)}{D_0(r)^2} \widehat{M}_{\f}(r),\end{align*}
with
\begin{align}
\widehat{M}_{\f}(r) = -\frac{q_{\f,+}(r)D_0(r)\sqrt{2}\left(q_{\f,+}(r)-2q_{\f,-}(r)\right)}{2(\mu_0(r) + 2)} \int_\R \re^{-\sqrt{2}\left(\frac{1}{2} - \frac{q_{\f,-}(r)}{q_{\f,+}(r)}\right)\chi} \phi'(\chi)^2 \de \chi < 0. \label{defhatMf}
\end{align}

Moreover, using~\eqref{cond:front},~\eqref{frontsol},~\eqref{melnint} and~\eqref{adjointsol} we obtain
\begin{align*} \partial_D Q_{\f}(\pmb \alpha_0(r);r) &= \int_\R \frac{p_{\f}(\xi;r)\left((\mu_0(r) + 2)p_{\f}(\xi;r) - q_{\f}(\xi;r)\left(r - (r+0.1)(q_{\f}(\xi;r) - 1)^2\right)\right)}{\re^{\frac{\mu_0(r)+2}{D_0(r)} \xi}D_0(r)^2} \de \xi\\
&= \frac{q_{\f,+}(r)^2 (r+0.1)}{D_0(r)^2} M_{\f}(r),
\end{align*}
with
\begin{align}
\begin{split}
& M_{\f}(r) = \frac{q_{\f,+}(r)}{2}\sqrt{2}\left(q_{\f,+}(r)-2q_{\f,-}(r)\right) \int_\R \re^{-\sqrt{2}\left(\frac{1}{2} - \frac{q_{\f,-}(r)}{q_{\f,+}(r)}\right)\chi} \phi'(\chi)^2 \de \chi\\
&\quad + \, \frac{0.1}{r+0.1} \int_\R \re^{-\sqrt{2}\left(\frac{1}{2} - \frac{q_{\f,-}(r)}{q_{\f,+}(r)}\right)\chi} \phi'(\chi)\phi(\chi) \de \chi - 2q_{\f,+}(r)\int_\R \re^{-\sqrt{2}\left(\frac{1}{2} - \frac{q_{\f,-}(r)}{q_{\f,+}(r)}\right)\chi} \phi'(\chi)\phi(\chi)^2 \de \chi \\
&\quad + \, q_{\f,+}(r)^2 \int_\R \re^{-\sqrt{2}\left(\frac{1}{2} - \frac{q_{\f,-}(r)}{q_{\f,+}(r)}\right)\chi} \phi'(\chi)\phi(\chi)^3 \de \chi.
\end{split}\label{defMfr}
\end{align}
Although the integrals in $M_\f(r)$ can be computed explicitly, we refrain from doing so as the obtained expressions are highly involved. Instead, we determine $M_\f(r)$ in the limit $r \to \infty$. Using~\eqref{ubvalues} and~\eqref{qlims} we find
\begin{align*} \lim_{r \to \infty} M_{\f}(r) = \frac{1}{3}.\end{align*}
Hence, for $r > \frac{2}{3}$ sufficiently large, $\partial_D Q_{\f}(\pmb \alpha_0(r);r)$ is positive. The plot in Figure~\ref{fig:11} suggests that $\partial_D Q_{\f}(\pmb \alpha_0(r);r)$ is in fact positive for all $r > \frac{2}{3}$.

As outlined in~\S\ref{sec:Meln}, the third Melnikov integral to compute along the heteroclinic front is $\partial_u \mathcal Q_{\f}(u_{\f}(r),D_0(r),\mu_0(r);r)$. Thus, using~\eqref{cond:front},~\eqref{frontsol},~\eqref{melnint2} and~\eqref{adjointsol2} we obtain
\begin{align*} \partial_u \mathcal Q_{\f}(u_{\f}(r),D_0(r),\mu_0(r);r) &= -\int_\R \frac{\re^{-\frac{\mu_0(r)+2}{D_0(r)} \xi} p_{\f}(\xi;r)\left(p_{\f}(\xi;r) - q_{\f}(\xi;r)\right)}{D_0(r)} \de \xi = \frac{q_{\f,+}(r)^2}{D_0(r)} \widetilde{M}_{\f}(r),
\end{align*}
with
\begin{align}
\begin{split}
\widetilde{M}_{\f}(r) &= -q_{\f,+}(r)\sqrt{\frac{r+0.1}{D_0(r)}} \int_\R \re^{-\sqrt{2}\left(\frac{1}{2} - \frac{q_{\f,-}(r)}{q_{\f,+}(r)}\right)\chi} \phi'(\chi)^2 \de \chi\\ 
&\quad + \,\int_\R \re^{-\sqrt{2}\left(\frac{1}{2} - \frac{q_{\f,-}(r)}{q_{\f,+}(r)}\right)\chi} \phi'(\chi)\phi(\chi) \de \chi.
\end{split} \label{defMtf}
\end{align}
Again, we refrain from computing $\smash{\widetilde{M}}_{\f}(r)$ explicitly, and determine its value in the limit $r \to \infty$. Using~\eqref{limitsspeed2},~\eqref{ubvalues} and~\eqref{qlims} we find
\begin{align*} \lim_{r \to \infty} \widetilde{M}_{\f}(r) = -\infty.\end{align*}
Hence, for $r > \frac{2}{3}$ sufficiently large, $\partial_u \mathcal Q_{\f}(u_{\f}(r),D_0(r),\mu_0(r);r)$ is negative. Numerical computations, see also the plot in Figure~\ref{fig:12}, suggest that $\partial_u \mathcal Q_{\f}(u_{\f}(r),D_0(r),\mu_0(r);r)$ is negative for all $r > 0.72946$, whereas it is positive for $r \in (\frac{2}{3}, 0.72946)$, see also Hypothesis~\ref{hyp}.

\begin{figure}
\centering
\begin{subfigure}[b]{0.45\textwidth}
\centering
\includegraphics[scale=0.6]{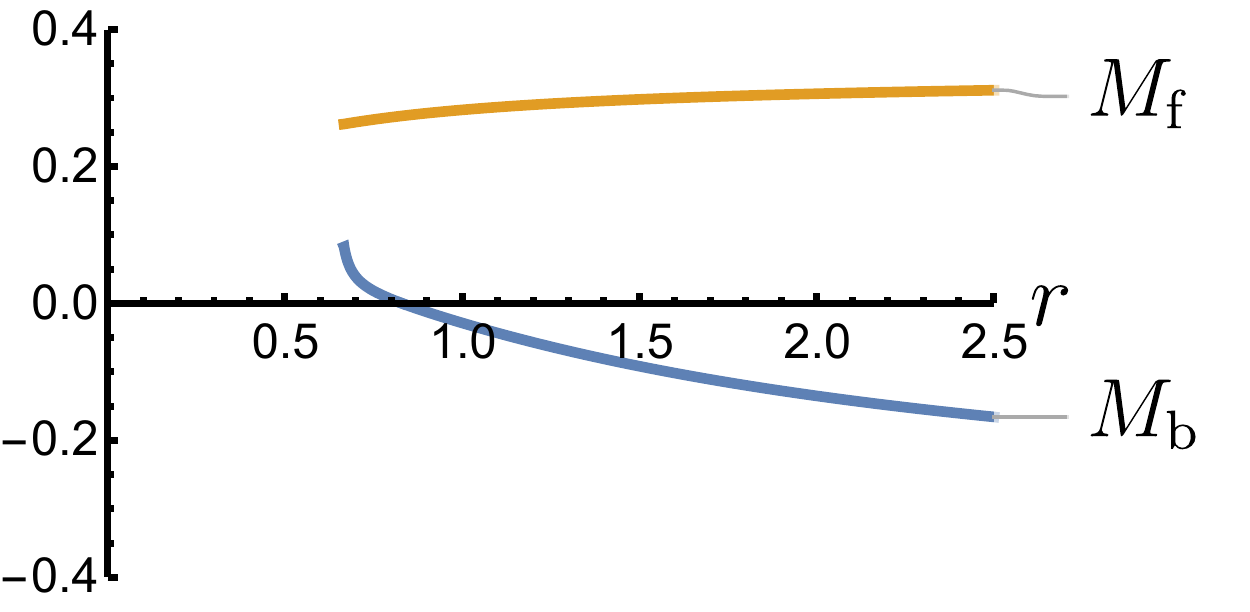}
\caption{Plots of $M_{\f}(r)$ and $M_{\bb}(r)$.}
\label{fig:11}
\end{subfigure}
\hfill
\begin{subfigure}[b]{0.45\textwidth}
\centering
\includegraphics[scale=0.6]{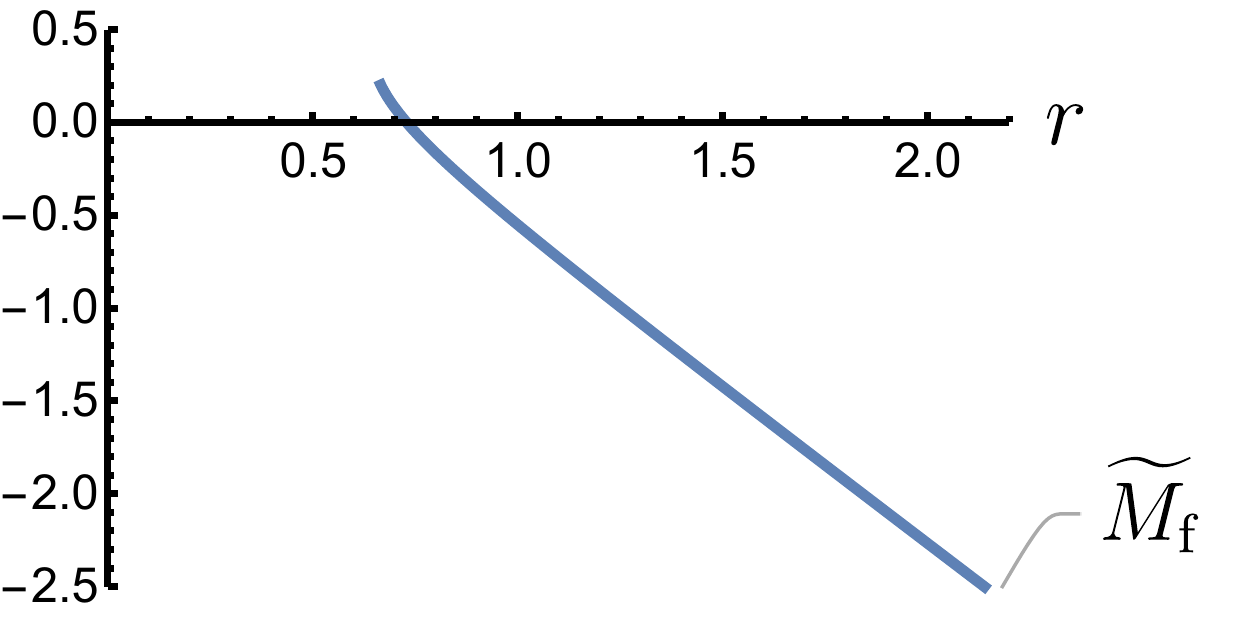}
\caption{Plot of $\smash{\widetilde{M}}_{\f}(r)$.}
\label{fig:12}
\end{subfigure}\\
\vspace{0.2cm}
\begin{subfigure}[b]{0.54\textwidth}
\centering
\includegraphics[scale=0.6]{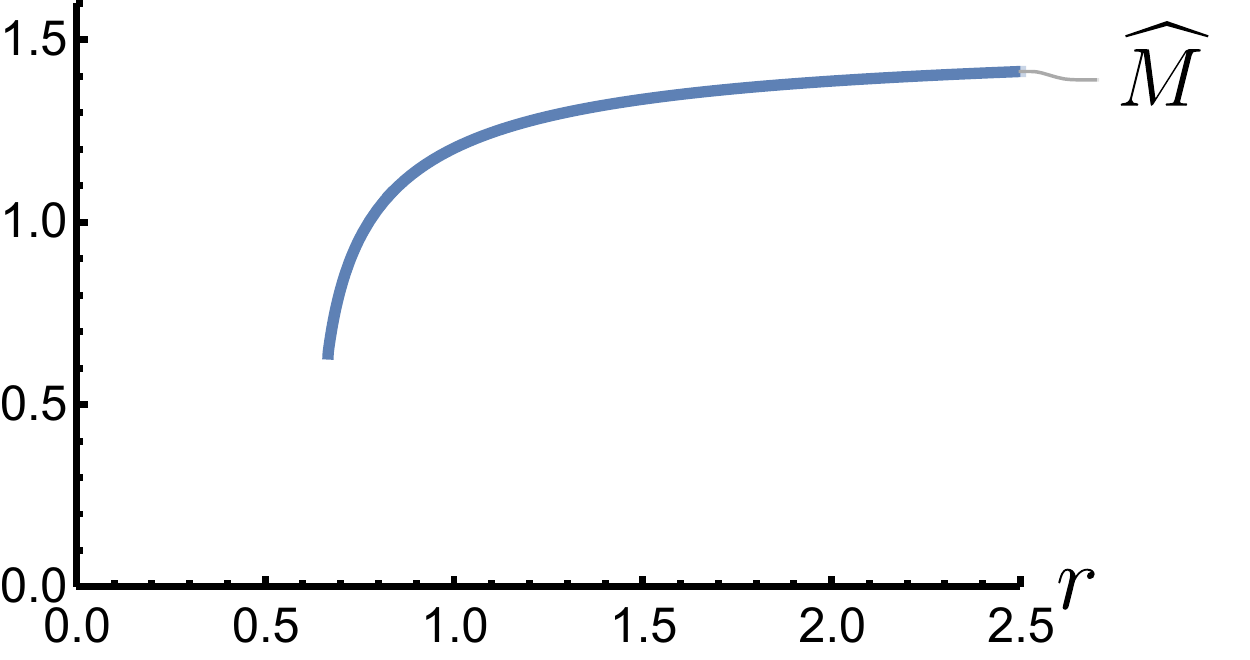}
\caption{Plot of $\smash{\widehat{M}}(r)$.}
\label{fig:122}
\end{subfigure}
\caption{Plots of Melnikov integrals as a function of $r$.}\label{fig:8}
\end{figure}

\subsubsection{Computation of Melnikov integrals along the heteroclinic back}  \label{meln:compback}

Using~\eqref{melnint} and~\eqref{adjointsol} we compute
\begin{align*} \partial_\mu Q_{\bb}(\pmb \alpha_0(r);r) = -\int_\R \frac{\re^{-\frac{\mu_0(r)+u_{\bb}(r)}{D_0(r)} \xi} p_{\bb}(\xi;r)^2}{D_0(r)} \de \xi
\frac{q_{\bb,+}(r)^2 (r+0.1)}{D_0(r)^2} \widehat{M}_{\bb}(r),\end{align*}
with
\begin{align}
\widehat{M}_{\bb}(r) = -\frac{q_{\bb,+}(r)\sqrt{2}\left(q_{\bb,+}(r)-2q_{\bb,-}(r)\right)}{2(\mu_0(r) + u_{\bb}(r))} \int_\R \re^{-\sqrt{2}\left(\frac{1}{2} - \frac{q_{\bb,-}(r)}{q_{\bb,+}(r)}\right)\chi}\phi'(\chi)^2 \de \chi < 0. \label{defhatMb}
\end{align}

Moreover, using~\eqref{cond:back},~\eqref{backsol},~\eqref{melnint} and~\eqref{adjointsol} we arrive at
\begin{align*} \partial_D Q_{\bb}(\pmb \alpha_0(r);r) &= \int_\R \frac{p_{\bb}(\xi;r)\left((\mu_0(r) +u_{\bb}(r))p_{\bb}(\xi;r) - q_{\bb}(\xi;r)\left(r - (r+0.1)(q_{\bb}(\xi;r) - 1)^2\right)\right)}{\re^{\frac{\mu_0(r)+u_{\bb}(r)}{D_0(r)} \xi} D_0(r)^2} \de \xi\\
&= \frac{q_{\bb,+}(r)^2 (r+0.1)}{D_0(r)^2} M_{\bb}(r),
\end{align*}
with
\begin{align}
\begin{split}
M_{\bb}(r) &= -\frac{q_{\bb,+}(r)}{2}\sqrt{2}\left(q_{\bb,+}(r)-2q_{\bb,-}(r)\right) \int_\R \re^{-\sqrt{2}\left(\frac{1}{2} - \frac{q_{\bb,-}(r)}{q_{\bb,+}(r)}\right)\chi}\phi'(\chi)^2 \de \chi\\
&\quad + \, \frac{u_{\bb}(r) - 2.1}{r+0.1} \int_\R \re^{-\sqrt{2}\left(\frac{1}{2} - \frac{q_{\bb,-}(r)}{q_{\bb,+}(r)}\right)\chi} \phi'(\chi)\phi(\chi) \de \chi\\ 
&\quad + \, 2q_{\bb,+}(r)\int_\R \re^{-\sqrt{2}\left(\frac{1}{2} - \frac{q_{\bb,-}(r)}{q_{\bb,+}(r)}\right)\chi} \phi'(\chi)\phi(\chi)^2 \de \chi \\
&\quad -  \, q_{\bb,+}(r)^2 \int_\R \re^{-\sqrt{2}\left(\frac{1}{2} - \frac{q_{\bb,-}(r)}{q_{\bb,+}(r)}\right)\chi} \phi'(\chi)\phi(\chi)^3 \de \chi.
\end{split} \label{defMbr}
\end{align}
Again, we proceed by computing $M_\bb(r)$ in the limit $r \to \infty$. Using~\eqref{ubvalues} and~\eqref{qlims} we find
\begin{align*} \lim_{r \to \infty} M_{\bb}(r) = -\frac{1}{3}.\end{align*}
Hence, for $r > \frac{2}{3}$ sufficiently large, $\partial_D Q_{\bb}(\pmb \alpha_0(r);r)$ is negative. 

For establishing a nondegenerate heteroclinic loop via the implicit function theorem in the upcoming~\S\ref{sec:hetloop}, one requires that the Jacobi matrix $\left(\partial_{D,\mu} Q_i(\pmb \alpha_0(r);r)\right)_{i = \f,\bb}$, associated with the algebraic system of equations $Q_\f(\pmb \alpha;r) = 0 = Q_\bb(\pmb \alpha;r)$, is invertible, which is the case precisely if the quantity
\begin{align} \widehat{M}(r) := \frac{\partial_D Q_{\bb}(\pmb \alpha_0(r);r)}{\partial_\mu Q_{\bb}(\pmb \alpha_0(r);r)} - \frac{\partial_D Q_{\f}(\pmb \alpha_0(r);r)}{\partial_\mu Q_{\f}(\pmb \alpha_0(r);r)} = \frac{M_{\bb}(r)}{\widehat{M}_{\bb}(r)} - \frac{M_{\f}(r)}{\widehat{M}_{\f}(r)}, \label{defhatM}\end{align}
is non-zero. For sufficiently large $r > \frac{2}{3}$, it follows rigorously by the previous observations that $\smash{\widehat{M}(r)}$ is positive, and hence non-zero. On the other hand, the plot in Figure~\ref{fig:122} suggests that  $\smash{\widehat{M}}(r)$ is in fact non-zero for all $r > \frac{2}{3}$, which indicates that the nondegeneracy hypothesis $\smash{\widehat{M}}(r) \neq 0$ is always satisfied.

The last Melnikov integral to compute along the heteroclinic back is $\partial_u \mathcal Q_{\bb}(u_{\bb}(r),D_0(r),\mu_0(r);r)$. Thus, using~\eqref{melnint2} and~\eqref{adjointsol2} we obtain
\begin{align*} \partial_u \mathcal Q_{\bb}(u_{\bb}(r),D_0(r),\mu_0(r);r) &= -\int_\R \frac{\re^{-\frac{\mu_0(r)+u_{\bb}(r)}{D_0(r)} \xi} p_{\bb}(\xi;r)\left(p_{\bb}(\xi;r) - q_{\bb}(\xi;r)\right)}{D_0(r)} \de \xi.
\end{align*}
It directly follows from~\eqref{backsol} that $q_{\bb}(\xi;r)$ is a monotonically decreasing back connecting $q = q_{\bb,+}(r) > 0$ with $q = 0$. Hence, the Melnikov integral $\partial_u \mathcal Q_{\bb}(u_{\bb}(r),D_0(r),\mu_0(r);r)$ is negative for all $r > \frac{2}{3}$.

\subsection{Establishing the heteroclinic loop} \label{sec:hetloop}
We can now use the Melnikov computations of the previous subsection to establish the existence of heteroclinic loops for $\varepsilon > 0$ sufficiently small and particular parameter combinations $(\smash{\widehat{D}}(\varepsilon,r), \smash{\widehat{\mu}}(\varepsilon,r))$. The heteroclinic loop is a critical dynamical object about which the dynamics are organized; in particular, its existence formalizes the occurrence of turbulent puffs, as transitions from the laminar to the turbulent state and back. The nondegeneracy of the heteroclinic loop guarantees the continuation of such heteroclinic connections along different parameter curves, yielding a large variety of traveling waves exhibiting turbulent patches, see Theorems~\ref{thm:single_twist} and~\ref{thm:double_twist} and their proofs in the upcoming~\S\ref{sec:proof_twists}.

\begin{lemma}[Existence of heteroclinic loop] \label{lem:nondegen_heteroclinic}
Let $r > \frac{2}{3}$ be such that $\smash{\widehat{M}}(r) \neq 0$. Then, there exists $\varepsilon_0(r) > 0$ such that for each $\varepsilon \in (0,\varepsilon_0(r))$ there is a parameter combination $\pmb \alpha(\varepsilon;r) = (\smash{\widehat{D}}(\varepsilon,r), \smash{\widehat{\mu}}(\varepsilon,r), \varepsilon)$ such that the ODE~\eqref{SYS2} has a heteroclinic loop, i.e.~it possesses both a heteroclinic orbit $X_{\f}^{\varepsilon}$ from $X_1$ to $X_2$ and a  heteroclinic orbit $X_{\bb}^{\varepsilon}$ from $X_2$ to $X_1$. The stable manifold $W^{\s}_{\pmb \alpha(\varepsilon;r)}(K_{2,\pmb \alpha(\varepsilon;r)})$ and the unstable manifold $W^{\uu}_{\pmb \alpha(\varepsilon;r)}(K_{1,\pmb \alpha(\varepsilon;r)})$ intersect transversally along $X_{\f}^\varepsilon$, and similarly $W^{\s}_{\pmb \alpha(\varepsilon;r)}(K_{1,\pmb \alpha(\varepsilon;r)})$ and $W^{\uu}_{\pmb \alpha(\varepsilon;r)}(K_{2,\pmb \alpha(\varepsilon;r)})$ intersect transversally along $X_{\bb}^\varepsilon$. Finally, the functions $\smash{\widehat{D}}(\varepsilon,r)$ and $\smash{\widehat{\mu}}(\varepsilon,r)$ are smoothly dependent on their parameters and satisfy~\eqref{limits}. 
\end{lemma}
\begin{proof}
Note that the persistence of the stable and unstable manifolds about the singular limit $\varepsilon = 0$ is discussed in~\S\ref{sec:persistence_manifolds}. In addition, we have already established there that, for ${\pmb \alpha}_0 = \pmb \alpha_0(r) = (D_0(r),\mu_0(r),0)$, the one-dimensional unstable manifold $W_{\pmb \alpha_0}^{\uu}(X_1)$ intersects the two-dimensional stable manifold $W_{\pmb \alpha_0}^{\s}(K_{2,0})$ along the heteroclinic front solution $X_{\f}(\xi) = \left(q_{\f}(\xi;r),p_{\f}(\xi;r),0\right)$ which implies that, due to the fibering~\eqref{fibering}, the two-dimensional unstable manifold $W_{\pmb \alpha_0}^{\uu}(K_{1,0})$ intersects $W_{\pmb \alpha_0}^{\s}(K_{2,0})$ along the heteroclinic front. Similarly, at ${\pmb \alpha} = {\pmb \alpha}_0$ the one-dimensional unstable manifold $W_{\pmb \alpha_0}^{\uu}(X_2)$, which is contained in the two-dimensional unstable manifold $W_{\pmb \alpha_0}^{\uu}(K_{2,0})$, intersects the two-dimensional stable manifold $W_{\pmb \alpha_0}^{\s}(K_{1,0})$ in system~\eqref{FAST} along the heteroclinic back solution $X_{\bb}(\xi) = \left(q_{\bb}(\xi;r),p_{\bb}(\xi;r),u_{\bb}(r)\right)$.

We can construct the continuation of the intersections of these manifolds for $\varepsilon > 0$ sufficiently small via Melnikov's method, using the computations in~\S\ref{sec:Meln}:
since $\partial_\mu Q_{\f}(\pmb \alpha_0(r);r) \neq 0$, we conclude with the implicit function theorem that the solutions of $Q_{\f}(\pmb \alpha;r)=0$ can be expressed as a function $\mu= \mu_{\f,0}(D; \varepsilon)$, near ${\pmb \alpha}_0=(D_0(r), \mu_0(r),0)$, such that $\mu_{\f,0}(D_0(r);0) = \mu_0(r)$. In particular, the continuation of the front heteroclinic, as an intersection of $W_{\pmb \alpha}^{\uu}(X_1)$ and $W_{\pmb \alpha}^{\s}(K_{2,0}) = W_{\pmb \alpha}^{\s}(X_2)$,  is established for the parameter values ${\pmb \alpha}=( D, \mu_{\f,0}(D; \varepsilon), \varepsilon)$ with $\varepsilon > 0$. The connection between $W_{\pmb \alpha}^{\uu}(X_2)$ and $W_{\pmb \alpha}^{\s}(K_{1,0}) = W_{\pmb \alpha}^{\s}(X_1)$ can be established analogously by recalling from~\S\ref{sec:Meln} that
$\partial_\mu Q_{\bb}(\pmb \alpha_0(r);r)\neq 0$, yielding a function $\mu= \mu_{\bb,0}(D; \varepsilon)$ satisfying $\mu_{\bb,0}(D_0(r);0) = \mu_0(r)$.

Additionally, by the fact that $\widehat{M}(r) \neq 0$, we know, by the implicit function theorem, that
$$\frac{\partial \mu_{f,0}}{\partial D}( D_0(r); 0) = \frac{\partial_D Q_{\f}(\pmb \alpha_0(r);r)}{\partial_\mu Q_{\f}(\pmb \alpha_0(r);r)} \neq \frac{\partial_D Q_{\bb}(\pmb \alpha_0(r);r)}{\partial_\mu Q_{\bb}(\pmb \alpha_0(r);r)} = \frac{\partial \mu_{\bb,0}}{\partial D}( D_0(r); 0).$$
Hence, provided $\varepsilon > 0$ is sufficiently small, the curves $\mu_{\f,0}(\cdot;\varepsilon)$ and $\mu_{\bb,0}(\cdot;\varepsilon)$ in the $(\mu,D)$-parameter plane intersect transversely near $(D_0(r), \mu_0(r))$ at some point $(\smash{\widehat{D}}(\varepsilon,r), \smash{\widehat{\mu}}(\varepsilon,r))$, where we have a heteroclinic loop of trajectories $X_{\f}^{\varepsilon}$ from $X_1$ to $X_2$ and $X_{\bb}^{\varepsilon}$ from $X_2$ to $X_1$. By the implicit function theorem the functions $\smash{\widehat{D}}(\varepsilon,r)$ and $\smash{\widehat{\mu}}(\varepsilon,r)$ depend smoothly on their variables.

Let $\pmb \alpha(\varepsilon;r) = (\smash{\widehat{D}}(\varepsilon,r), \smash{\widehat{\mu}}(\varepsilon,r),\varepsilon)$. Note that $\widehat{M}(r) \neq 0$ implies that the Jacobi matrix
\begin{align*}\begin{pmatrix} \partial_\mu Q_{\f}(\pmb \alpha(\varepsilon;r);r) & \partial_D Q_{\f}(\pmb \alpha(\varepsilon;r);r) \\ \partial_\mu Q_{\bb}(\pmb \alpha(\varepsilon;r) & \partial_D Q_{\bb}(\pmb \alpha(\varepsilon;r)\end{pmatrix},\end{align*}
is invertible for $\varepsilon > 0$ sufficiently small. Hence, the transverse crossing of the stable and unstable manifolds $W^{\uu}_{\pmb \alpha(\varepsilon;r)}(K_{1,\pmb\alpha(\varepsilon;r)})$ and $W^{\s}_{\pmb \alpha(\varepsilon;r)}(K_{2,\pmb\alpha(\varepsilon;r)})$ along the heteroclinic front $X_{\f}^\varepsilon$ follows. Similarly, $W^{\uu}_{\pmb \alpha(\varepsilon;r)}(K_{2,\pmb\alpha(\varepsilon;r)})$ and $W^{\s}_{\pmb \alpha(\varepsilon;r)}(K_{1,\pmb\alpha(\varepsilon;r)})$ intersect transversely along the heteroclinic back $X_{\bb}^\varepsilon$.
\end{proof}

We note that the proof of Theorem~\ref{heteroclinicloop} follows directly from Corollary~\ref{cor:hetloop} and the above Lemma~\ref{lem:nondegen_heteroclinic}.
 
\section{Proofs of Theorems~\ref{thm:single_twist} and~\ref{thm:double_twist}} \label{sec:proof_twists}

We wish to apply Deng's general results~\cite{Deng91a} on the bifurcations of a single twisted and double twisted heteroclinic loop in order to prove Theorems~\ref{thm:single_twist} and~\ref{thm:double_twist}, respectively. To do so, we must verify five conditions for the heteroclinic loop obtained in Lemma~\ref{lem:nondegen_heteroclinic}, see also~\cite[Theorem~2.1]{Deng91b}. In this section, we will establish those conditions one by one.

The first condition concerns the relative expansion of the equilibria $X_1$ and $X_2$ of the heteroclinic loop and is already given by our statement in Lemma~\ref{lem:relative_expansion}. The second condition is the transverse crossing of the stable and unstable manifolds $\smash{W^{\s/\uu}_{\pmb \alpha(\varepsilon;r)}(K_{i,\pmb\alpha(\varepsilon;r)}), i = 1,2}$ along the heteroclinic connections $X_{\f}^\varepsilon$ and $X_{\bb}^\varepsilon$, which was established in Lemma~\ref{lem:nondegen_heteroclinic}. The third condition concerns the nondegeneracy of the heteroclinic loop, which is the content of the following lemma.

\begin{lemma}[Nondegeneracy of heteroclinic loop] \label{lem:nondegen}
Assume Hypothesis~\ref{hyp} is satisfied. There exists $r_0 > \frac{2}{3}$ such that for $r \in (\frac{2}{3},\frac{2}{3}+\gamma) \cup (r_0,\infty)$ the heteroclinic loop established in Lemma~\ref{lem:nondegen_heteroclinic} is nondegenerate in the sense that:
\begin{enumerate}
\item the heteroclinic front $X_\f^{\varepsilon}(\xi)$ is asymptotically tangent to the principal stable eigenvector $e_2(X_2)$ of $X_2$  as $\xi \to \infty$, and the principal unstable eigenvector $e_3(X_1)$ of $X_1$ as $\xi \to -\infty$, respectively, and the same holds true for the heteroclinic back $X_\bb^{\varepsilon}(\xi)$ with $X_2$ and $X_1$ interchanged;
\item the strong inclination conditions
$$ \lim_{\xi \to - \infty} T_{X_\f^{\varepsilon}(\xi)} W_{\pmb \alpha(\varepsilon;r)}^{\s}(X_2) = T_{X_1} W_{\pmb \alpha(\varepsilon;r)}^{\uu}(X_1) + T_{X_1} W_{\pmb \alpha(\varepsilon;r)}^{\s\s}(X_1), $$
are met, where $T_{p} W$ denotes the tangent space of a manifold $W$ at the base point $p \in W$ and
$$W_{\pmb \alpha(\varepsilon;r)}^{\s}(X_i), \ W_{\pmb \alpha(\varepsilon;r)}^{\uu}(X_i), \ \text{and } W_{\pmb \alpha(\varepsilon;r)}^{\s\s}(X_i)$$
 are the (two-dimensional) stable, (one-dimensional) unstable and (one-dimensional) strong stable manifolds, respectively. Similarly, we have
$$ \lim_{\xi \to - \infty} T_{X_\bb^{\varepsilon}(\xi)} W_{\pmb \alpha(\varepsilon;r)}^\s(X_1) = T_{X_2} W_{\pmb \alpha(\varepsilon;r)}^\uu(X_2) + T_{X_2} W_{\pmb \alpha(\varepsilon;r)}^{\s\s}(X_2). $$
\end{enumerate}
\end{lemma}
\begin{proof}
(i) Convergence along the principal stable eigenvector $e(X_2)$ for the heteroclinic front can be derived by observing that $X_{\f}^{\varepsilon}$ is, by continuity, not contained in the strong stable manifold $W_{{\pmb \alpha(\varepsilon;r)}}^{\s\s}(X_2)$, because $W_{\pmb \alpha_0(r)}^{\s\s}(X_2)$ lies in the layer $u = u_{\bb}(r)$ in the fast subsystem~\eqref{FAST}, whereas the orbit of $X_{\f}^\varepsilon$ approaches the singular heteroclinic front as $\varepsilon \downarrow 0$, which consists of the heteroclinic $X_\f$ in the fast subsystem~\eqref{FAST} in the layer $u = 2$ and the slow orbit segment in the slow subsystem~\eqref{SLOW} connecting $(q_{f,+}(r),0,2)$ with $X_2$, cf.~Figure~\ref{fig:5}. Convergence along the principal unstable eigenvector of $X_1$ follows directly from the fact that the linearization about the equilibrium $X_1$ has only one unstable eigenvalue, cf.~Lemma~\ref{lem:relative_expansion}. The statement for the heteroclinic back is obtained analogously.

(ii) For showing the strong inclination property, we can use
\begin{align*}
\partial_{u} \mathcal Q_{j}(u_{j}(r),D_0(r),\mu_0(r);r) \neq 0,
\end{align*}
for $j = \f,\bb$, which follows from the analysis in~\S\ref{sec:Meln} in combination with Hypothesis~\ref{hyp} for $r \in (\frac{2}{3},\frac{2}{3}+\gamma) \cup (r_0,\infty)$ for some $r_0 > 0$ sufficiently large. Specifically, this implies that $W_{\pmb \alpha_0}^{\s}(K_{j,0})$ and $W_{\pmb \alpha_0}^{\uu}(K_{i,0})$ intersect transversely and the strong inclination property is then satisfied by the strong $\lambda$-lemma~\cite{Deng90} for the manifolds at $\varepsilon =0$. Due to robustness of the strong inclination property, we can deduce the statement for sufficiently small $\varepsilon >0$.
\end{proof}

The fourth condition, as stated in the lemma below, concerns the continuation of the corresponding heteroclinics associated with fronts and backs, respectively.

\begin{lemma}[Continuation of $X_i^{\varepsilon}$] \label{lem:transverse} Consider the heteroclinic loop established in Lemma~\ref{lem:nondegen_heteroclinic}.
There are two curves $0$-$\text{het}_{12}^\varepsilon$ and $0$-$\text{het}_{21}^\varepsilon$ in the $(D, \mu)$-parameter plane (see also Figure~\ref{fig:twist}), which intersect transversely at $(\smash{\widehat{D}}(\varepsilon,r), \smash{\widehat{\mu}}(\varepsilon,r))$ such that for $\pmb \alpha = (D,\mu,\varepsilon)$, with $\varepsilon > 0$ sufficiently small and $(D,\mu) \in$ $0$-$\text{het}_{12}^\varepsilon$, there is a simple heteroclinic orbit $X_{\f,\pmb \alpha}$ from $X_1$ to $X_2$ being the continuation of $X_\f$ in the sense that $X_{\f,\pmb \alpha(\varepsilon;r)}=X_\f^{\varepsilon}$ and $X_{\f, \pmb \alpha}(\xi)$ is continuous in $\xi$ and $\pmb \alpha = (D,\mu,\varepsilon)$ for $(D,\mu)$ along the curve $0$-$\text{het}_{12}^\varepsilon$. The analogous holds for the curve $0$-$\text{het}_{21}^\varepsilon$ corresponding to heteroclinic orbits $X_{\bb,\pmb \alpha}(\xi)$ being the continuation of the heteroclinic back $X_{\bb}^\varepsilon$.
\end{lemma}
\begin{proof}
The follows directly from the existence of the curves $\mu = \mu_{\f,0}(D;\varepsilon)$ and $\mu = \mu_{\bb,0}(D;\varepsilon)$, corresponding to simple heteroclinic front and back connections between $X_1$ and $X_2$, as established in the proof of Lemma~\ref{lem:nondegen_heteroclinic}.
\end{proof}

The final condition concerns the twist properties of the heteroclinic loop and yields the distinction between the statements in Theorem~\ref{thm:single_twist} (single twisted) and Theorem~\ref{thm:double_twist} (double twisted). We will establish this condition in the upcoming two subsections.

\subsection{Single twisted regime}

The twist properties of the established nondegenerate heteroclinic loop are crucial for obtaining the adequate bifurcation diagram, see~\cite[Hypothesis 5.16]{HomburgSandstede}. In fact, for large $r$, the geometry of our model yields a single twist. To describe such a twist along the back we use that the principal stable and unstable eigenvectors $e_2(X_2)$ and $e_3(X_1)$ are transverse to the principal eigenvectors $e_3(X_2)$ and $e_2(X_1)$, which are by Lemma~\ref{lem:nondegen} tangent to the back as $\xi \to \infty$ and $\xi \to -\infty$, respectively, see Figure~\ref{fig:manifolds}a. Similarly, we use the principal eigenvectors $e_2(X_1)$ and $e_3(X_1)$ to describe twist properties along the front, see Figure~\ref{fig:manifolds}b.

\begin{lemma}[Single twist of the heteroclinic loop]
\label{lem:twist}
Consider the heteroclinic loop established in Lemma~\ref{lem:nondegen_heteroclinic}. For sufficiently large $r > \frac{2}{3}$,
the principal eigenvectors $e_2(X_2)$ and $e_3(X_1)$ point to opposite sides of the tangent space $T_{X_\bb^\varepsilon(\xi)} W_{{\pmb \alpha(\varepsilon;r)}}^{\s}(X_1)$ as $\xi \to -\infty$ and $\xi \to +\infty$, respectively, which means that the heteroclinic back $X_\bb^{\varepsilon}$ is twisted. Furthermore, $e_2(X_1)$ and $e_3(X_2)$ point to the same sides of $T_{X_\f^{\varepsilon}(\xi)} W_{{\pmb \alpha(\varepsilon;r)}}^{\s}(X_2)$ as $\xi \to -\infty$ and $\xi \to +\infty$, respectively, which means that the heteroclinic front $X_\f^{\varepsilon}$ is not twisted.
\end{lemma}
\begin{proof}
We argue at the hand of Figure~\ref{fig:manifolds}.
Recall from~\S\ref{sec:persistence_manifolds} that the stable manifold $W^{\s}_{\pmb \alpha}(K_{1,\pmb \alpha})$ depends smoothly on $\pmb \alpha$ for $\pmb \alpha$ close to $\pmb \alpha_0$. Moreover, for $\varepsilon > 0$ we have $W^{\s}_{\pmb \alpha}(K_{1,\pmb \alpha}) = W^{\s}_{\pmb \alpha}(X_1)$. So, by continuous dependency on $\varepsilon$, it suffices to show that $e_3(X_1)$ and $e_2(X_2)$ point to opposite sides of $W_{\pmb \alpha_0}^{\s}(K_{1,0})$. We can deduce this fact from $\partial_u \mathcal Q_{\bb}(\pmb \alpha_0(r);r) <0$, see~\S\ref{sec:Meln},  which determines the relative positions of $W_{\pmb \alpha_0}^{\s}(K_{1,0})$ and $W_{\pmb \alpha_0}^{\uu}(K_{2,0})$, see Figure~\ref{fig:manifolds}a. Hence, we establish the twist of the heteroclinic $X_\bb^{\varepsilon}$.

On the other hand, one can observe analogously that $\partial_u \mathcal Q_{\f}(\pmb \alpha_0(r);r) <0$ implies that $e_3(X_2)$ and $e_2(X_1)$ point to the same side of $W_{\pmb \alpha_0}^{\s}(K_{2,0})$, see Figure~\ref{fig:manifolds}b.
Hence, we obtain that $X_\f^{\varepsilon}$ is not twisted.
\end{proof}

Having shown the above Lemmas~\ref{lem:relative_expansion},~\ref{lem:nondegen_heteroclinic},~\ref{lem:nondegen},~\ref{lem:transverse} and~\ref{lem:twist}, Theorem~\ref{thm:single_twist} now follows from~\cite[Theorem 5.27]{HomburgSandstede}, taking $s = - \mu - \zeta$, see also~\cite{Deng91b}.

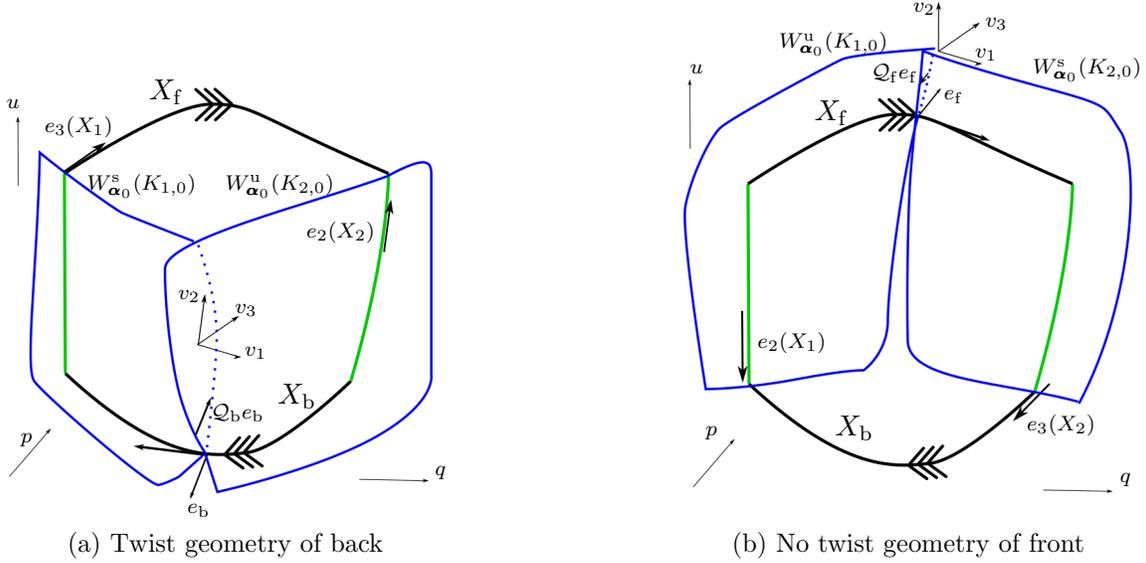
\begin{figure}[ht!]
\centering
\begin{subfigure}[b]{0.45\textwidth}
\centering
\begin{overpic}[scale=0.4]{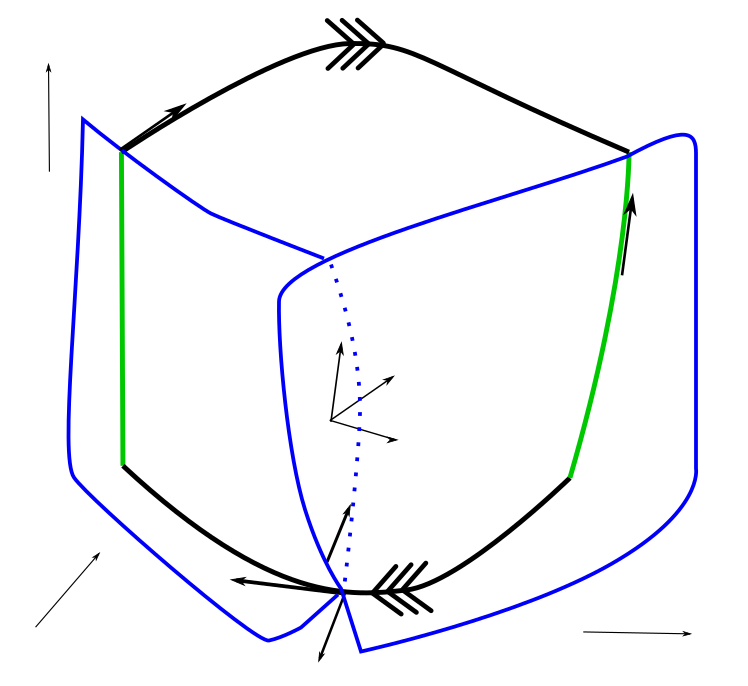}
\put(92, 8) { \scriptsize $q$}
\put(5, 15) { \scriptsize $p$}
\put(2, 85) { \scriptsize $u$}
\put(10, 80) { \scriptsize $e_3(X_1)$}
\put(65, 58) { \scriptsize $e_2(X_2)$}
\put(32, 87) { $X_{\f}$}
\put(59, 23) { $X_{\bb}$}
\put(52, 33) { \scriptsize $v_1$}
\put(38, 45) { \scriptsize $v_2$}
\put(50, 42) { \scriptsize $v_3$}
\put(21, 68) {\scriptsize $W_{\pmb \alpha_0}^{\s}(K_{1,0})$}
\put(50, 68) {\scriptsize $W_{\pmb \alpha_0}^{\uu}(K_{2,0})$}
\put(40, 0) { \scriptsize $e_{\bb}$}
\put(45, 20) { \scriptsize $\mathcal Q_{\bb} e_{\bb}$}
\end{overpic}
\caption{Twist geometry of back}
\end{subfigure}
\hfill
\begin{subfigure}[b]{0.45\textwidth}
\centering
\begin{overpic}[scale=0.4]{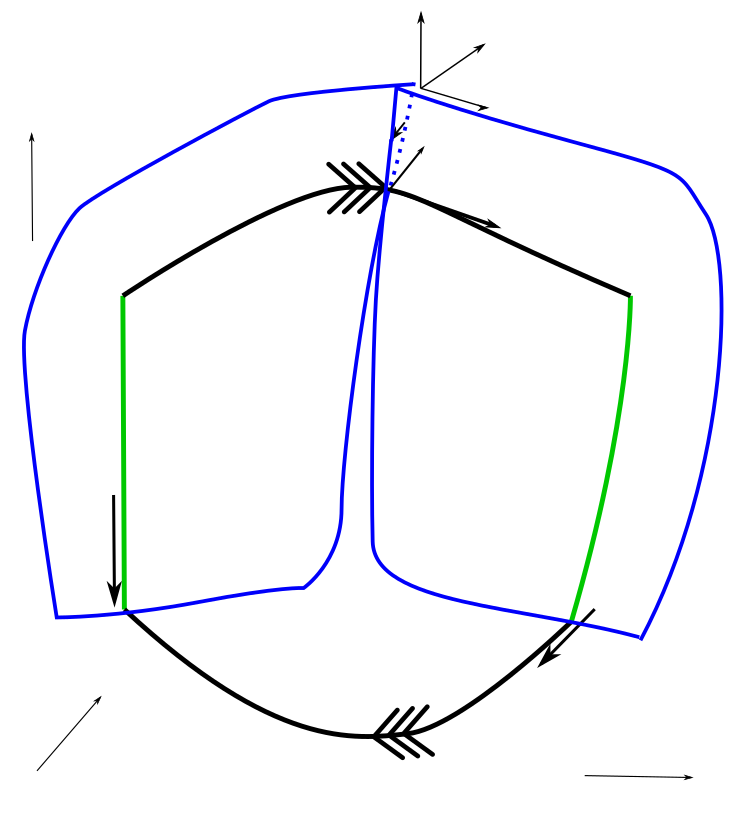}
\put(85, 6) { \scriptsize $q$}
\put(5, 15) { \scriptsize $p$}
\put(2, 85) { \scriptsize $u$}
\put(15, 32) { \scriptsize $e_2(X_1)$}
\put(67, 16) { \scriptsize $e_3(X_2)$}
\put(26, 76) { $X_{\f}$}
\put(30, 15) { $X_{\bb}$}
\put(57, 88) { \scriptsize $v_1$}
\put(45, 97) { \scriptsize $v_2$}
\put(59, 94) { \scriptsize $v_3$}
\put(21, 90) {\scriptsize $W_{\pmb \alpha_0}^{\uu}(K_{1,0})$}
\put(70, 85) {\scriptsize $W_{\pmb \alpha_0}^{\s}(K_{2,0})$}
\put(51, 80) { \scriptsize $e_{\f}$}
\put(37, 84) { \scriptsize $\mathcal Q_{\f} e_{\f}$}
\end{overpic}
\caption{No twist geometry of front}
\end{subfigure}
\caption{Sketches of the intersections of the stable and unstable manifolds and the directions of the principal stable and unstable eigenvectors in the singular limit $\varepsilon \downarrow 0$.
In the left panel, we depict the separation of the manifolds $W_{\pmb \alpha_0}^{\s}(K_{1,0})$
and $W_{\pmb \alpha_0}^{\uu}(K_{2,0})$, as measured by $Q_{\bb}(\pmb \alpha_0(r);r) e_b$. We have $\partial_u \mathcal Q_{\bb}(\pmb \alpha_0(r);r) < 0$ so that $W_{\pmb \alpha_0}^{\s}(K_{1,0})$ points towards the viewer and $e_2(X_2)$ and $e_3(X_1)$ point to opposite sides  of $W_{\pmb \alpha_0}^{\s}(K_{1,0})$, yielding the twist property. In the right panel, we depict the separation of the manifolds  $W_{\pmb \alpha_0}^{\s}(K_{2,0})$
and $W_{\pmb \alpha_0}^{\uu}(K_{1,0})$, as measured by $\mathcal Q_{\f}(\pmb \alpha_0(r);r)e_f$. We depict the situation for $\partial_u \mathcal Q_{\f}(\pmb \alpha_0(r);r) < 0$ so that $W_{\pmb \alpha_0}^{\s}(K_{1,0})$ points towards the viewer and $e_3(X_2)$ and $e_2(X_1)$ point to the same side of $W_{\pmb \alpha_0}^{\s}(K_{2,0})$, yielding no twist property as in Lemma~\ref{lem:twist}. Note that for $\partial_u \mathcal Q_{\f}(\pmb \alpha_0(r);r) > 0$, as in Lemma~\ref{lem:twistw}, there is a twist about the front as well.}
\label{fig:manifolds}
\end{figure}

\subsection{Double twisted regime}

We now turn to the intermediate Reynolds number regime, where we establish that the heteroclinic loop is double twisted.

\begin{lemma}[Double twist of the heteroclinic loop]
\label{lem:twistw}
Consider the heteroclinic loop established in Lemma~\ref{lem:nondegen_heteroclinic}. Assume Hypothesis~\ref{hyp} is satisfied. For $r \in (\frac{2}{3},\frac{2}{3}+\gamma)$, the principal eigenvectors $e_2(X_2)$ and $e_3(X_1)$ point to opposite sides of the tangent space $T_{X_\bb^\varepsilon(\xi)} W_{{\pmb \alpha(\varepsilon;r)}}^{\s}(X_1)$ as $\xi \to -\infty$ and $\xi \to +\infty$, respectively, which means that the heteroclinic back $X_\bb^{\varepsilon}$ is twisted. Furthermore, $e_2(X_1)$ and $e_3(X_2)$ also point to opposite sides of $T_{X_\f^{\varepsilon}(\xi)} W_{{\pmb \alpha(\varepsilon;r)}}^{\s}(X_2)$ as $\xi \to -\infty$ and $\xi \to +\infty$, respectively, which means that the heteroclinic front $X_\f^{\varepsilon}$ is also twisted.
\end{lemma}
\begin{proof}
Consider again Figure~\ref{fig:manifolds}. Observe that the twist of $X_\bb^{\varepsilon}$ is established as in the proof of Lemma~\ref{lem:twist}, since we have again $\partial_u \mathcal Q_{\bb} (\pmb \alpha_0(r);r) u <0$, cf.~\S\ref{sec:Meln}. 

On the other hand, due to $\partial_u \mathcal Q_{\f} (\pmb \alpha_0(r);r) >0$, cf.~Hypothesis~\ref{hyp}, the vectors $e_2(X_1)$ and $e_3(X_2)$ point to opposite sides of $W_{\pmb \alpha_0}^{\s}(K_{2,0})$. Hence, we obtain that, for sufficiently small $\varepsilon >0$, the heteroclinic front $X_\f^{\varepsilon}$ is also twisted. In particular, we are in the situation of~\cite[Hypothesis 5.16 (iii)]{HomburgSandstede}.
\end{proof}

Having shown the above Lemmas~\ref{lem:relative_expansion},~\ref{lem:nondegen_heteroclinic},~\ref{lem:nondegen},~\ref{lem:transverse} and~\ref{lem:twistw}, Theorem~\ref{thm:double_twist} now follows from~\cite[Theorem 5.27]{HomburgSandstede}, see also~\cite{Deng91b}.

\section{Discussion and outlook} \label{sec:outlook}

Our main result identifies the precise mechanism leading to the rise of turbulent spatio-temporal structures in the intermediate Reynolds number regime for the model~\eqref{SYS1}-\eqref{SYS11}. The model was previously validated experimentally and numerically via simulations of the full Navier-Stokes equations, precisely in this intermediate Reynolds number regime~\cite{Barkleyetal}. The main organizing center of the chaotic structures turns out to be a heteroclinic loop in the associated singularly perturbed traveling-wave equation. We anticipate that, although we have identified a large variety of spatio-temporal structures mediating the transition to turbulence, there could be many further interesting bifurcations of~\eqref{SYS1}-\eqref{SYS11}. To investigate these, it seems natural to exploit relations of~\eqref{SYS1}-\eqref{SYS11} with the FitzHugh-Nagumo (FHN) PDE, which is also bistable but has no advective term. For FHN, there are several works, which study bifurcation structures varying multiple parameters over broad ranges~\cite{Sneydetal,Deng91b,GuckenheimerKuehn1,Hastings4}. In the context of FHN, the question of stability of complex spatio-temporal waves has also raised quite a bit of attention~\cite{CarterdeRijkSandstede,Nii1,Sandstede2}. However, it seems doubtful whether long-term asymptotic stability of waves is the correct concept to unravel the transition to turbulence, where transient chaotic structures might play a much more prominent role~\cite{GrebogiOttYorke,LaiTel}. Therefore, we believe that it is of primary importance to understand even more about the organization of invariant solutions in phase and parameter space. 

Furthermore, we have also studied~\eqref{SYS1}-\eqref{SYS11} in the large Reynolds number regime. The question regarding the practical applicability of our results in this regime has already been discussed in Remark~\ref{rem:largeRE}. In fact, our results indicate that some complex spatio-temporal patterns mediating the transition to turbulence in~\eqref{SYS1}-\eqref{SYS11} disappear as the Reynolds number is increased further into the fully turbulent regime. Whether one observes in experiments only dynamics about the one turbulent state, i.e.~the equilibrium $X_2$, and corresponding simple heteroclinic and homoclinic connections towards it, still has to be studied for large $r$. Answering this question will decide if the model~\eqref{SYS1}-\eqref{SYS11} can give physically meaningful answers to turbulence also for large Reynolds number.\medskip

\textbf{Acknowledgments:} CK would like to thank the VolkswagenStiftung for support via a Lichtenberg Professorship. 
ME would like to thank the DFG for support within Germany's Excellence Strategy -- The Berlin Mathematics Research Center MATH+ (EXC-2046/1, project ID: 390685689).


\begin{thebibliography}{10}

\bibitem{Avilaetal}
K.~Avila, D.~Moxey, A.~de~Lozar, M.~Avila, D.~Barkley, and B.~Hof.
\newblock The onset of turbulence in pipe flow.
\newblock {\em Science}, 333(6039):192--196, 2011.

\bibitem{AvilaWillisHof}
M.~Avila, A.P. Willis, and B.~Hof.
\newblock On the transient nature of localized pipe flow turbulence.
\newblock {\em J. Fluid Mech.}, 646:127--136, 2010.

\bibitem{BardosTiti}
C.W. Bardos and E.S. Titi.
\newblock Mathematics and turbulence: where do we stand?
\newblock {\em J. Turbulence}, 14(3):42--76, 2013.

\bibitem{Barkleyetal}
D.~Barkley, B.~Song, V.~Mukund, G.~Lemoult, M.~Avila, and B.~Hof.
\newblock The rise of fully turbulent flow.
\newblock {\em Nature}, 526(7574):550--553, 2015.

\bibitem{BecKhanin}
J.~Bec and K.~Khanin.
\newblock Burgers turbulence.
\newblock {\em Phys. Rep.}, 447(1):1--66, 2007.

\bibitem{Bertozzi1}
A.L. Bertozzi.
\newblock Heteroclinic orbits and chaotic dynamics in planar fluid flows.
\newblock {\em SIAM J. Math. Anal.}, 19(6):1271--1294, 1988.

\bibitem{BuckmasterVicol}
T.~Buckmaster and V.~Vicol.
\newblock Nonuniqueness of weak solutions to the {Navier-Stokes} equation.
\newblock {\em Ann. Math.}, 189(1):101--144, 2019.

\bibitem{BudanurHof1}
N.B. Budanur and B.~Hof.
\newblock Heteroclinic path to spatially localized chaos in pipe flow.
\newblock {\em J. Fluid Mech.}, 827:R1, 2017.

\bibitem{BudanurHof}
N.B. Budanur and B.~Hof.
\newblock Complexity of the laminar-turbulent boundary in pipe flow.
\newblock {\em Phys. Rev. Fluids}, 3(5):054401, 2018.

\bibitem{CarterdeRijkSandstede}
P.~Carter, B.~de~Rijk, and B.~Sandstede.
\newblock Stability of traveling pulses with oscillatory tails in the
  {FitzHugh-Nagumo} system.
\newblock {\em J. Nonlinear Science}, 26(5):1369--1444, 2016.

\bibitem{Sneydetal}
A.R. Champneys, V.~Kirk, E.~Knobloch, B.E. Oldeman, and J.~Sneyd.
\newblock {When Shil'nikov meets Hopf in excitable systems}.
\newblock {\em SIAM J. Appl. Dyn. Syst.}, 6(4):663--693, 2007.

\bibitem{ChateManneville}
H.~Chat{\'e} and P.~Manneville.
\newblock Spatiotemporal intermittency in coupled map lattices.
\newblock {\em Phys. D}, 32:409--422, 1988.

\bibitem{DarbyshireMullin}
A.G. Darbyshire and T.~Mullin.
\newblock Transition to turbulence in constant-mass-flux pipe flow.
\newblock {\em J. Fluid Mech.}, 289:83--114, 1995.

\bibitem{DoeringGibbon}
C.R. Doering and J.D. Gibbon.
\newblock {\em Applied Analysis of the Navier-Stokes Equations}.
\newblock CUP, 1995.

\bibitem{Deng90}
B.~Deng.
\newblock Homoclinic bifurcations with nonhyperbolic equilibria.
\newblock {\em SIAM J. Math. Anal.}, 21:693--720, 1990.

\bibitem{Deng91a}
B.~Deng.
\newblock The bifurcations of countable connections from a twisted heteroclinic loop.
\newblock {\em SIAM J. Math. Anal.}, 22(3):653--679, 1991.

\bibitem{Deng91b}
B.~Deng.
\newblock The existence of infinitely many traveling front and back waves in the FitzHugh-Nagumo equations.
\newblock {\em SIAM J. Math. Anal.}, 22(6):1631--1650, 1991.

\bibitem{DuguetWillisKerswell}
Y.~Duguet, A.P. Willis, and R.R. Kerswell.
\newblock Transition in pipe flow: the saddle structure on the boundary of
  turbulence.
\newblock {\em J. Fluid Mech.}, 613:255--274, 2008.

\bibitem{Eckhardt}
B.~Eckhardt.
\newblock Turbulence transition in pipe flow: some open questions.
\newblock {\em Nonlinearity}, 21(1):T1, 2007.

\bibitem{EckhardtSchneiderHofWesterweel}
B.~Eckhardt, T.M. Schneider, B.~Hof, and J.~Westerweel.
\newblock Turbulence transition in pipe flow.
\newblock {\em Annu. Rev. Fluid Mech.}, 39:447--468, 2007.

\bibitem{FaistEckhardt}
H.~Faisst and B.~Eckhardt.
\newblock Traveling waves in pipe flow.
\newblock {\em Phys. Rev. Lett.}, 91:224502, 2003.

\bibitem{FEN2}
N.~Fenichel.
\newblock Geometric singular perturbation theory for ordinary differential
  equations.
\newblock {\em J. Differential Equations}, 31(1):53--98, 1979.

\bibitem{GrebogiOttYorke}
C.~Grebogi, E.~Ott, and J.A. Yorke.
\newblock Crises, sudden changes in chaotic attractors, and transient chaos.
\newblock {\em Physica D}, 7(1):181--200, 1983.

\bibitem{GH}
J.~Guckenheimer and P.~Holmes.
\newblock {\em Nonlinear Oscillations, Dynamical Systems, and Bifurcations of
  Vector Fields}.
\newblock Springer, New York, NY, 1983.

\bibitem{GuckenheimerKuehn1}
J.~Guckenheimer and C.~Kuehn.
\newblock {Homoclinic orbits of the FitzHugh-Nagumo equation: The singular
  limit}.
\newblock {\em DCDS-S}, 2(4):851--872, 2009.

\bibitem{Hastings4}
S.P. Hastings.
\newblock Single and multiple pulse waves for the {FitzHugh-Nagumo}.
\newblock {\em SIAM J. Appl. Math.}, 42(2):247--260, 1982.

\bibitem{Hinrichsen}
H.~Hinrichsen.
\newblock Non-equilibrium critical phenomena and phase transitions into
  absorbing states.
\newblock {\em Adv. Phys.}, 49(7):815--958, 2000.

\bibitem{Hofetal}
B.~Hof, A.~de~Lozar, D.J. Kuik, and J.~Westerweel.
\newblock Repeller or attractor? {Selecting} the dynamical model for the onset
  of turbulence in pipe flow.
\newblock {\em Phys. Rev. Lett.}, 101:214501, 2008.

\bibitem{HofJuelMullin}
B.~Hof, A.~Juel, and T.~Mullin.
\newblock Scaling of the turbulence transition threshold in a pipe.
\newblock {\em Phys. Rev. Lett.}, 91(24):244502, 2003.

\bibitem{HofWesterweelSchneiderEckhardt}
B.~Hof, J.~Westerweel, T.M. Schneider, and B.~Eckhardt.
\newblock Finite lifetime of turbulence in shear flows.
\newblock {\em Nature}, 443(7107):59--62, 2006.

\bibitem{Holmes},
P.~J.~Holmes.
\newblock Averaging and chaotic motions in forced oscillations.
\newblock {\em SIAM J. Appl. Math.}, 38(1):65--80, 1980.

\bibitem{HomburgSandstede}
A.~J. Homburg and B.~Sandstede.
\newblock Homoclinic and heteroclinic bifurcations in vector fields.
\newblock In H.~W. Broer, B.~Hasselblatt, and F.~Takens, editors, {\em Handbook
  of Dynamical Systems}, chapter~8, pages 379--524. Elsevier/North-Holland,
  Amsterdam, 2010.

\bibitem{IoossMielke}
G.~Iooss and A.~Mielke.
\newblock Bifurcating time-periodic solutions of {Navier-Stokes} equations in
  infinite cylinders.
\newblock {\em J. Nonl. Sci.}, 1(1):107--146, 1991.

\bibitem{Jones}
C.K.R.T. Jones.
\newblock Geometric singular perturbation theory.
\newblock In {\em Dynamical Systems (Montecatini Terme, 1994)}, volume 1609 of
  {\em Lect. Notes Math.}, pages 44--118. Springer, 1995.
  
\bibitem{Kaneko2}
K.~Kaneko.
\newblock Spatiotemporal intermittency in coupled map lattices.
\newblock {\em Prog. Theor. Phys.}, 74(5):1033--1044, 1985.

\bibitem{Kaper}
T.J. Kaper.
\newblock An introduction to geometric methods and dynamical systems theory for
  singular perturbation problems. analyzing multiscale phenomena using singular
  perturbation methods.
\newblock In J.~Cronin and R.E. O'Malley, editors, {\em Analyzing Multiscale
  Phenomena Using Singular Perturbation Methods}, pages 85--131. Springer,
  1999.
  
\bibitem{KawaharaUhlmannVanVeen}
S.~Kawashima and Y.~Shizuta.
\newblock The significance of simple invariant solutions in turbulent flows.
\newblock {\em Ann. Rev. Fluid Mech.}, 44:203--225, 2012.

\bibitem{Kerswell}
R.R. Kerswell.
\newblock Recent progress in understanding the transition to turbulence in a
  pipe.
\newblock {\em Nonlinearity}, 18(6):R17, 2005.

\bibitem{Kida}
S.~Kida.
\newblock Asymptotic properties of {Burgers} turbulence.
\newblock {\em J. Fluid Mech.}, 93(2):337--377, 1979.

\bibitem{KuehnBook}
C.~Kuehn.
\newblock {\em Multiple Time Scale Dynamics}.
\newblock Springer, 2015.

\bibitem{KuehnBook1}
C.~Kuehn.
\newblock {\em PDE Dynamics: An Introduction}.
\newblock SIAM, 2019.

\bibitem{KuikPoelmaWesterweel}
D.J. Kuik, C.~Poelma, and J.~Westerweel.
\newblock Quantitative measurement of the lifetime of localized turbulence in
  pipe flow.
\newblock {\em J. Fluid Mech.}, 645:529--539, 2010.

\bibitem{LaiTel}
Y.-C. Lai and T.~T{\'e}l, editors.
\newblock {\em Transient Chaos: Complex Dynamics on Finite Time Scales}.
\newblock Springer, 2011.

\bibitem{Lemoultetal}
G.~Lemoult, L.~Shi, K.~Avila, S.V. Jalikop, M.~Avila, and B.~Hof.
\newblock Directed percolation phase transition to sustained turbulence in
  {Couette} flow.
\newblock {\em Nat. Phys.}, 12(3):254--258, 2016.

\bibitem{Mellibovskyetal}
F.~Mellibovsky, A.~Meseguer, T.M. Schneider, and B.~Eckhardt.
\newblock Transition in localized pipe flow turbulence.
\newblock {\em Phys. Rev. Lett.}, 103(5):054502, 2009.

\bibitem{Melnikov}
V.~K.~Melnikov.
\newblock On the stability of a center for time-periodic perturbations.
\newblock {\em (Russian) Trudy Moskov. Mat. Ob.} 12:3--52, 1963.

\bibitem{MesguerTrefethen}
A.~Meseguer and L.N. Trefethen.
\newblock Linearized pipe flow to reynolds number {$10^7$}.
\newblock {\em J. Comput. Phys.}, 186(1):178--197, 2003.

\bibitem{MisRoz}
E.F. Mishchenko and N.Kh. Rozov.
\newblock {\em Differential Equations with Small Parameters and Relaxation
  Oscillations (translated from Russian)}.
\newblock Plenum Press, 1980.

\bibitem{MoxeyBarkley}
D.~Moxey and D.~Barkley.
\newblock Distinct large-scale turbulent-laminar states in transitional pipe
  flow.
\newblock {\em Proc. Nat. Acad. Sci. USA}, 107(18):8091--8096, 2010.

\bibitem{MukundHof}
V.~Mukund and B.~Hof.
\newblock The critical point of the transition to turbulence in pipe flow.
\newblock {\em J. Fluid Mech.}, 839:76--94, 2018.

\bibitem{Nii1}
S.~Nii.
\newblock Stability of travelling multiple-front (multiple-back) wave solutions
  of the {FitzHugh-Nagumo} equations.
\newblock {\em SIAM J. Math. Anal.}, 28(5):1094--1112, 1997.

\bibitem{Ottino}
J.M. Ottino and D.V. Khakhar.
\newblock Mixing, chaotic advection, and turbulence.
\newblock {\em Ann. Rev. Fluid Mech.}, 22(1):207--254, 1990.

\bibitem{PAL}
K.~J. Palmer.
\newblock Exponential dichotomies and transversal homoclinic points.
\newblock {\em J. Differential Equations}, 55(2):225--256, 1984.

\bibitem{PeixinhoMullin}
J.~Peixinho and T.~Mullin.
\newblock Decay of turbulence in pipe flow.
\newblock {\em Phys. Rev. Lett.}, 96:094501, 2006.

\bibitem{Pomeau}
Y.~Pomeau.
\newblock Front motion, metastability and subcritical bifurcations in
  hydrodynamics.
\newblock {\em Phys. D}, 23:3--11, 1986.

\bibitem{Pomeau1}
Y.~Pomeau.
\newblock The long and winding road.
\newblock {\em Nat. Phys.}, 12(3):198--199, 2016.

\bibitem{Prandtl}
L.~Prandtl.
\newblock {\"{U}ber Fl\"{u}ssigkeiten bei sehr kleiner Reibung}.
\newblock In {\em Verh. III - International Math. Kongress}, pages 484--491.
  Teubner, 1905.

\bibitem{PringleKerswell}
C.C. Pringle and R.R. Kerswell.
\newblock Asymmetric, helical, and mirror-symmetric traveling waves in pipe
  flow.
\newblock {\em Phys. Rev. Lett.}, 99(7):074502, 2007.

\bibitem{Reynolds}
O.~Reynolds.
\newblock An experimental investigation of the circumstances which determine
  whether the motion of water shall be direct or sinuous, and of the law of
  resistance in parallel channels.
\newblock {\em Phil. Proc. R. Soc. London}, 174:935--982, 1883.

\bibitem{Reynolds1}
O.~Reynolds.
\newblock On the dynamical theory of incompressible viscous fluids and the
  determination of the criterion.
\newblock {\em Phil. Proc. R. Soc. London}, 186:123--164, 1895.

\bibitem{RitterMellibovskyAvila}
P.~Ritter, F.~Mellibovsky, and M.~Avila.
\newblock Emergence of spatio-temporal dynamics from exact coherent solutions
  in pipe flow.
\newblock {\em New J. Phys.}, 18(8):083031, 2016.

\bibitem{SalwenCottonGrosch}
H.~Salwen, F.W. Cotton, and C.E. Grosch.
\newblock Linear stability of poiseuille flow in a circular pipe.
\newblock {\em J. Fluid Mech.}, 98(2):273--284, 1980.

\bibitem{SamantaDeLozarHof}
D.~Samanta, A.~De Lozar, and B.~Hof.
\newblock Experimental investigation of laminar turbulent intermittency in pipe
  flow.
\newblock {\em J. Fluid Mech.}, 681:193--204, 2011.

\bibitem{Sandstede1}
B.~Sandstede.
\newblock Stability of travelling waves.
\newblock In B.~Fiedler, editor, {\em Handbook of Dynamical Systems}, volume~2,
  pages 983--1055. Elsevier, 2001.

\bibitem{Sandstede2}
B.~Sandstede.
\newblock Stability of {$N$}-fronts bifurcating from a twisted heteroclinic
  loop and an application to the {FitzHugh-Nagumo} equation.
\newblock {\em SIAM J. Math. Anal.}, 29(1):183--207, 1998.

\bibitem{SanoTamai}
M.~Sano and K.~Tamai.
\newblock A universal transition to turbulence in channel flow.
\newblock {\em Nat. Phys.}, 12(3):249--253, 2016.

\bibitem{Schneider6}
G.~Schneider.
\newblock Global existence results for pattern forming processes in infinite
  cylindrical domains - {Applications} to {3D Navier-Stokes} problems.
\newblock {\em J. Math. Pures Appl.}, 78(3):265--312, 1999.

\bibitem{SchneiderEckhardtYorke}
T.M. Schneider, B.~Eckhardt, and J.A. Yorke.
\newblock Turbulence transition and the edge of chaos in pipe flow.
\newblock {\em Phys. Rev. Lett.}, 99(3):034502, 2007.

\bibitem{Shanetal}
H.~Shan, B.~Ma, Z.~Zhang, and F.T.M. Nieuwstadt.
\newblock Direct numerical simulation of a puff and a slug in transitional
  cylindrical pipe flow.
\newblock {\em J. Fluid Mech.}, 387:39--60, 1999.

\bibitem{ShihHsiehGoldenfeld}
H.Y. Shih, T.L. Hsieh, and N.~Goldenfeld.
\newblock Ecological collapse and the emergence of travelling waves at the
  onset of shear turbulence.
\newblock {\em Nat. Phys.}, 12(3):245--248, 2016.

\bibitem{ShimizuManneville}
M.~Shimizu and P.~Manneville.
\newblock Bifurcations to turbulence in transitional channel flow.
\newblock {\em Phys. Rev. Fluids}, 4(11):113903, 2019.

\bibitem{SiposGoldenfeld}
M.~Sipos and N.~Goldenfeld.
\newblock Directed percolation describes lifetime and growth of turbulent puffs
  and slugs.
\newblock {\em Phys. Rev. E}, 84(3):035304, 2011.

\bibitem{SreenivasanRamshankhar}
K.R. Sreenivasan and R.~Ramshankhar.
\newblock Transition intermittency in open flows, and intermittency routes to
  chaos.
\newblock {\em Physica D}, 23(1):246--258, 1986.

\bibitem{Szmolyan1}
P.~Szmolyan.
\newblock Transversal heteroclinic and homoclinic orbits in singular
  perturbation problems.
\newblock {\em J. Differential Equat.}, 92:252--281, 1991.

\bibitem{TohItano}
S.~Toh and T.~Itano.
\newblock A periodic-like solution in channel flow.
\newblock {\em J. Fluid Mech.}, 481:67--76, 2003.

\bibitem{Wiggins1}
S.~Wiggins.
\newblock {\em Global Bifurcations and Chaos}.
\newblock Springer, 1998.

\bibitem{WygnanskiChampagne}
I.~Wygnanski and F.~Champagne.
\newblock On transition in a pipe. {Part 1. The} origin of puffs and slugs and
  the flow in a turbulent slug.
\newblock {\em J. Fluid Mech.}, 59(2):281--335, 1973.

\bibitem{WygnanskiSokolovFriedman}
I.~Wygnanski, M.~Sokolov, and D.~Friedman.
\newblock On transition in a {pipe. Part 2. The} equilibrium puff.
\newblock {\em J. Fluid Mech.}, 69(2):283--304, 1975.

\end{thebibliography}
\end{document}